\documentclass{LMCS}

\usepackage{graphicx,amssymb,stmaryrd,proof,pst-tree,cite,wasysym,lineno}
\usepackage[all,2cell,dvips,ps]{xy}
\CompileMatrices 
\xyoption{v2}
\xyoption{curve}
\xyoption{2cell}
\SelectTips{cm}{}  % Tips (of arrows) are in accordance with Computer Modern
\UseAllTwocells
\SilentMatrices

\newcommand{\Sets}{\mathbf{Sets}}

\newcommand{\co}{\mathrel{\circ}}
\newcommand{\congrightarrow}{\mathrel{\stackrel{
           \raisebox{.5ex}{$\scriptstyle\cong\,$}}{
           \raisebox{0ex}[0ex][0ex]{$\rightarrow$}}}}
\newcommand{\iso}{\congrightarrow}
\newcommand{\place}{\underline{\phantom{n}}\,} % place holder
\newcommand{\pow}{\mathcal{P}}
\newcommand{\dist}{\mathcal{D}}
\newcommand{\lift}{\mathcal{L}}

\newcommand{\Cppo}{\mathbf{Cppo}}
\newdir{ >}{{}*!/-8pt/@{>}}  % for mono @{ >->}
\newcommand{\mono}{\rightarrowtail}
\newcommand{\epi}{\twoheadrightarrow}
\newcommand{\op}{\mathop{\mathrm{op}}\nolimits}
\newcommand{\weg}[1]{}

\newcommand{\dst}{\mathsf{dst}}
\newcommand{\shifted}[3]{\save[]!<#1,#2>*{#3}\restore}
\newcommand{\Kleisli}[1]{\mathcal{K}{\kern-.2ex}\ell(#1)}
\newcommand{\trace}[1]{\mathsf{tr}_{#1}}
\newcommand{\inftytrace}[1]{\mathsf{tr}^{\infty}_{#1}}
\newcommand{\toTerm}{\,\mathop{\text{\rm !}}\,}
\newcommand{\fromInit}{\,\mathop{\text{\rm \textexclamdown}}\,}
\newcommand{\C}{\mathbb{C}}
\newcommand{\D}{\mathbb{D}}
\newcommand{\A}{\mathbb{A}}
\newcommand{\V}{\mathbb{V}}
\newcommand{\Coalg}[1]{\mathbf{Coalg}(#1)}
\newcommand{\Alg}[1]{\mathbf{Alg}(#1)}
\newcommand{\sem}[1]{\llbracket #1 \rrbracket} 
\newcommand{\theory}{\mathsf{th}}
\newcommand{\id}{\mathrm{id}}
\newcommand{\TestEq}{\mathsf{TestEq}}
\newcommand{\FCSEq}{\mathsf{FCSEq}}
\newcommand{\beh}{\mathsf{beh}}
\newcommand{\tuple}[1]{\langle #1 \rangle}
\newcommand{\oF}{\overline{F}}
\newcommand{\Op}{\mathit{Op}}
\newcommand{\Rel}{\mathbf{Rel}}
\newcommand{\Relf}{{\mathop{\mathrm{Rel}}\nolimits}_{F}}

\newcommand{\Fix}{\mathop\mathrm{Fix}}
\newcommand{\KT}{\Kleisli{T}}

\newdir{pb}{:(1,-1)@^{|-}}  % for pullback
  \def\pb#1{\save[]+<20 pt,0 pt>:a(#1)\ar@{pb{}}[]\restore}
\newif\ifignore % when set to true, additional text appears containing
\ignorefalse
\newcommand{\auxproof}[1]{
\ifignore\mbox{}\newline
\textbf{PROOF:} \dotfill\newline
{\it #1}\mbox{}\newline
\textbf{ENDPROOF}\dotfill
\fi}

\theoremstyle{plain}

\usepackage{eclbkbox}	% required for `\breakbox' (yatex added)
\usepackage{enumerate,hyperref}

\def\doi{3 (4:11) 2007}
\lmcsheading%
{\doi}
{1--36}
{}
{}
{Jul.~21, 2007}
{Nov.~19, 2007}
{}   

\begin{document}

\title[Generic Trace Semantics via Coinduction]{Generic Trace
Semantics via Coinduction\rsuper*}

\author[I.~Hasuo]{Ichiro Hasuo\rsuper a}	%required
\address{{\lsuper{a}}Institute for Computing and Information Sciences, Radboud University Nijmegen, the Netherlands 
and Research Institute for Mathematical Sciences, Kyoto University, Japan
}	%required
\urladdr{http://www.cs.ru.nl/\~{}ichiro}  %optional
\thanks{{\lsuper a}Supported by PRESTO research promotion program, 
  Japan Science and Technology Agency.}	%optional

\author[B.~Jacobs]{Bart Jacobs\rsuper b}	%optional
\address{{\lsuper b}Institute for Computing and Information Sciences, Radboud University Nijmegen, the Netherlands}
\urladdr{http://www.cs.ru.nl/\~{}bart}  %optional
\thanks{{\lsuper b}Also part-time at
   Technical University Eindhoven, the Netherlands.}	%optional

\author[A.~Sokolova]{Ana Sokolova\rsuper c}	%optional
\address{{\lsuper c}Department of Computer Sciences, University of
Salzburg, Austria}	%optional
\email{anas@cs.uni-salzburg.at}  %optional
\thanks{{\lsuper c}Supported by the Austrian Science Fund (FWF) project P18913-N15.}	%optional

%% etc.

%% required for running head on odd and even pages, use suitable
%% abbreviations in case of long titles and many authors:

%2% mandatory lists of keywords and classifications:
\keywords{coalgebra, category theory, trace semantics, monad, Kleisli category, process
semantics, non-determinism, probability}
\subjclass{F.3.1, F.3.2, G.3}
\titlecomment{{\lsuper*}Earlier versions~\cite{HasuoJ05b,HasuoJS06a} of this paper
have been presented at
the 1st International Conference on Algebra and Coalgebra in Computer Science
                  (CALCO 2005), Swansea, UK, September 2005, and at
the 8th International Workshop on Coalgebraic Methods in Computer
 Science  (CMCS 2006),
Vienna, Austria, March 2006.}
%%%%%%%%%%%%%%%%%%%%%%%%%%%%%%%%%%%%%%%%%%%%%%%%%%%%%%%%%%%%%%%%%%%%%%%%%%%

%% the abstract has to PRECEED the command \maketitle:
%% be sure not to issue the \maketitle command twice!

\begin{abstract}
  \noindent 
  Trace semantics has been defined for various kinds of state-based
systems, notably with different forms of branching such as
non-determinism vs.\ probability. In this paper we claim to identify one
underlying mathematical structure behind these ``trace semantics,''
namely coinduction in a Kleisli category. This claim is based on our
technical result that, under a suitably order-enriched setting, a final
coalgebra in a Kleisli category is given by an initial algebra in the
category $\Sets$. Formerly the theory of coalgebras has been employed
mostly in $\Sets$ where coinduction yields a finer process semantics of
bisimilarity.  Therefore this paper extends the application field of
coalgebras, providing a new instance of the principle ``process
semantics via coinduction.''

%  \noindent The abstract has to preceed the maketitle command.  Be
%  sure not to issue the maketitle command twice!  Preferably, the
%  abstact should consist of plain ASCII text, without mathematical
%  expressions or \TeX-commands, including explicit references using
%  the cite command.  Presently we are not able to automatically
%  extract an abstract containing such data and relyably turn it into
%  html code.  If you cannot meet these criteria, it is your
%  responsibility to provide us with an html-version of your abstract.
%  Please keep the abstract fairy short to prevent it from spilling
%  over to the second page!
\end{abstract}

\maketitle

%% start the paper here:
\section{Introduction}\label{section:introduction}
%hogehoge?

%Points to note, include:
%\begin{enumerate}
% \item Philosophy: overview of notions \& techniques in the recent line
% of work: trace theory in Kleisli categories. State of art, open
% questions and research directions.
% \item Attract non-categorical audience. (Cf. CONCUR paper)

%\end{enumerate}

%\dotfill

Trace semantics is a commonly used semantic relation for reasoning about
 state-based systems. Trace semantics for labeled transition
systems is found on the coarsest edge of the linear time-branching time
spectrum~\cite{vanGlabbeek01}. Moreover, trace semantics is defined for a variety of
systems, among which are probabilistic systems~\cite{Segala95}.  

In this paper we claim that  these various forms of ``trace semantics'' are
instances of a general construction, namely coinduction in a Kleisli
category. 
%Moreover we identify a final coalgebra in a Kleisli
%category: it is the initial algebra in $\Sets$. In particular,
%\begin{center}
% the initial algebra in $\Sets$ is the final coalgebra in $\Rel$
%\end{center}
%because the category $\Rel$ of sets and relations is the Kleisli
%category for a suitable monad.
Our point of view here is categorical, coalgebraic in
particular. Hence this paper demonstrates the abstraction power of
categorical/coalgebraic methods in computer science, uncovering basic
mathematical structures underlying  various concrete examples.

\subsection{``Trace semantics'' in various contexts}
\label{subsection:introVariousTraceSemantics}
First we motivate our contribution through examples of various forms of
``trace semantics.''
%We start with some examples: these will  be shown to be instances of
%our generic construction. 
Think of the following three
state-based, branching systems.
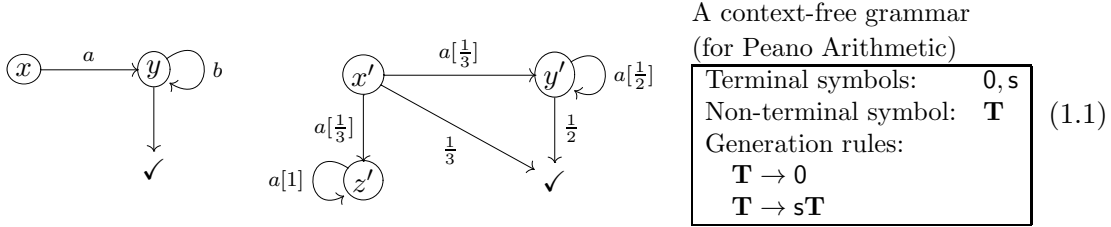
\begin{equation}\label{diagram:firstExampleOfSystems}
    \vcenter{\xymatrix@R+0em@C+1em{
    *+[o][F-]{x}
             \ar[r]^{a}
%  	     \ar@(dl,dr)[]_{a}
   &
    *+[o][F-]{y}
             \ar[d]
             \ar@(ru,rd)[]^{b}
   \\
   &
    {\checkmark}
   }}
 \quad
    \vcenter{\xymatrix@R+0em@C+3em{
    *+[o][F-]{x'}
             \ar[r]^{a[\frac{1}{3}]}
             \ar[d]_{a[\frac{1}{3}]}
             \ar[rd]_{\frac{1}{3}}
   &
    *+[o][F-]{y'}
             \ar[d]^{\frac{1}{2}}
             \ar@(ru,rd)[]^{a[\frac{1}{2}]}
   \\
    *+[o][F-]{z'}
             \ar@(lu,ld)[]_{a[1]}
   &
    {\checkmark}
   }}
 \quad
\begin{minipage}{25ex}
{\small
 A context-free grammar\\ (for Peano Arithmetic)\\
 \begin{tabular}{|ll|} 
 \hline
  Terminal symbols: &  $\mathsf{0, s}$
 \\
  Non-terminal symbol:  & $\mathbf{T}$
 \\
  Generation rules: &
 \\
 \tabcolsep = 1em
  \begin{tabular}{l}
   ${\mathbf{T}\to\mathsf{0}}$
   \\
   ${\mathbf{T}\to\mathsf{s}\mathbf{T}}$
  \end{tabular}
 &
 \\
 \hline
 \end{tabular}
}
\end{minipage}
\end{equation}
\begin{enumerate}[$\bullet$]
 \item 
 The first one is a non-deterministic system with a special state $\checkmark$ denoting
 successful termination. To its state $x$ we can assign its \emph{trace
 set}:
 \begin{equation}\label{equation:introExampleTraceSet}
 \trace{}(x) = \{a, ab, abb, \dotsc\} = ab^{*}\enspace,
 \end{equation}
 that is, the \emph{set} of the possible linear-time behavior (namely
 words) that can arise through an execution of the system.\footnote{The
       infinite trace $ab^{\omega}$ is out of our scope here: we will elaborate this point later in Section~\ref{subsection:infiniteTraces}.} In this case
 the trace set $\trace{}(x)$ is also called the \emph{accepted
       language};
 formally it is defined (co)recursively by the following equations.
 For an arbitrary state $x$,
  \begin{equation}\label{equation:conventionalDefOfTraceSetForLTS}
\begin{array}{rcll}
    \tuple{} \in \trace{}(x) 
  & \;\Longleftrightarrow\;
  & x\to\checkmark 
%  & \text{where $\tuple{}$ is the empty word;}
  \\
   a\cdot \sigma \in \trace{}(x)
  & \;\Longleftrightarrow\;
  & \exists y.\; (\, x\stackrel{a}{\to} y \;\land\; \sigma\in \trace{}(y)\,)  
%  \\
% \multicolumn{4}{r}{\text{where $\sigma=a_{1}a_{2}\dotsc a_{n}$ is a word in actions.}}
\end{array}  
\end{equation}
 Here $\tuple{}$ denotes the empty word; $\sigma=a_{1}a_{2}\dotsc a_{n}$ is a word.
 \item 
 The second system has a different type of branching, namely
       probabilistic branching. Here  $x'\stackrel{a[1/3]}{\longrightarrow} y'$
       denotes: at the state $x'$, a transition  to $y'$ outputting $a$
       occurs with probability $1/3$.
 Now, to the state $x'$, we can assign 
       its \emph{trace distribution}:
        \begin{equation}\label{equation:introExampleTraceDistribution}
 \trace{}(x) = 
      \left[	 \begin{array}{l}
	   \tuple{}\mapsto\frac{1}{3},
	    \quad
	    a
	    \mapsto
	    \frac{1}{3}\cdot \frac{1}{2},
	    \quad
	    a^{2}
	    \mapsto
	    \frac{1}{3}\cdot \frac{1}{2}\cdot \frac{1}{2},
	    \quad
	    \cdots
	  \\
            a^{n} 
	     \mapsto
	     \frac{1}{3}\cdot \left(\frac{1}{2}\right)^{n},\quad\cdots
	 \end{array}
\right]\enspace,
 \end{equation}
       that is, the probability distribution over the set of linear-time
       behavior.\footnote{Here again, we do not consider the infinite
       trace $a^{\omega}\mapsto 1/3$.} Its formal (corecursive) definition is as
       follows.
\begin{equation}\label{equation:conventionalDefOfTraceDistributionForProbLTS}
 \begin{array}{rcl}
\trace{}(x)(\langle\rangle)
& = &
{\mathsf{Pr}(x\to\checkmark)}\enspace,
\\
%& \quad &
\trace{}(x)(a\cdot\sigma)
& = &
{\sum_{y\in X}\mathsf{Pr}(x\stackrel{a}{\to}y)\cdot \trace{}(y)(\sigma)\enspace,}
\end{array}
\end{equation}
 where $\mathsf{Pr}(\dotsc)$ denotes the probability of a transition.
 \item The third example can be thought of as a
       state-based system, with non-terminal symbols as states. 
       It is non-deterministic because a state $\mathbf{T}$ has two
       possible transitions. It is natural to call the following set of
       parse-trees its ``trace semantics.''
       \begin{displaymath}
        \trace{}(\mathbf{T}) = \left\{\quad
        \begin{minipage}[c]{.35\textwidth}
	 \psset{levelsep=1em,treenodesize=.7em,linewidth=.3pt}
 	  \pstree{\TR{$\bullet$}}{\pstree{\TR{$\mathsf{0}$}}{}}
 \quad
 \pstree{\TR{$\bullet$}}{
 \TR{$\mathsf{s}$}
 \pstree{\TR{$\bullet$}}{\TR{$\mathsf{0}$}}
 }
 \quad
 \pstree{\TR{$\bullet$}}{
 \TR{$\mathsf{s}$}
 \pstree{\TR{$\bullet$}}{
 \TR{$\mathsf{s}$}
 \pstree{\TR{$\bullet$}}{\TR{$\mathsf{0}$}}
 }
 }
\end{minipage}
 \cdots\right\}
       \end{displaymath}
  It is again a set of ``linear-time behavior'' as in the first
       example, although the notion of linear-time behavior is different
       here. Linear-time behavior---that is, what we observe after we have resolved
       all the non-deterministic branchings in the system---is now 
       a parse-tree instead of a word.
\end{enumerate}

%#END
\subsection{Coalgebras and coinduction}
In recent years the theory of \emph{coalgebras} has emerged as
 the ``mathematics of state-based
 systems''~\cite{JacobsR97,Rutten00a,Jacobs05:coalgebraBook}. 
%The reasons for the success of coalgebras include: \emph{bisimilarity via
%coinduction} and \emph{bisimilarity via coalgebaic modal logics}.
%\subsubsection{Bisimilarity via coinduction.}
In the categorical
theory of coalgebras, an important definition/reasoning principle is \emph{coinduction}: a system
(identified with a coalgebra $c: X\to FX$) is assigned a unique morphism
$\beh_{c}$ into
the final coalgebra.
   \begin{align*}
 \vcenter{
 \xymatrix@C+5em@R=1em{
  {F X}
                \ar@{-->}[r]^{F (\beh_{c})} 
 &
  {F Z}
 \\
  {X}
                \ar[u]^{c}
                \ar@{-->}[r]_{\beh_{c}}
 &
  {Z}
                \ar[u]_{\text{final}}^{\cong}
}
}
% \qquad
% \begin{array}{ll}
%  \text{in}& \text{$\KT$,}
% \\
%  &
%  \text{the Kleisli category for $T$}
% \end{array}
 \end{align*}
The success of coalgebras is largely due to the fact that, when
$\Sets$ is taken as the base category, the final coalgebra semantics is
fully abstract with respect to the conventional notion of
\emph{bisimilarity}: for states $x$ and $y$ of coalgebras
$X\stackrel{c}{\to} FX$ and $Y\stackrel{d}{\to} FY$,
\begin{displaymath}
 \beh_{c}(x) = \beh_{d}(y) \qquad
 %\text{if and only if} 
 \Longleftrightarrow
 \qquad
 \text{$x$ and $y$ are bisimilar.}
\end{displaymath}
This is the case for a wide variety of systems (i.e.\ for a variety of
functors $F$), hence \emph{coinduction in $\Sets$ captures
bisimilarity}.

However, there is not so much work so far that captures other behavioral
equivalences (coarser than bisimilarity) by the categorical
principle of coinduction. 
 The current work---capturing
 trace semantics by coinduction in a Kleisli category---therefore extends the
 application field of the theory of coalgebras.

%\subsubsection{Coalgebraic modal logics.}
%It is standard that the Hennessy-Milner logic~\cite{HennessyM85} is
%``expressive,'' in the sense that two systems are bisimilar \emph{if and only if}
%they satisfy exactly the same set of formulas of the logic. Genericity
%of the coalgebraic framework---in which various systems are modeled as
%coalgebras for different functors $F$---inspired the exploration for a 
%generic framework for modal logic as well. In particular, a substantial
%body of work including~[citations?] introduce procedures which yield
%an expressive modal logic for a given functor $F$.

%One reason for the success of  coalgebras is the fact that, when
%$\Sets$ is taken as the base category, the final coalgebra semantics is
%fully abstract with respect to the conventional notion of
%\emph{bisimilarity}: for states $x$ and $y$ of coalgebras
%$X\stackrel{c}{\to} FX$ and $Y\stackrel{d}{\to} FY$,
%\begin{displaymath}
% \beh_{c}(x) = \beh_{d}(y) \qquad
% %\text{if and only if} 
% \Longleftrightarrow
% \qquad
% \text{$x$ and $y$ are bisimilar.}
%\end{displaymath}
%This is the case for a wide variety of systems (i.e.\ for a variety of functors $F$).
%However, there is not so much work so far that capture other behavioral
%equivalences (coarser than bisimilarity) by the categorical
%principle of coinduction. 
% Therefore the current work---capturing
% trace semantics by coinduction in a Kleisli category---extends the
% application field of the theory of coalgebras.

\subsection{Our contributions}
Our technical contributions are summarized as follows.
Assume that $T$ is a monad on $\Sets$ which has a suitable order
structure; we shall denote its Kleisli category by $\Kleisli{T}$.
\begin{enumerate}[$\bullet$]
 \item \emph{Trace semantics via coinduction in a Kleisli category.} 
Commutativity of the coinduction diagram
   \begin{equation}\label{diagram:coinductionInKleisliCategoryInIntro}
 \vcenter{
 \xymatrix@C+5em@R=1em{
  {\oF X}
                \ar@{-->}[r]^{\oF (\trace{c})} 
 &
  {\oF Z}
 \\
  {X}
                \ar[u]^{c}
                \ar@{-->}[r]_{\trace{c}}
 &
  {Z}
                \ar[u]_{\text{final}}^{\cong}
}
} \qquad
 \begin{array}{ll}
  \text{in}& \text{$\KT$,}
 \\
  &
  \text{the Kleisli category for $T$}
 \end{array}
 \end{equation}
 is shown to be equivalent to the conventional recursive definition of
       trace semantics
       such as (\ref{equation:conventionalDefOfTraceSetForLTS}) and 
       (\ref{equation:conventionalDefOfTraceDistributionForProbLTS}).
       This is true for both trace set semantics (for
       non-deterministic systems) and trace distribution semantics (for
       probabilistic
       systems).
 The induced arrow $\trace{c}$ thus gives (conventional) trace
 semantics for a system $c$.
 \item \emph{Identification of the final coalgebra in a Kleisli
       category.} We show that
       \begin{center}
	\fbox{  	\begin{tabular}{c}
	 an initial algebra in $\Sets$ 
        \\
         coincides with
        \\
         a final coalgebra in $\Kleisli{T}$.
			\end{tabular}}
       \end{center}
       In particular,
 the final coalgebra in $\Rel$
  is
 the initial algebra in $\Sets$,
  because the category $\Rel$ of sets and relations is a Kleisli
       category for a suitable monad.
  This coincidence happens in the following two steps:
  \begin{enumerate}[-]
   \item the initial algebra in $\Sets$ lifts to a Kleisli category, due
         to a suitable adjunction-lifting result;
   \item in a Kleisli category we have
	 \emph{initial algebra-final coalgebra coincidence}. 
         Here we use the classical result by Smyth and
	 Plotkin~\cite{SmythP82}, namely \emph{limit-colimit coincidence}
         which is applicable in a suitably order-enriched category. 
  \end{enumerate}
\end{enumerate}
Note the presence of two parameters in
(\ref{diagram:coinductionInKleisliCategoryInIntro}):
 a monad $T$ and an endofunctor
$F$, both on $\Sets$. 
 The monad $T$ specifies the \emph{branching type} of systems. We have three
 leading
       examples:\footnote{Other examples include the monad
       $X\mapsto (\mathbb{N}\cup \{\infty\})^{X}$ for multisets, the monad $X\mapsto [0,\infty]^{X}$
       for real valuations, and the monad  $X\mapsto
       \pow(M\times\place)$ with a monoid $M$ for timed systems (cf.~\cite{KickPS06}).
       These monads can be treated in a similar
       way as our leading examples. We leave out the details.
%Detailed accounts for these monads are out of
%       our scope here.
}
       \begin{enumerate}[$\bullet$]
	\item the powerset monad $\pow$ modeling
	      \emph{non-deterministic}
	      or \emph{possibilistic} branching;
	\item the subdistribution monad $\dist$ 
	      \begin{displaymath}
	       \dist X = \{d: X\to[0,1]\mid \sum_{x\in X}d(x) \le 1\}
	      \end{displaymath}
	      modeling
	      \emph{probabilistic} branching; and
	\item the lift monad $\lift = 1+(\place)$ modeling system with 
	      \emph{exception} (or \emph{deadlock},
	      \emph{non-termination}).
% No!!   Kl(id) is not suitably enriched
%	\item and the identity monad $\id$ modeling systems
%	      with no branching (i.e.\ deterministic systems). 
       \end{enumerate}
 The functor $F$ specifies the \emph{transition type} of systems: our
 understanding of  ``transition type''  shall be clarified by the following examples.
       \begin{enumerate}[$\bullet$]
	\item In labeled transition systems (LTSs) with explicit termination---no matter if they are non-deterministic
	      or even probabilistic---a state either
	      \begin{enumerate}[-]
	       \item terminates ($x\to\checkmark$), or
	       \item outputs one symbol and moves to another state
		     ($x\stackrel{a}{\to} x'$), 
	      \end{enumerate} 
	      in
	      one transition.
              This ``transition type'' is expressed by the functor
	      $FX=1+\Sigma\times X$,
	      where $\Sigma$ is the output alphabet and $1 =\{\checkmark\}$.
	\item In context-free grammars (CFGs) as state-based systems, a
	state evolves into a sequence of terminal and non-terminal
	symbols in a transition. The functor
	      \begin{displaymath}
	       FX = (\Sigma + X)^{*}
%               \qquad\text{with $\Sigma$ being the set of terminal symbols}
	      \end{displaymath}
	      with $\Sigma$ being the set of terminal symbols,
              expresses this transition type.
%	     By the way CFGs as systems are non-deterministic because
%	     one non-terminal symbol can produce multiple words.
       \end{enumerate}
Clear separation of branching and transition types is important in
our generic treatment of trace semantics. The transition type $F$
determines the set of linear-time behavior (which is in fact given by the initial
$F$-algebra in $\Sets$).  We model a system by a coalgebra
$X\stackrel{c}{\to}\oF X$ in the Kleisli category $\KT$---see
(\ref{diagram:coinductionInKleisliCategoryInIntro})---where $\oF$ is a suitable
lifting of $F$ in $\KT$.
By the definition of a Kleisli category we will easily see the following bijective
correspondence.
\begin{displaymath}
        \vcenter{\infer={
        \xymatrix@C+1em@1{X\ar[r]^-{c}&TFX} \text{ in $\Sets$}
       }{
        \xymatrix@C+1em@1{X\ar[r]^-{c}&\oF X} \text{ in $\KT$} 
       }}
\end{displaymath}
Hence our system---a \emph{function} of the type $X\to TFX$---first resolves
a branching of type $T$ and then makes a transition of type $F$. 
Many branching systems allow such representation so that  our
generic coalgebraic trace semantics applies to them.

\subsection{Generic theory of traces and simulations}
In the study of coalgebras as `categorical presentation of state-based
systems', there are three ingredients playing crucial roles:
\emph{coalgebras} as systems; \emph{coinduction} yielding process
semantics; and \emph{morphisms of coalgebras} as behavior-preserving
maps. In this paper we study the first two in a Kleisli
category. What about morphisms of coalgebras?

In~\cite{Hasuo06a}  this question is answered by identifying
\emph{lax/oplax morphisms of coalgebras} in a Kleisli category as
\emph{forward/backward simulations}. Use of traces and simulations is a common
technique in formal verification of systems (see e.g.~\cite{LynchV95}): a desirable property is
expressed in terms of traces; and then a system is shown to satisfy the
property
by finding a suitable simulation. Therefore this paper, together
with~\cite{Hasuo06a}, forms an essential part of developing 
a ``generic theory of traces and simulations'' using coalgebras in
a Kleisli category. The categorical genericity---especially the fact
that we can treat non-deterministic and probabilistic branching in a
uniform manner---is exploited in~\cite{HasuoK06} to obtain a
simulation-based proof method for a probabilistic notion of anonymity
for network protocols. Currently we are investigating how much more 
applicational impact can be brought about by our generic theory of traces and simulations.

\subsection{Testing and trace semantics}
Since the emergence of the theory of coalgebras, the significance of
\emph{modal logics} as specification languages has been noticed by many
authors. This is exemplified by the slogan in~\cite{kurz:Phdthesis}: `modal
logic is to coalgebras what equational logic is to algebras'.  Inspired
by coalgebras on Stone spaces and the corresponding modal logic, recent
developments~\cite{KupkeKV04,BonsangueK05,BonsangueK06,Kurz06sigact,PavlovicMW06,Klin07,Klin07b}
have identified the following situation as the essential mathematical
structure underlying  modal logics for coalgebras.
 \begin{displaymath}
 \xymatrix@1@C+3em{
  {\C^{\op}}
              \ar@(lu,ld)[]_{F^{\op}}
              \ar@/^.7em/[r]^{P}
              \ar@{}[r]|{\top}
 &
  {\A}
              \ar@(rd,ru)[]_{M}
              \ar@/^.7em/[l]^{S^{\op}}
 }
 \quad\text{together with}\quad
  MP\stackrel{\delta}{\Longrightarrow} PF^{\op}
 \end{displaymath}
In fact, it is noticed in~\cite{PavlovicMW06} that such a situation not
only hosts a modal logic but also a more general notion of
\emph{testing} (in the sense of~\cite{vanGlabbeek01,StoelingaV03}, also
called \emph{testing scenarios}).
Therefore we shall call the above situation a \emph{testing situation}.

In the last technical section of the paper we investigate
coalgebraic trace semantics for the special case  $T=\pow$ (modeling non-determinism)
 from this testing point of view.
First, we present some basic facts on testing situations, 
especially on the relationship between the induced \emph{testing
equivalence} and the final coalgebra semantics.  
These two process equivalences are categorically presented as
kernel pairs, which enables a fairly simple presentation of the theory
of coalgebraic testing.
In addition, we observe that the coinduction scheme in the Kleisli category $\Kleisli{\pow}$
gives rise to a canonical testing situation, in which the set of tests
is given by an
initial $F$-algebra.

The material on testing in the last section has not been presented
in the earlier versions~\cite{HasuoJ05b,HasuoJS06a} of this paper.

\subsection{Organization of the paper}
In Section~\ref{section:coalgebrasInKleisliCat} we observe that
a coalgebra in a Kleisli category is an appropriate ``denotation'' of a
branching system, when we focus on trace semantics. 
In Section~\ref{section:finalCoalgebraInKleisliCat} we present our main
technical result that 
 an initial algebra in $\Sets$ yields
a final coalgebra in $\Kleisli{T}$.
The relationship to axiomatic domain
theory---which employs similar mathematical arguments---is also
discussed here. Section~\ref{section:traceSemanticsViaCoinduction}
presents some examples of the use of coinduction in $\Kleisli{T}$
and argues that the coinduction principle is a general form of trace semantics.  In
Section~\ref{section:traceSemanticsAsTestingEquivalence}
we review the preceding material from the testing point of view.

\section{Coalgebras in a Kleisli category}\label{section:coalgebrasInKleisliCat}
In the study of coalgebras as 
``categorical presentations of
state-based systems,'' the category $\Sets$ of sets and functions has
been traditionally taken as a base category (see e.g.~\cite{JacobsR97,Rutten00a}).
An important fact in such a setting is that 
bisimilarity is often captured by coinduction.\footnote{Non-examples
include LTSs with unbounded branching degree. They are modeled as
coalgebras for $FX=\pow(\Sigma\times X)$. Lambek's Lemma readily
shows that  this choice of $F$ 
does not have a final coalgebra in $\Sets$, because it would imply
an isomorphism $Z\iso \pow(\Sigma\times Z)$ which is impossible for 
cardinality reasons.
}

However, bisimilarity is not the only process equivalence. In some
applications one would like coarser equivalences, for example in order to abstract
away internal branching structures. One of such coarser semantics, which
has been extensively studied, is \emph{trace
equivalence}.
For example, the process algebra CSP~\cite{Hoare85} has trace semantics
as its operational
model.
Trace equivalence is coarser than bisimilarity, as the following classic example of ``trace-equivalent but
not bisimilar'' systems illustrates. 
 
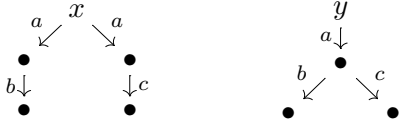
\begin{displaymath}
  \xymatrix@=.7em{
   &
    {x}
          \ar[dl]_{a}
          \ar[dr]^{a}
   &
   \\
    {\bullet}
          \ar[d]_{b}
   &&
    {\bullet}
          \ar[d]^{c}
   \\
    {\bullet}
   &&
    {\bullet}
   }
  \qquad\qquad
  \xymatrix@=.7em{
   &
    {y}
          \ar[d]_{a}
   &
   \\
   &
    {\bullet}
          \ar[ld]_{b}
          \ar[rd]^{c}
   &
   \\
    {\bullet}
   &&
    {\bullet}
   }
\end{displaymath}

It is first noticed in~\cite{PowerT99} that the Kleisli category 
for the powerset monad is an appropriate base category for trace
semantics for non-deterministic systems.
This observation is pursued further in~\cite{Jacobs04c,HasuoJ05b,HasuoJS06a}.
 In~\cite{HasuoJ05c} it is recognized that the same is true for 
the subdistribution monad for probabilistic systems.
The current paper provides a unified framework which yields those
preceding results, in terms of $\Cppo$-enrichment of a Kleisli
category; see Section~\ref{subsection:kleisliIsCppoEnriched}.
 In this section we first aim to justify the use of coalgebras
in a Kleisli category.

\subsection{Monads and Kleisli categories}
Here we recall the relevant facts about monads and Kleisli categories.
For simplicity we exclusively consider monads on $\Sets$.

A \emph{monad} on $\Sets$ is a categorical construct.
It consists of 
\begin{enumerate}[$\bullet$]
 \item an endofunctor $T$ on $\Sets$;
 \item a \emph{unit} natural transformation $\eta:\id\Rightarrow T$, that
       is, a function
       \begin{math}
	X \stackrel{\eta_{X}}{\rightarrow} TX
       \end{math}
       for each set $X$ satisfying a suitable naturality condition; and
 \item a \emph{multiplication} natural transformation $\mu: T^{2}\Rightarrow T$,
       consisting of functions $T^{2}X\stackrel{\mu_{X}}{\rightarrow} TX$ with $X$
       ranging over sets.
\end{enumerate}
The unit and multiplication are required to satisfy the following
compatibility conditions.
\begin{displaymath}
 \xymatrix@C=2em@R=2em{
  {TX}
           \ar[r]^{\eta_{TX}}
           \ar[rd]_{\id}
 &
  {T^{2}X}
           \ar[d]_{\mu_{X}}
 &
  {TX}
           \ar[l]_{T\eta_{X}}
           \ar[ld]^{\id}
 &
 &
  {T^{3}X}
           \ar[r]^{T\mu_{X}}
           \ar[d]_{\mu_{TX}}
 &
  {T^{2}X}
           \ar[d]^{\mu_{X}}
 \\
 &
  {TX}
 &
 &
 &
  {T^{2}X}
           \ar[r]_{\mu_{X}}
 &
  {TX}
}
\end{displaymath}
See~\cite{MacLane71,BarrW85} for the details.

The monad structures play a crucial role in modeling ``branching.''
Intuitively, the unit $\eta$ embeds a non-branching behavior as a trivial
branching (with only one possibility to choose). 
The multiplication $\mu$ ``flattens'' two successive branchings into 
one branching, abstracting away internal branchings:
\begin{equation}\label{diagram:flatteningNonDetBranching}
 \vcenter{\xymatrix@R=-.3em@C=2em{
  &&
   {x}
  \\
  &
   {\bullet}
              \ar@{~>}[ru]\ar@{~>}[rd]
  \\
   {\bullet}
              \ar@{~>}[ru]\ar@{~>}[rdd]
  &&
   {y}
  \\
  \\
  &
   {\bullet}
              \ar@{~>}[r]
  &
   {z}
}}
  \qquad
  \stackrel{\mu}{\longmapsto}
  \qquad
 \vcenter{\xymatrix@R=0em@C=4em{
  &
   {x}
  \\
   {\bullet}
              \ar@/^/@{~>}[ru]
	      \ar@/_/@{~>}[rd]
	      \ar@{~>}[r]
  &
   {y}
  \\
  &
   {z}
}}
\end{equation}
The following examples will illustrate how this flattening phenomenon  is a crucial feature of trace semantics.

In this paper we concentrate on the three monads mentioned in the introduction:
$\lift$, $\pow$ and $\dist$. 
\begin{enumerate}[$\bullet$]
 \item The lift monad $\lift = 1+(\place)$---where we denote $1 = \{\bot\}$
       with $\bot$ meaning \emph{deadlock}---has a standard monad
       structure induced by a coproduct. For example, 
       the multiplication $\mu^{\lift}_{X}: 1+1+X\to 1+X$ carries $x\in X$
       to itself and both $\bot$'s to $\bot$.
%       \begin{displaymath}
%	\vcenter{\xymatrix@R=0em@C+5em{
%         {1+1+X}
%                   \ar[r]^-{\mu^{\lift}_{X}}
%        &
%         {1+X}
%        \\
%         {\bot}
%                   \ar@{|->}[r]
%        &
%         {\bot}
%        \\
%         {x\in X}
%                   \ar@{|->}[r]
%        &
%         {x}              
%	}}
%       \end{displaymath}
 \item The powerset monad $\pow$ has a unit given by singletons and a
       multiplication given by unions. The monad $\pow$ models
       non-deterministic branching: the ``flattening''
       in~(\ref{diagram:flatteningNonDetBranching})
       corresponds to the following application of the multiplication of
       $\pow$.
       \begin{displaymath}
	\vcenter{\xymatrix@R=0em@C+5em{
         {\pow\pow X}
                   \ar[r]^-{\mu^{\pow}_{X}}
        &
         {\pow X}
        \\
         {\bigl\{\,\{x,y\}, \{z\}\,\bigr\}}
                   \ar@{|->}[r]
        &
         {\{x,y,z\}}
	}}
       \end{displaymath}
        The monad $\pow$'s action on arrows (as a functor)  is given by
       direct images: for $f:X\to Y$, the function $\pow f:\pow X\to\pow Y$ carries a subset
       $u\subseteq X$ to the subset $\{f(x)\mid x\in u\}\subseteq Y$.
\item The subdistribution monad $\dist$ has a unit given by the \emph{Dirac
       distributions}.
       \begin{displaymath}
	\vcenter{\xymatrix@R=0em@C+5em{
         {X}
                   \ar[r]^-{\eta^{\dist}_{X}}
        &
         {\dist X}
        \\
         {x}
                   \ar@{|->}[r]
        &
         {\left[
	   \begin{array}{ll}
	    x\mapsto 1 & \\
	    x'\mapsto 0 & (\text{for }x'\neq x) \\
	   \end{array}
	  \right]}
	}}
       \end{displaymath}
       Its multiplication is given by multiplying the probabilities
       along the way. That is,
       \begin{displaymath}
		\mu^{\dist}_{X}(\xi) = \lambda x. \sum_{d\in \dist X}\xi(d)\cdot d(x)\enspace,
       \end{displaymath}
       which models ``flattening'' of the following kind.
       \begin{align*}
	 &\vcenter{\xymatrix@R=-.3em@C=2em{
	   &&
	    {x}
	   \\
	   &
	    {\bullet}
		       \ar@{~>}[ru]^{1/2}\ar@{~>}[rd]_{1/2}
	   \\
	    {\bullet}
		       \ar@{~>}[ru]^{1/3}\ar@{~>}[rdd]_{2/3}
	   &&
	    {y}
	   \\
	   \\
	   &
	    {\bullet}
		       \ar@{~>}[r]_{1}
	   &
	    {z}
	 }}
	   \qquad
	   \stackrel{\mu}{\longmapsto}
	   \qquad
	  \vcenter{\xymatrix@R=.6em@C=4em{
	   &
	    {x}
	   \\
	    {\bullet}
		       \ar@{~>}[ru]^{1/6}\ar@{~>}[rd]_{2/3}\ar@{~>}[r]|{1/6}
	   &
	    {y}
	   \\
	   &
	    {z}
	 }}\enspace,
      \end{align*}
      that is,
      \begin{align*}
	\left[
         \begin{array}{ccc}
	  \left[
           \begin{array}{l}
	    x\mapsto 1/2
           \\
            y\mapsto 1/2
	   \end{array}
          \right]
         &
          \mapsto 
         &
          1/3
         \\[+.3em]
	  [
	    z\mapsto 1
          ]
         &
          \mapsto 
         &
          2/3
	 \end{array}
	\right]
	   \qquad
	   \stackrel{\mu}{\longmapsto}
	   \qquad
        \left[
         \begin{array}{c}
	  x\mapsto 1/6
         \\
	  y\mapsto 1/6
         \\
	  z\mapsto 2/3
	 \end{array}
	\right]\enspace.
       \end{align*}
    The monad $\dist$'s action on arrows (as a functor) is given as a
       suitable adaptation of ``direct images.'' Namely, for $f:X\to Y$,
 the function
   $\dist f:\dist X\to\dist Y$ carries $d\in \dist X$ to $[y\mapsto \sum_{x\in f^{-1}(y)} d(x)]\in\dist Y$.
\end{enumerate}

Given any monad $T$, its \emph{Kleisli category} $\KT$ is
defined as follows. Its objects are the objects of the base
category, hence sets in our consideration. An arrow $X\to Y$ in $\KT$ is the same
thing as an arrow $X\to TY$ in the base category, here $\Sets$. 
\begin{displaymath}
        \vcenter{\infer={
        \xymatrix@C+1em@1{X\ar[r]&TY} \text{ in $\Sets$}
       }{
        \xymatrix@C+1em@1{X\ar[r]& Y} \text{ in $\KT$} 
       }}
% \infer={X\longrightarrow TY\quad\text{in $\Sets$}}
%        {X\longrightarrow Y \quad \text{in $\KT$}}
\end{displaymath}
Identities and composition of arrows are defined using the unit and
the multiplication of $T$. Moreover, there is a canonical adjunction
\begin{equation}\label{diagram:canonicalKleisliAdjunction}
 \vcenter{\xymatrix@R=.5em@C+3em{
%  {X\stackrel{f}{\longrightarrow} Y}
%                    \ar@{|->}[r]
% &
%  {X\stackrel{\eta_{Y}\co f}{\longrightarrow} Y}
% \\
  {\Sets}
              \ar@/^/[r]^{J}
              \ar@/^/[r];[]^{K}
              \ar@{}[r]|{\bot}
 &
  {\KT}
}}
\end{equation}
such that $J$ carries 
  ${X\stackrel{f}{\longrightarrow} Y}$ in $\Sets$
to
  ${X\stackrel{\eta_{Y}\co f}{\longrightarrow} Y}$ in $\KT$.
See~\cite{MacLane71,BarrW85} for details.

The relevance in this paper is that a Kleisli category can be thought of
as a category where the branching is implicit. For example, an arrow
$X\to Y$ in the Kleisli category $\Kleisli{\pow}$ is a function 
$X\to \pow Y$ hence a ``non-deterministic function.'' When $T=\dist$,
then by writing $X\to Y$ in the Kleisli category we mean a function
with probabilistic branching. Moreover, composition of arrows in
$\KT$ is given by
% \begin{displaymath}
%  (Z\stackrel{g}{\leftarrow} Y)\co
%  (Y\stackrel{f}{\leftarrow} X) \quad\text{in $\KT$}\quad
%  =\quad
%  TZ\stackrel{\mu_{Z}}{\longleftarrow}{T^{2} Z} \stackrel{Tg}{\longleftarrow} TY
%   \stackrel{f}{\longleftarrow} X\quad\text{in $\Sets$;}
% \end{displaymath}
 \begin{displaymath}
 X \stackrel{f}{\longrightarrow} 
 Y \stackrel{g}{\longrightarrow} 
 Z
\quad\text{in $\KT$}\quad
  =\quad
 X
   \stackrel{f}{\longrightarrow}   
 TY
   \stackrel{Tg}{\longrightarrow} 
 {T^{2} Z} 
   \stackrel{\mu_{Z}}{\longrightarrow}
  TZ
 \quad\text{in $\Sets$;}
 \end{displaymath}
that is, making one transition (by $g$) after another (by $f$), and then
flattening (by $\mu_{Z}$).
 For example, this general definition instantiates as follows when $T=\dist$.
 For $X\stackrel{f}{\to}Y\stackrel{g}{\to} Z$,
 \begin{equation*}
  (g \co f)(x)(z) = \textstyle\sum_{y \in Y} f(x)(y)\cdot g(y)(z)\enspace.
 \end{equation*}

 \begin{rem}
 Our use of the \emph{sub}-distribution monad instead of the
 distribution monad
	      \begin{displaymath}
	       \dist_{=1} (X) = \{d: X\to[0,1]\mid \sum_{x\in X}d(x) = 1\}
	      \end{displaymath}
 needs some justification. 
 Looking at the trace
  distribution~(\ref{equation:introExampleTraceDistribution}), one sees that
 the probabilities add up only to $2/3$ and not to $1$; this is because
  the infinite trace (namely
 $a^{\omega}\mapsto 1/3$) are not present. Therefore in this example,
 although the state-based system can be modeled as a coalgebra in the category
 $\Kleisli{\dist_{=1}}$, its trace semantics can only be expressed as an
 arrow in $\Kleisli{\dist}$.

 When a system is modeled as a coalgebra in $\Kleisli{\dist}$, a state
 may have a (sub)distribution over possible transitions which adds up to less
 than $1$. In that case the missing probability can be understood as the
 probability for \emph{deadlock}.

 Technically, we use the monad $\dist$ instead of $\dist_{=1}$ because we need the minimum
 element (a \emph{bottom}) so that 
  the Kleisli category becomes $\Cppo$-enriched (Theorem~\ref{theorem:main}).
 A bottom is available for $\dist$ as the zero distribution $[x\mapsto 0]$, but
 not for $\dist_{=1}$.

 \end{rem}

\subsection{Lifting functors by distributive laws}
\label{subsection:liftingFunctorsByDistrLaws}
In this paper a state-based system is presented as a coalgebra
$X\to\oF X$ in $\KT$, where $\oF:\KT\to\KT$ 
is a lifting of $F: \Sets\to\Sets$. This lifting $F\mapsto \oF$ is
equivalent to a \emph{distributive law} $FT\Rightarrow TF$. The rest
of this section elaborates on this point.

 Various kinds of state-based, branching systems
are expressed as a function of the form $X\stackrel{c}{\to} TFX$ with $T$ a monad
(for branching type) and $F$ a functor (for transition type). The
following examples are already hinted at in the introduction.
\begin{enumerate}[$\bullet$]
 \item For $T=\pow$ and $F=1+\Sigma\times\place$, a function $X\stackrel{c}{\to} TFX$ is
       an LTS with explicit termination. For example, consider the following system 
       \begin{displaymath}
	\vcenter{\xymatrix@R=0em@C+5em{
         {X}
                   \ar[r]^-{c}
        &
         {\pow (1+\Sigma\times X)}
        \\
         {x}
                   \ar@{|->}[r]
        &
         {\{\checkmark, (a_{1}, x_{1}), (a_{2}, x_{2})\}}
	}}
       \end{displaymath}
       where $\checkmark$ is the element of $1$.\footnote{Note that the
       singleton $1=\{\checkmark\}$ here in $F=1+\Sigma\times\place$
       has a different interpretation from $1=\{\bot\}$ in
       $T=\lift=1+\place$. The intuition is as follows. On the one hand, when an
       execution hits successful
       termination $\checkmark$, it
       yields its history of observations as its trace.
       On the other hand, when an execution hits deadlock $\bot$ then it
       yields no trace no matter what is the history before hitting
       $\bot$.
       This distinction will be made formal in Example~\ref{example:traceSemForLiftMonadAndLTSWithTerm}.} Then the state $x$ can make
       three possible transitions, namely: $x\to\checkmark$ (successful termination),
       $x\stackrel{a_{1}}{\to}x_{1}$, and $x\stackrel{a_{2}}{\to} x_{2}$,
       when written in a conventional way.
 \item By replacing $T=\pow$ by $\dist$, but keeping $F$ the same, we obtain a probabilistic
       system such as the one in the middle of
       (\ref{diagram:firstExampleOfSystems}). For example, 
       \begin{displaymath}
	\vcenter{\xymatrix@R=0em@C+5em{
         {X}
                   \ar[r]^-{c}
        &
         {\dist (1+\Sigma\times X)}
        \\
         {x'}
                   \ar@{|->}[r]
        &
         {\left[\begin{array}{r}
	         (a,y')\mapsto 1/3
                \\
	         (a,z')\mapsto 1/3
                \\
                 \checkmark \mapsto 1/3
		\end{array}\right]}
	  \shifted{5em}{-1em}{.}
	}}
       \end{displaymath}
 \item For $T=\pow$ and $F=(\Sigma+\place)^{*}$, a function $X\stackrel{c}{\to} TFX$ 
       is a CFG with $\Sigma$ the terminal alphabet (but without
       finiteness conditions e.g.\ on the state space). See~\cite{HasuoJ05b} for more details.
\end{enumerate}
All these systems are modeled by a function $X\stackrel{c}{\to} TFX$, hence an arrow $X\stackrel{c}{\to}FX$
in $\KT$. Our question here is: is $c$ a coalgebra in
$\KT$? In other words: is the functor $F$ on $\Sets$ also a
functor on $\KT$?

Hence, to develop a generic theory of traces in $\KT$, we need
to lift $F$ to a functor $\oF$ on $\KT$. A functor
$\oF$ is said to be a \emph{lifting} of
$F$ if the following diagram commutes. Here $J$ is the left adjoint in~(\ref{diagram:canonicalKleisliAdjunction}).
\begin{equation}\label{diagram:definitionOfLiftingOfF}
 \xymatrix@R=1em@C+2em{
   {\KT}
                   \ar[r]^{\oF}
  &
   {\KT}
  \\
   {\Sets}
                   \ar[u]^{J}
                   \ar[r]_{F}
  &
   {\Sets}
                   \ar[u]_{J}
}
\end{equation}
The following fact is presented in~\cite{Murly94}; see also~\cite{LenisaPW00,LenisaPW04}.
Its proof is straightforward.
\begin{lem}\label{lemma:liftingIffDistrLaw}
 A lifting $\oF$ of $F$ is in bijective correspondence with a
 \textbf{distributive law} $\lambda: FT{\Rightarrow} TF$. 
 A distributive law $\lambda$ is a natural transformation
 which is compatible with $T$'s monad structure, in the following way.
 \begin{align*}
  \xymatrix@R=1em@C+1em{
   {FX}
                 \ar[r]^{F\eta_{X}}
                 \ar[rd]_{\eta_{FX}}
  &
   {FTX}
                 \ar[d]^{\lambda_{X}}
  &
  &
   {FT^{2}X}
                 \ar[r]^{\lambda_{TX}}
                 \ar[d]_{F\mu_{X}}
  &
   {TFTX}
                 \ar[r]^{T\lambda_{X}}
  &
   {T^{2}FX}
                 \ar[d]^{\mu_{FX}}
  \\
  &
   {TFX}
  &&
   {FTX}
                 \ar[rr]_{\lambda_{X}}
  &&
   {TFX}
} 
 \end{align*} 

\vspace{-1.8em}
\qed
\end{lem}
\noindent A distributive law $\lambda$ induces a lifting $\oF$ as
follows. On objects: $\oF X=FX$. Given $f: X\to Y$ in
 $\KT$, we need an arrow $\oF f: FX\to FY$ in $\KT$. 
 Recall that $f$ is a function $X\to TY$ in $\Sets$; one takes $\oF f$
 to be the arrow
 which corresponds to the function 
 \begin{displaymath}
  FX \stackrel{Ff}{\longrightarrow} FTY \stackrel{\lambda_{Y}}{\longrightarrow} TFY \quad
  \text{in $\Sets$}\enspace.
 \end{displaymath}

A distributive law specifies how
a transition (of type $F$)
 ``distributes'' over 
 a branching (of type $T$). Let us look at an example. 
For $T=\pow$ and $F=1+\Sigma\times\place$ (the combination for LTSs with
explicit termination), we have the following distributive law.
       \begin{displaymath}
	\vcenter{\xymatrix@R=0em@C+5em{
         {1+\Sigma\times(\pow X)}
                   \ar[r]^{\lambda_{X}}
        &
         {\pow (1+\Sigma\times X)}
        \\
         {\checkmark}
                   \ar@{|->}[r]
        &
         {\{\checkmark\}}
        \\
         {(a, S)}
                   \ar@{|->}[r]
        &
         {\bigl\{(a,x)\mid x\in S\bigr\}}                
	}}
%	 \lambda_{X}
%        \;:\;
%         {1+\Sigma\times(\pow X)}
%        \longrightarrow
%         {\pow (1+\Sigma\times X)} 
%        \enspace,\qquad
%         {\checkmark}
%        \longmapsto
%         {\{\checkmark\}}
%        \enspace,\qquad
%         {(a, S)}
%        \longmapsto
%         {\bigl\{(a,x)\mid x\in S\bigr\}}                
        \end{displaymath}
    For example,
        \begin{displaymath}
         \vcenter{\xymatrix@R=0em@C-.5em{
	  &&
           {x}
          \\
           {\bullet}
                             \ar[r]^{a}
          &
           {}
                             \ar@{~>}[ru]
                             \ar@{~>}[r]
                             \ar@{~>}[rd]
          &
           {y}
          \\
          &&
           {z}
         }}
         \;\stackrel{\lambda}{\longmapsto}\;
         \vcenter{\xymatrix@R=0em@C-.5em{
	  &
           {}   
                             \ar[r]^{a}
	  &
           {x}
          \\
           {\bullet}
                             \ar@{~>}[ru]
                             \ar@{~>}[r]
                             \ar@{~>}[rd]
          &
           {}
                             \ar[r]^{a}
          &
           {y}
          \\
          &
                     {}   
                             \ar[r]^{a}
          &
           {z}
         }}
        \qquad\text{that is}\qquad
        \bigl(a,\{x,y,z\}\bigr) \stackrel{\lambda}{\longmapsto} \bigl\{(a,x), (a,y), (a,z)\bigr\}\enspace,
       \end{displaymath}
       where waving arrows $\leadsto$ denote branchings.

Throughout the paper
we need the global assumption  that a functor $F$ has
a lifting $\oF$ on $\KT$, or equivalently, that there is a
distributive law $\lambda: FT\Rightarrow TF$. Now we present some
sufficient conditions for existence of $\lambda$. In most examples one
of these conditions holds.

First, take $T=\pow$, in which case we have $\Kleisli{\pow}\cong\Rel$, the
category of sets and binary relations.
 We can provide the following
condition that uses relation liftings, whose definition is found~\cite{Jacobs04c}.
\begin{lem}[From~\cite{Jacobs04c}]
\label{lemma:distributiveLawForPowersetMonad}
Let $F\colon\Sets\rightarrow\Sets$ be a functor that preserves weak
pullbacks. Then there exists a distributive law $\lambda\colon F\pow
\Rightarrow \pow F$ given by
\begin{displaymath}
 \lambda_{X}(u)
 = 
 \bigl\{\, v\in FX\mid (v,u)\in\Relf(\in_{X})\,\bigr\}\enspace,
\end{displaymath}
 where $u\in F\pow X$ and $\Relf(\in_{X})\subseteq FX\times F\pow X$ is
 the $F$-relation lifting of the membership relation $\in_{X}$. \qed
\end{lem}
\noindent In fact, the functor $\oF:\Rel\to\Rel$ induced by this distributive
law carries an arrow $R: X\to Y$ in $\Kleisli{\pow}$---which is a
binary relation between $X$ and $Y$---to its $F$-relation lifting $\Relf(R)$. That is, 
\begin{equation}\label{equation:liftingOfFIsRelationLifting}
 \oF R = \Relf(R)\quad :\; FX\longrightarrow F Y% \quad\text{in $\Kleisli{\pow}\cong\Rel$.}
\end{equation}
in $\Kleisli{\pow}\cong\Rel$.

Now let us consider a monad $T$ which is not $\pow$. When a monad $T$ is \emph{commutative} and a functor $F$
is \emph{shapely}, we can provide a canonical distributive law. 
The class of such monads and functors is wide and all the examples in
this paper are contained.
\begin{enumerate}[$\bullet$]
 \item A \emph{commutative} monad~\cite{Kock70a} is 
       intuitively a monad whose corresponding algebraic theory 
       has only commutative operators. We exploit the fact that 
       a commutative monad is equipped with an arrow 
       called \emph{double strength}
       \begin{displaymath}
	\dst_{X,Y}: TX\times TY \longrightarrow T(X\times Y)
       \end{displaymath}
       for any sets $X$ and $Y$; the double strength must be compatible with the monad
       structure of $T$ in an obvious way.
%       This notion is also used in~\cite{Pardo01} for a similar purpose.

       Our three examples of monads are all commutative, with the
       following double strengths.
       \begin{equation}\label{equation:doubleStrengthForExamplesOfMonads}
       \begin{array}{llll}
       \dst_{X,Y}^{\lift}(u,v) & = & \left\{\begin{array}{ll}
       (u,v) & \text{ if } u \in X \text{ and } v \in Y,\\
       \bot & \text { if } u = \bot \text{ or } v = \bot, \\
       \end{array}\right.\qquad
%       & \text{for $u\in \lift X$ and $v\in \lift Y$,} 
       \\
       \dst_{X,Y}^{\pow}(u,v) & = & u \times v\enspace,
%       & \text{for $X'\in\pow X$ and $Y'\in\pow Y$,}
       \\
       \dst_{X,Y}^{\dist}(u,v) & = & 
       \lambda (x,y).\;\, u(x)\cdot v(y)\enspace.
%       & \text{for $d\in \dist X$ and $e\in \dist Y$.}
       \end{array}
       \end{equation}

 \item The family of \emph{shapely
       functors}~\cite{Jay95}\footnote{Shapely functors here are called
       \emph{polynomial} functors by some authors, although other authors
       allow infinite powers or the powerset construction.}
       on $\Sets$ is defined inductively by the following BNF notation:
\begin{equation*}
  F ::= \id\mid {\Sigma}\mid F_{1}\times F_{2}\mid 
   {\textstyle\coprod_{i\in I}}F_{i}\enspace,
\end{equation*}
where $\Sigma$ denotes the constant functor into an arbitrary set
       $\Sigma$.
       Notice that taking infinite product is not allowed, nor 
       exponentiation to the power of an infinite set. This is 
       in order to ensure that we find an initial $F$-algebra
       as a suitable $\omega$-colimit---see
       Proposition~\ref{LemInitialAlgebraViaInitialSequence}.
\end{enumerate}

%The construction of a distributive law is done inductively on the
%construction of shapely $F$.

\begin{lem}
\label{lemma:distrLawBetweenCommMonadAndShapelyFunctor}
Let $T\colon\Sets\rightarrow\Sets$ be a commutative monad, and
$F\colon\Sets\rightarrow\Sets$ a shapely functor. Then there
is a distributive law $\lambda\colon FT \Rightarrow TF$.
\end{lem}
\proof
The construction of a distributive law is done inductively on the
construction of shapely $F$.
\begin{enumerate}[$\bullet$]
\item If $F$ is the identity functor, then the $\lambda$ is the
identity natural transformation $T\Rightarrow T$.

\item If $F$ is a constant functor, say $X\mapsto \Sigma$, then 
$\lambda$ is the unit $\eta_{\Sigma}\colon \Sigma\rightarrow T\Sigma$ at $\Sigma\in\Sets$.

\item If $F = F_{1}\times F_{2}$ we use induction in the form of
distributive laws $\lambda^{F_{i}}\colon F_{i}T \Rightarrow TF_{i}$ for 
$i\in\{1,2\}$ to form the composite:
$$
\xymatrix@1@C+2em{
  {F_{1}TX\times F_{2}TX}
               \ar[r]^{\lambda^{F_1}\times\lambda^{F_2}}
 &
  {TF_{1}X \times TF_{2}X}
               \ar[r]^{\mathsf{dst}}
 & 
  {T(F_{1}X\times F_{2}X)\enspace.}
}
$$

\item If $F$ is a coproduct $\coprod_{i\in I}F_{i}$ then we use
laws $\lambda^{F_i}\colon F_{i}T \Rightarrow TF_{i}$ for 
$i\in I$ in:
\begin{displaymath}
 \diagram
 {\coprod_{i\in I}F_{i}(TX)}\xto[rrr]^-{[T(\kappa_{i}) \co \lambda^{F_i}]_{i\in I}} 
   & & & {T(\coprod_{i\in I}F_{i}X)\enspace.}
 \enddiagram
% \tag*{$\myQEDbox$}
\end{displaymath}
\end{enumerate}
It is straightforward to check that such $\lambda$ is natural and
compatible with the monad structure.
\qed

We have provided some \emph{sufficient} conditions for a distributive
law to exist, that is, for a functor $F$ to be lifted to $\KT$.  This does not mean the results
in the sequel hold exclusively for commutative monads and shapely functors.

\subsection{Order-enriched structures of Kleisli categories}
\label{subsection:kleisliIsCppoEnriched}
The notion of branching naturally involves a partial order:
one branching is bigger than another if the former offers ``more
possibilities'' than the latter. 
Formally, this order
appears as the \emph{$\Cppo$-enriched structure} of a Kleisli category.
 It plays an important role in the initial algebra-final coalgebra
coincidence in Section~\ref{subsection:finalCoalgebraInKleisli}.

A \emph{$\Cppo$-enriched category} $\C$ is a category where:
\begin{enumerate}[$\bullet$]
 \item Each homset $\C (X,Y)$ carries a partial order   
       $\sqsubseteq$ as in
       \begin{displaymath}
	\vcenter{\xymatrix@1@C+3em{
         {X}
               \ar@/^/[r]^{g}
               \ar@/_/[r]_{f}
               \ar@{}[r]|*[@u]{\sqsubseteq}
	&
         {Y}
	}}
       \end{displaymath}
       which makes $\C (X,Y)$ an $\omega$-cpo with a bottom. This means:
       \begin{enumerate}[-]
	\item for an increasing $\omega$-chain of arrows from $X$ to $Y$,
	      \begin{displaymath}
	      f_{0}\sqsubseteq f_{1}\sqsubseteq \dotsc\quad:\quad X\longrightarrow Y\enspace,
	      \end{displaymath}
              there exists its join $\bigsqcup_{n<\omega} f_{n}: X\to Y$;
	\item for any $X$ and $Y$ there exists a bottom arrow
	      $\bot_{X,Y}: X\to Y$ which is the minimum in $\C(X,Y)$.
       \end{enumerate}
 \item Moreover, composition of arrows is continuous as a function
       $\C(X,Y)\times\C(Y,Z)\to\C(X,Z)$.
       This means that the following joins are preserved:\footnote{This
       component-wise preservation of joins is equivalent to 
       the continuity of the composition function. See~\cite[Lemma~3.2.6]{AbramskyJ94}.}
	\begin{displaymath}
	 \begin{array}{rclcrcl}
	 g\co \left(\bigsqcup_{n<\omega}f_{n}\right)
	 & = &
	 \bigsqcup_{n<\omega}(g\co f_{n})
	 & \mbox{ and } &
	 \left(\bigsqcup_{n<\omega}f_{n}\right) \co h
	 & = &
	 \bigsqcup_{n<\omega}(f_{n}\co h)\enspace.
	 \end{array}
	\end{displaymath}
       Note that composition need not preserve bottoms (i.e.\ it is not necessarily \emph{strict}).
\end{enumerate}
This is in fact an instance of a more general notion of
\emph{$\V$-enriched categories} where $\V$ is the category $\Cppo$ of
pointed (i.e.\ with $\bot$) cpo's and continuous (but not necessarily strict) functions.
See~\cite{Lawvere73,Kelly82,Borceux94} for more details on enriched
category theory, and~\cite{AbramskyJ94} on cpo's and domain theory.

\begin{lem}\label{lemma:KleisliIsCppoEnriched}
 For our three examples $\lift$, $\pow$ and $\dist$ of a monad $T$, the
 Kleisli category $\KT$ is $\Cppo$-enriched. Moreover, composition of arrows
 is left-strict: $\bot\co f =\bot$.
\end{lem}
\noindent The left-strictness of composition will be necessary later.
\proof
 Notice first that a set $TY$ for $T\in\{\lift,\pow,\dist\}$ carries a cpo
 structure with $\bot$. The set $\lift Y=\{\bot\}+Y$ carries the flat
 order with a bottom:
 \begin{displaymath}
  \xymatrix@R=.8em@C=.2em{
   {y}
  &&
   {y'}
  &&
   {y''}
  &&
   {\cdots}
  \\
  &&&
   {\bot}
           \ar@{-}[lllu]
           \ar@{-}[lu]
           \ar@{-}[ru]
}
 \end{displaymath}
 embodying the idea that $\bot$ denotes non-termination or \emph{deadlock}---in contrast to
 $\checkmark$ for successful termination.
 The set $\pow Y$ carries an inclusion order; in $\dist Y$ we define
 $d\sqsubseteq e$ if $d(y)\leq e(y)$ for each $y\in Y$.
 The bottom element in $\dist Y$ is the zero distribution $[ y\mapsto 0]$: this belongs
 to the set $\dist Y$ because $\dist$ is the \emph{sub}-distribution monad.

 The cpo structure of a homset $\KT(X, Y)$ comes from that of
 $TY$ in a pointwise manner: 
 \begin{displaymath}
	\vcenter{\xymatrix@1@C+3em{
         {X}
               \ar@/^/[r]^{g}
               \ar@/_/[r]_{f}
               \ar@{}[r]|*[@u]{\sqsubseteq}
	&
         {Y}
	}}\qquad \text{if and only if}\qquad
        \forall x\in X.\;\; f(x) \sqsubseteq_{TY} g(x)\enspace.
 \end{displaymath}
 It is  laborious but straightforward to show that composition in $\KT$ is continuous and 
 left-strict.
 \qed

\vskip 1em

We are concerned with coalgebras $X\to \oF X$ in the category $\KT$, which we assume
is $\Cppo$-enriched. Hence it
comes natural to require that functor $\oF$ is somehow compatible with
the $\Cppo$-enriched structure of $\KT$. The obvious choice is
to require that $\oF$ is a \emph{$\Cppo$-enriched functor} (see
e.g.~\cite{Borceux94}), i.e.\  $\oF$ is \emph{locally
continuous}. It means that for an increasing $\omega$-chain $f_{n}: X\to
Y$, we have
  \begin{displaymath}
   \oF(\bigsqcup_{n<\omega} f_{n}) = \bigsqcup_{n<\omega} (\oF f_{n})\enspace.
  \end{displaymath}
This is indeed the assumption chosen in axiomatic domain theory. We will
come back to this point later in
Section~\ref{subsection:relatedWorkAxiomaticDomainTheory}.
However, for our later purpose, we only need the weaker condition of
\emph{local monotonicity}: $f\sqsubseteq g$ implies $\oF f \sqsubseteq
\oF g$.

For a monad $T=\{\lift,\pow,\dist\}$ and a shapely functor $F$ (recall
Lemma~\ref{lemma:distrLawBetweenCommMonadAndShapelyFunctor}), the lifted
$\oF$ is indeed locally continuous. We emphasize again that this does
not mean our  results in Section~\ref{section:finalCoalgebraInKleisliCat} hold exclusively for shapely functors.
\begin{lem}
Let $F$ be a shapely functor and $T \in \{\lift, \pow, \dist\}$.
The lifting $\oF:\KT\to\KT$ induced by
 Lemma~\ref{lemma:distrLawBetweenCommMonadAndShapelyFunctor}
is locally continuous. % Moreover it is strict: $\oF(\bot_{X,Y})=\bot_{FX,FY}$.
\end{lem}
\proof
 By induction on the construction of shapely functors. 
%We actually prove that $\oF$ obtained by lifting $F$ with help
% of the distributive law presented in Lemma~\ref{GenPowerLawLem} 7 is continuous.\\
\begin{enumerate}[$\bullet$]
 \item $F=\id$, the identity functor. Then $\oF=\id$ which
       satisfies the condition.
 \item $F=\Sigma$, a constant functor. Then $\oF$ maps every
       arrow to the identity map on $\Sigma$ in $\KT$.
       This is obviously locally continuous.
 \item $F=F_{1}\times F_{2}$. 
       First notice that, for $f: X\to Y$ in $\KT$, 
       we obtain $\oF f$ as the following composite in $\Sets$.
       \begin{displaymath}
	\xymatrix@C+5em@R=1em{
	 {F_{1}X\times F_{2}X}
	             \ar[r]^{\overline{F_{1}} f\times \overline{F_{2}}f}
	             \ar[rd]_{\oF f}
        &
         {TF_{1}Y\times TF_{2}Y}
                     \ar[d]^{\dst_{F_{1}Y, F_{2}Y}}
        \\
        &
         {T(F_{1}Y\times F_{2}Y)}
	}
       \end{displaymath}
       Because the order in $\KT(FX, FY)$ is pointwise, 
       it suffices to show the following:
       $\dst: TX\times TY\to T(X\times Y)$ is
       a continuous map between cpo's. 
       It is easy to check that this is indeed the case. See (\ref{equation:doubleStrengthForExamplesOfMonads}).
 \item $F=\coprod_{j\in J}F_{j}$. 
       For $f: X\to Y$ in $\KT$, we obtain the map $\oF f$ as
       the composite $[T\kappa_{j}]_{j\in J}\co \coprod_{j\in J}\Kleisli{F_{j}}(f)$ in $\Sets$.
       Since the order on the homset is pointwise, it suffices to show that each $T\kappa_{j}: TF_{j}Y\to T(\coprod_{j\in J}F_{j}Y)$ is
       continuous. This is easy. \qed
\end{enumerate}

\section{Final coalgebra in a Kleisli category}\label{section:finalCoalgebraInKleisliCat} 
In this section we shall prove our main technical result: the initial
$F$-algebra in $\Sets$ yields the final $\oF$-coalgebra in $\KT$.  It
happens in the following two steps: first, the initial algebra in
$\Sets$ is lifted to the initial algebra in $\KT$; second we
have the initial algebra-final coalgebra coincidence in $\KT$. For the
latter we use the classical result~\cite{SmythP82} of limit-colimit
coincidence. This is where the $\Cppo$-enriched structure of $\KT$ plays
a role.

In the proof we use two standard constructions: initial/final
sequences~\cite{AdamekK79} and limit-colimit
coincidence~\cite{SmythP82}.  The reader who is not familiar with these
constructions is invited to look at
Appendices~\ref{appendix:preliminariesInitialFinalSequence}
and~\ref{appendix:preliminariesLimitColimitCoincidence} where we briefly
recall them.

 \begin{rem}
 The proof of our main theorem (Theorem~\ref{theorem:main}) can be
 simplified if we suitably strengthen the assumptions. First, if we
 assume local \emph{continuity} of the lifted functor $\oF$ (instead of local
 \emph{monotonicity} that is assumed in our main theorem), then the
 initial algebra-final coalgebra coincidence  follows from a
 standard result in axiomatic domain theory; see
 Section~\ref{subsection:relatedWorkAxiomaticDomainTheory}.  Furthermore,
 for the special case $T=\pow$ in which case $\Kleisli{\pow}\cong\Rel$,
 the initial algebra-final coalgebra coincidence is almost obvious due
 to the duality $\Rel\cong\Rel^{\op}$; see
 Section~\ref{subsection:simplerProofInRel}.
 \end{rem}

%Now we prove our main technical result: for a monad $T$
%with a suitable order structure, an initial algebra in $\Sets$ yields a 
%final coalgebra in $\Kleisli{T}$.
\subsection{The initial algebra in \texorpdfstring{$\Sets$}{Sets} is
 the final coalgebra in \texorpdfstring{$\Kleisli{T}$}{Kl(T)}}%
\label{subsection:finalCoalgebraInKleisli}
First, it is standard that an initial algebra in $\Sets$ is lifted
to an initial algebra in $\Kleisli{T}$. Such a phenomenon is studied for
instance
in~\cite{Fokkinga94,Pardo01} in the context of combining datatypes
(modeled by an initial algebra) and effectful computations (modeled by
a Kleisli category). For this result we do not need an order structure.
\begin{prop}\label{proposition:initialAlgebraLiftsToKleisli}
 Let $T$ be a monad and $F$ be a endofunctor, both on a category $\C$.
 Assume that we have a distributive law $FT\Rightarrow TF$---or
 equivalently, we have a lifting $\oF$ on $\KT$. If $F$ has an initial
 algebra $\alpha: FA\iso A$ in $\C$, then
  \begin{displaymath}
   J\alpha = \eta_{A}\co \alpha \; :\quad \oF A\longrightarrow A\qquad \text{in $\KT$}
  \end{displaymath}
 is an initial $\oF$-algebra. Here $J$ is the canonical Kleisli left
 adjoint as in~(\ref{diagram:canonicalKleisliAdjunction}).
\end{prop}
\noindent We will use an instance of this result for
$\C=\Sets$. 
\proof
It follows from~\cite[Theorem~2.14]{HermidaJ98} that a distributive law
lifts the canonical Kleisli adjunction to an adjunction between
the categories $\Alg{F}$ and $\Alg{\oF}$ of algebras.
\begin{displaymath}
 \vcenter{\xymatrix@C+4em@R=1.5em{
  {\Alg{F}}
              \ar@{->}@/^/[r]^{J'}
              \ar@{->}@/^/[r];[]
              \ar@{}[r]|{\bot}
              \ar[d]
 &
  {\Alg{\oF}}
              \ar[d]
 \\
  {\C}
              \ar@/^/[r]^{J}
              \ar@/^/[r];[]^{K}
              \ar@{}[r]|{\bot}
 &
  {\KT} 
}}
\end{displaymath}
The left adjoint $J'$ preserves the initial object (see
e.g.~\cite{MacLane71}). \qed

Second, we use the initial algebra-final coalgebra coincidence in
$\Kleisli{T}$---which holds in a suitable order-enriched setting---to
identify the final coalgebra in $\Kleisli{T}$. This is our main theorem.
\begin{thm}[Main theorem]\label{theorem:main}
Assume the following:
 \begin{enumerate}[\em(1)]
  \item\label{assumption:KleisliIsCppoEnrichedAndLeftStrict} A monad $T$ on $\Sets$ is such that its Kleisli category
	$\Kleisli{T}$ is $\Cppo$-enriched and composition in
	$\Kleisli{T}$ is left-strict.
  \item For an endofunctor $F$ on $\Sets$, we have a distributive law $\lambda: FT\Rightarrow TF$.
        Equivalently, $F$ has a lifting $\oF$ on
	$\Kleisli{T}$. Moreover, the lifting $\oF$ is locally monotone.
  \item\label{assumption:FPreservesOmegaColimit} The functor $F$
       preserves $\omega$-colimits in $\Sets$, hence has an
	initial algebra via the initial sequence (see Proposition~\ref{LemInitialAlgebraViaInitialSequence}).
 \end{enumerate}
Then
the initial 
$F$-algebra $\alpha: FA\iso A$ yields a final $\oF$-coalgebra in
 $\Kleisli{T}$ by
  \begin{displaymath}
    (J\alpha)^{-1} = J(\alpha^{-1}) = \eta_{FA}\co \alpha^{-1}\; :\quad A\longrightarrow \oF A\qquad \text{in $\KT$}\enspace.
  \end{displaymath}
\end{thm}
\noindent 
We first present the main line of the proof. Some details are provided in the form of subsequent lemmas.
Note that the assumptions are satisfied by $T\in\{\lift,\pow,\dist\}$ and shapely
$F$; see Lemmas~\ref{lemma:KleisliIsCppoEnriched} and
\ref{lemma:distrLawBetweenCommMonadAndShapelyFunctor}.
\proof
By the assumption~(\ref{assumption:FPreservesOmegaColimit})
 we obtain the initial algebra via the initial sequence in $\Sets$.
\begin{equation}\label{Diagram:initSeqInSets}
{   \vcenter{\xymatrix@C+1em@R-.8em{
   \shifted{-1em}{0em}{\text{In $\Sets$}}
  &&&&
   A
   \shifted{4em}{0em}{\text{(colimit)}}
            \ar@<1.5ex>@/^1pc/[dd]^-{\alpha^{-1}}_-{\cong}
  \\
   {\cdots}
            \ar[r]^{F^{n-1}{\fromInit}}
  &
   F^{n}0
            \ar@/^1pc/[rrru]|{\alpha_{n}}
            \ar@/_1pc/[rrrd]|{F\alpha_{n-1}}
            \ar[r]
  &
   F^{n+1}0
            \ar@/^/[rru]|{\alpha_{n+1}}
            \ar@/_/[rrd]|{F\alpha_{n}}
            \ar[r]
  &
   {\cdots}
  &
  \\
  &&&&
   FA
   \shifted{4em}{0em}{\text{(colimit)}}
            \ar@<1.5ex>@/_1pc/[uu]^-{\alpha}
}}
}\end{equation}
Here $0=\emptyset\in\Sets$ is initial and $\fromInit: 0\to X$ is the
unique arrow from $0$ to an arbitrary $X$.
We apply the functor $J: \Sets\to\Kleisli{T}$ to the whole diagram.
Since $J$ is a left adjoint it preserves colimits: hence the two cocones
in the following diagram are both colimits again.
\begin{equation}\label{Diagram:InitialSeqMappedByJ}
{   \vcenter{\xymatrix@C+1em@R-0.8em{
   {\text{In $\Kleisli{T}$}}
  &&&&
   A
   \shifted{4em}{0em}{\text{(colimit)}}
            \ar@<1.5ex>@/^1pc/[dd]^-{J\alpha^{-1}}_-{\cong}
  \\
   {\cdots}
            \ar[r]^{JF^{n-1}{\fromInit}}
  &
   F^{n}0
            \ar@/^1pc/[rrru]|{J\alpha_{n}}
            \ar@/_1pc/[rrrd]|{JF\alpha_{n-1}}
            \ar[r]
  &
   F^{n+1}0
            \ar@/^/[rru]|{J\alpha_{n+1}}
            \ar@/_/[rrd]|{JF\alpha_{n}}
            \ar[r]
  &
   {\cdots}
  &
  \\
  &&&&
   FA
   \shifted{4em}{0em}{\text{(colimit)}}
            \ar@<1.5ex>@/_1pc/[uu]^-{J\alpha}
}}
}\end{equation}
The $\omega$-chain in this diagram is in fact the initial sequence for the
functor $\oF $ (Lemma \ref{lemmaForMain:indeedInitialSequence}) because,
for example, a left adjoint $J$ preserves initial 
objects.
Moreover the lower cone is the image of the upper cone under $\oF$; see
the diagram~(\ref{diagram:definitionOfLiftingOfF}).
Hence the diagram (\ref{Diagram:InitialSeqMappedByJ}) is equal to the
 following one. Recall that $\oF X = FX$ on objects.
% where
%${\fromInit}$ denotes the unique arrow ${\fromInit}: 0\to F0$ in
%$\Kleisli{T}$.
\begin{equation}\label{Diagram:InitialSeqInKleisliCat}
{   \vcenter{\xymatrix@C+1em@R-0.8em{
   {\text{In $\Kleisli{T}$}}
  &&&&
   A 
   \shifted{4em}{0em}{\text{(colimit)}}
            \ar@<1.5ex>@/^1pc/[dd]^-{J\alpha^{-1}}_-{\cong}
  \\
   {\cdots}
            \ar[r]^{\oF ^{n-1}{\fromInit}}%_{\txt{Initial sequence\\ for $\oF $}}
  &
   {\oF}^{n}0
            \ar@/^1pc/[rrru]|{J\alpha_{n}}
            \ar@/_1pc/[rrrd]_{\oF J\alpha_{n-1}}
            \ar[r]
  &
   {\oF}^{n+1}0
            \ar@/^/[rru]|{J\alpha_{n+1}}
            \ar@/_/[rrd]|{\oF J\alpha_{n}}
            \ar[r]
  &
   {\cdots}
  &
  \\
  &&&&
   {\oF}A
   \shifted{4em}{0em}{\text{(colimit)}}
            \ar@<1.5ex>@/_1pc/[uu]^-{J\alpha}
}}
}\end{equation}
Thus Proposition \ref{LemInitialAlgebraViaInitialSequence} 
yields that $J\alpha: {\oF}A\iso A$ is an initial $\oF$-algebra.
This can be seen as a more concrete proof of
Proposition~\ref{proposition:initialAlgebraLiftsToKleisli}.

Now we show the initial algebra-final coalgebra coincidence in
$\Kleisli{T}$.
This is done by reversing all the arrows in
(\ref{Diagram:InitialSeqInKleisliCat}) and transforming the diagram into
the one of the
final sequence and its limits.

We notice 
(Lemma
\ref{lemmaForMain:arrowsInInitialSequenceAreEmbeddings}) 
that each arrow $\oF^{n}{\fromInit}$ in the initial sequence
is an embedding (Definition~\ref{definition:embeddingProjectionPairs}).
Hence the limit-colimit coincidence Theorem \ref{limit_colimit_coincidence} says that
every arrow in the diagram is an embedding.
Note that $J\alpha$ and $J\alpha^{-1}$, inverse to each other,
form an embedding-projection pair.

By taking the corresponding projections---they are uniquely determined 
(Lemma~\ref{proposition:embeddingDeterminesCorrProjection}) and are
denoted by $(\place)^{P}$---we obtain the next diagram.
The limit-colimit coincidence Theorem \ref{limit_colimit_coincidence} says that
the two resulting cones are both limits.
It is also obvious that the whole diagram commutes.
\begin{equation}\label{Diagram:TakenProjectionsOfInitialSequenceAndItsColimits}
{    \vcenter{\xymatrix@C+1em@R-0.8em{
   {\text{In $\Kleisli{T}$}}
  &&&&
   A 
   \shifted{4em}{0em}{\text{(limit)}}
            \ar@<-1.5ex>@/_1pc/[dd];[]_-{(J\alpha^{-1})^{P}}^-{\cong}
  \\
   {\cdots}
            \ar[r];[]_{(\oF ^{n-1}{\fromInit})^{P}}%_{\txt{Initial sequence\\ for $\oF $}}
  &
   {\oF}^{n}0
            \ar@/_1pc/[rrru];[]|{(J\alpha_{n})^{P}}
            \ar@/^1pc/[rrrd];[]^{(\oF J\alpha_{n-1})^{P}}
            \ar[r];[]
  &
   {\oF}^{n+1}0
            \ar@/_/[rru];[]|{(J\alpha_{n+1})^{P}}
            \ar@/^/[rrd];[]|{(\oF J\alpha_{n})^{P}}
            \ar[r];[]
  &
   {\cdots}
  &
  \\
  &&&&
   {\oF}A
   \shifted{4em}{0em}{\text{(limit)}}
            \ar@<-1.5ex>@/^1pc/[uu];[]_-{(J\alpha)^{P}}
}}}
\end{equation}
The $\omega^{\op}$-chain here is indeed a final sequence: Lemma
\ref{OInFinLem} shows---using the assumption 
(\ref{assumption:KleisliIsCppoEnrichedAndLeftStrict}) on 
left-strictness---that
 $0$ is also final in $\Kleisli{T}$, and according to
Lemma \ref{lemmaForMain:arrowsInInitialSequenceAreEmbeddings}
we have $(\oF^{n}{\fromInit})^{P}=\oF^{n}{\toTerm}$
where $\toTerm: X\to 0$ is the unique arrow to the final object $0$ in $\Kleisli{T}$.
As to the lower cone we have 
\begin{math}
 \bigl(\oF J\alpha_{n}\bigr)^{P} = \oF\bigl((J\alpha_{n})^{P}\bigr)
\end{math} by
Lemma
\ref{lemmaForMain:projectionOfLowerCoconeIsKleisliFAppliedToUpperCone}.

Hence the diagram
 (\ref{Diagram:TakenProjectionsOfInitialSequenceAndItsColimits})
is equal to the following one, showing the final sequence for
 $\oF $, its limit (the upper one) and that limit mapped by 
$\oF $ (the lower one) which is again a limit.
\begin{equation}\label{Diagram:FinalSequenceAndItsLimitsInKleisliCategory}
{    \vcenter{\xymatrix@C+1em@R-0.8em{
   {\text{In $\Kleisli{T}$}}
  &&&&
   A 
   \shifted{4em}{0em}{\text{(limit)}}
            \ar@<-1.5ex>@/_1pc/[dd];[]_-{J\alpha}^-{\cong}
  \\
   {\cdots}
            \ar[r];[]_{\oF ^{n-1}{\toTerm}}%_{\txt{Initial sequence\\ for $\oF $}}
  &
   {\oF}^{n}0
            \ar@/_1pc/[rrru];[]|{(J\alpha_{n})^{P}}
            \ar@/^1pc/[rrrd];[]^{\oF (J\alpha_{n-1})^{P}}
            \ar[r];[]
  &
   {\oF}^{n+1}0
            \ar@/_/[rru];[]|{(J\alpha_{n+1})^{P}}
            \ar@/^/[rrd];[]|{\oF (J\alpha_{n})^{P}}
            \ar[r];[]
  &
   {\cdots}
  &
  \\
  &&&&
   {\oF}A
   \shifted{4em}{0em}{\text{(limit)}}
            \ar@<-1.5ex>@/^1pc/[uu];[]_-{J\alpha^{-1}}
}}}
\end{equation}
%Here $\toTerm: X\to 0$ is the unique arrow to the final object $0$ in $\Kleisli{T}$.
By Proposition \ref{LemFinalCoalgebraViaFinalSequence}
we conclude that $J\alpha^{-1}$ is a final $\oF $-coalgebra.
\qed

In the remainder of this section the lemmas used in the above proof are
presented. We rely on the same assumptions as in Theorem~\ref{theorem:main}.
\begin{lem}\label{lemmaForMain:indeedInitialSequence}
 The $\omega$-chain in the diagram (\ref{Diagram:InitialSeqMappedByJ})
 is indeed the initial sequence for $\oF$. 
 That is, we have for each $n < \omega$,
 \begin{displaymath}
 JF^n \bigl(\,{\fromInit}^{\Sets}\,\bigr) = 
 \oF^n\bigl(\,{\fromInit}^{\Kleisli{T}}\,\bigr)
 \; : \;
 JF^{n}0 \longrightarrow JF^{n+1}0 \quad\text{in $\Kleisli{T}$,}
 \end{displaymath} 
 where ${\fromInit}^{\Sets}: 0\to F0$ in $\Sets$ and
 ${\fromInit}^{\Kleisli{T}}: 0\to F0$ in $\Kleisli{T}$ denote the unique maps.
\end{lem}
\proof
 By induction on $n$. For $n=0$ the two maps are equal due to the initiality
 of $J0=0$ in $\Kleisli{T}$. For the step case we use
 the commutativity $JF=\oF J$ of~(\ref{diagram:definitionOfLiftingOfF}). 
\qed

\begin{lem}\label{OInFinLem}
The empty set $0$ is both an initial and a final object in $\Kleisli{T}$.
\end{lem}
\noindent In particular, this implies that the object $T0$ is final in $\Sets$.
\proof
The functor $J: \Sets \to \Kleisli{T}$ preserves initial objects since
 it is a left adjoint. Therefore $0 = J0$ is initial in $\Kleisli{T}$. 
Finality follows essentially from the left-strictness assumption:
for each set $X$ there exists at least one arrow $X\to 0$ in $\Kleisli{T}$, for example
$\bot_{X,0}$.
%For an arbitrary set $X$, there always exists the bottom map $\bot_{X, 0}: X \to 0$ in $\Kleisli{T}$,
%which is the bottom in the poset $\Kleisli{T}(X, 0)$.
% i.e.\ the bottom arrow $\bot_{X, 0}: X \to 0$ in $\Kleisli{T}$. 
To show the uniqueness of such an arrow, 
take an arbitrary arrow $f: X \rightarrow 0$ in $\Kleisli{T}$. 
Recalling that the bottom map $\bot_{0, 0}: 0 \to 0$ is also the identity
 arrow in $\Kleisli{T}$ because of initiality,
% since, by initiality, there exists a unique map from 0 to $T0$. 
we obtain
$$f \;=\; \id \co f \;=\; \bot_{0, 0} \co f \;\stackrel{(*)}{=}\;
\bot_{X, 0} 
%\;\stackrel{(*)}{=}\; \bot_{0, 0} \co g \;=\; \id\co g
%\;=\;g
\enspace,
$$
where the compositions are taken in $\Kleisli{T}$ and the equality
 marked by $(*)$ holds by  left-strictness of composition.
%The second point holds because the right adjoint $K$ in the standard
% adjunction $J\dashv K$ preserves final objects.
\qed

\begin{lem}\label{lemmaForMain:arrowsInInitialSequenceAreEmbeddings}
 Each arrow $\oF ^{n}{\fromInit}$ in the initial sequence for
 $\oF $, as in the diagram (\ref{Diagram:InitialSeqInKleisliCat}),
 is an embedding. 
 Its corresponding projection is given by
\begin{displaymath}
  \bigl(\oF ^{n}{\fromInit}\bigr)^P =
 \oF ^{n}{\toTerm}
  \qquad\text{in}\quad
  \vcenter{\xymatrix@1@C+1em{
   {F^{n}0}
            \ar@{ >->}@/_/[r]_{\oF^{n}\fromInit}
  &
   {F^{n+1}0}
            \ar@{->>}@/_/[l]_{\oF^{n}\toTerm}
}}\enspace.
\end{displaymath}
%where ${\toTerm}$ denotes the unique arrow from $F0$ to 
%the final object $0$ in $\Kleisli{T}$ (cf. Lemma \ref{OInFinLem}). 
\end{lem}
\proof
We show that 
$(\oF ^n {\fromInit},\, \oF ^n {\toTerm})$ is an
 embedding-projection pair
 for all $n < \omega$. 
We have $\oF ^n{\toTerm} \co \oF^n{\fromInit}=\id$ because
$\toTerm\co\fromInit = \id$.
For the other half we have
\begin{align*}
%  \oF ^n(!) \co \oF ^n({\fromInit})
% &= 
%  \oF ^n(! \co {\fromInit})
% \\
% &\stackrel{(a)}{=}
%   \oF ^n(\id)
% = \id\enspace, \text{ and}
% \\
  \oF ^n {\fromInit} \co \oF ^n{\toTerm}
 &= 
  \oF ^n({\fromInit}\co {\toTerm})
  \\
 &=
  \oF ^n(\bot_{0,F0} \co {\toTerm})
 &&\text{initiality of $0$ in $\Kleisli{T}$}
 \\
  &=
   \oF ^n(\bot_{F0,F0})
 &&\text{composition is left-strict}
 \\
 &\sqsubseteq
  \oF ^n(\id)
 =\id
 &&\text{$\oF $ is locally monotone.}
 \tag*{$\qEd$}
 \end{align*}
% where the equality marked with $(a)$ holds since $0$ is initial in
% $\Kleisli{T}$, the equality marked with $(b)$ by the left-strictness of
% the composition, and the one marked with $(c)$ since $\oF $ is
% locally monotone and $\bot_{F0,F0}$ is a bottom element. 

\begin{lem}\label{lemmaForMain:projectionOfLowerCoconeIsKleisliFAppliedToUpperCone}
 We have 
 \begin{math}
 \bigl(\oF J\alpha_{n}\bigr)^{P} = \oF\bigl((J\alpha_{n})^{P}\bigr)
 \end{math}.
 Hence the lower cone in the diagram 
 (\ref{Diagram:TakenProjectionsOfInitialSequenceAndItsColimits})
 is the image of the upper cone under $\oF$.
\end{lem}
\begin{proof}
 It is easy to check that 
 $\bigl(\,\oF J\alpha_{n},\,
  \oF\bigl((J\alpha_{n})^{P}\bigr)\,\bigr)$ indeed form an
 embedding-projection pair.
 Therein we use the monotonicity of $\oF$'s action on arrows.
\end{proof}

\subsection{Simpler proof in
 \texorpdfstring{$\Kleisli{\pow}\cong\Rel$}{Kl(P)cong rel}}
\label{subsection:simplerProofInRel}
When $T=\pow$ we have the self-duality
\begin{displaymath}
 \Op\quad:\quad \Kleisli{\pow}^{\op}\stackrel{\cong}{\longrightarrow}\Kleisli{\pow}\enspace.
\end{displaymath}
This is
because of the following bijective correspondence between functions
\begin{displaymath}
        \vcenter{\infer={
        \xymatrix@C+1em@1{Y\ar[r]^-{f^{\lor}}& {\pow X}} \text{ in $\Sets$}
       }{
        \xymatrix@C+1em@1{X\ar[r]^-{f}&{\pow Y}} \text{ in $\Sets$} 
       }}
% \infer={Y\stackrel{f^{\lor}}{\longrightarrow} \pow X \quad \text{in $\Sets$}}
%        {X\stackrel{f}{\longrightarrow} \pow Y\quad\text{in $\Sets$}}      
\end{displaymath}
given by $f^{\lor}(y)=\{x\in X\mid y\in f(x)\}$. Recalling
$\Kleisli{\pow}\cong\Rel$, this mapping $f\mapsto f^{\lor}$ corresponds to taking the opposite
relation.

Due to this ``global'' duality
$\Kleisli{\pow}\cong\Kleisli{\pow}^{\op}$, the proof of
Theorem~\ref{theorem:main} is drastically simplified for $T=\pow$.
It essentially relies on the lifted self duality
$\Alg{\oF}\cong\Alg{\oF^{\op}}$, where  the latter is isomorphic to
$({\Coalg{\oF}})^{\op}$. We do not need here an order structure of
$\Kleisli{\pow}$ nor local monotonicity of $\oF$.
\begin{thm}\label{theorem:simplifiedMainForPowersetMonad}
 Let $F:\Sets\to\Sets$ be a functor which preserves weak pullbacks,
 and $\oF:\Kleisli{\pow}\to\Kleisli{\pow}$ be its lifting induced by
 relation lifting (Lemma~\ref{lemma:distributiveLawForPowersetMonad}).
 Then the initial $F$-algebra in $\Sets$ yields the final
 $\oF$-coalgebra in $\Kleisli{\pow}$.
\end{thm}
\proof
 We have the following situation because of the self-duality of $\Kleisli{\pow}$.
 \begin{displaymath}
 \vcenter{\xymatrix@C+4em@R=1.5em{
  {\Sets}
              \ar@/^/[r]^{J}
              \ar@/^/[r];[]^{K}
              \ar@{}[r]|{\bot}
              \ar@(ld,rd)[]_(.25){F}
 &
  {\Kleisli{\pow}} 
              \ar[r]^{\Op^{\op}}_{\cong}
%              \ar@/^/[r];[]^{\Op}
%             \ar@{}[r]|{\cong}
              \ar@(ld,rd)[]_(.25){\oF}
 &
  {\Kleisli{\pow}^{\op}}
              \ar@(ld,rd)[]_(.25){\oF^{\op}}
}}
\end{displaymath}
 The adjunction $J\dashv K$ and the isomorphism $\Op:
 \Kleisli{\pow}^{\op}\iso\Kleisli{\pow}$ lift to those between
 the categories of algebras.
 \begin{displaymath}
 \vcenter{\xymatrix@C+4em@R=1.5em{
  {\Alg{F}}
              \ar@/^/[r]^{J'}
              \ar@/^/[r];[]
              \ar@{}[r]|{\bot}
              \ar[d]
 &
  {\Alg{\oF}}
%              \ar@/^/[r]
%              \ar@/^/[r];[]
%              \ar@{}[r]|{\cong}
              \ar[r]^{(\Op')^{\op}}_{\cong}
              \ar[d]
 &
  {\Alg{\oF^{\op}}}
              \ar[r]^{\cong}
              \ar[d]
 &
  {(\Coalg{\oF})^{\op}}
 \\
  {\Sets}
              \ar@/^/[r]^{J}
              \ar@/^/[r];[]^{K}
              \ar@{}[r]|{\bot}
%              \ar@(ld,rd)[]_(.25){F}
 &
  {\Kleisli{\pow}} 
              \ar[r]^{\Op^{\op}}_{\cong}
%              \ar@/^/[r];[]^{\Op}
%             \ar@{}[r]|{\cong}
%              \ar@(ld,rd)[]_(.25){\oF}
 &
  {\Kleisli{\pow}^{\op}}
%              \ar@(ld,rd)[]_(.25){\oF^{\op}}
}}
\end{displaymath}
 Indeed, $J\dashv K$ lifts due to
 Proposition~\ref{proposition:initialAlgebraLiftsToKleisli};
 the lifted isomorphism $\Op':\Alg{\oF}\iso\Alg{\oF^{\op}}$ is because of 
 the following commutativity:
 \begin{equation}\label{diagram:liftedFAndSelfDualityOfRelAreCompatible}
\vcenter{  \xymatrix@R=1em@C+2em{
   {\Kleisli{\pow}^{\op}}
                 \ar[r]^{\Op}
                 \ar[d]_{\oF^{\op}}
  &
   {\Kleisli{\pow}}
                 \ar[d]^{\oF}               
  \\
   {\Kleisli{\pow}^{\op}}
                 \ar[r]_{\Op}
  &
   {\Kleisli{\pow}}             
}}
 \end{equation}
 which is because: $\oF R = \Relf(R)$
 (see~(\ref{equation:liftingOfFIsRelationLifting})); and  taking relation liftings is compatible with
 opposite relations (i.e.\ $\Relf(R^{\op})=(\Relf R)^{\op}$,
 see~\cite{HughesJ04}). Moreover the category $\Alg{\oF^{\op}}$ is
 obviously isomorphic to $(\Coalg{\oF})^{\op}$.

 Therefore the initial object in $\Alg{F}$ is carried to that in
 $(\Coalg{\oF})^{\op}$, hence the final object in $\Coalg{\oF}$.
\qed

For monads such as $T=\dist$ a ``global'' self-duality
$\Kleisli{T}\cong\Kleisli{T}^{\op}$ is not available. Instead, in the proof of
Theorem~\ref{theorem:main}, we exploit the ``partial'' duality
which holds between the colimit/limit of the initial/final sequence.

\subsection{Related work: axiomatic domain theory}\label{subsection:relatedWorkAxiomaticDomainTheory}
 The  initial algebra-final coalgebra coincidence is 
 heavily exploited in the field of \emph{axiomatic domain theory},
 e.g.\  in~\cite{Freyd91,Freyd92,Fiore96b,Simpson92}. 
 There, categories
 which have coinciding initial algebra and final coalgebra for each endofunctor
 are called 
 \emph{algebraically compact categories}. They draw special 
 attention as suitable ``categories of domains'' for denotational
 semantics of datatype construction. The relevance comes as follows.

 Let $\C$ be a ``category of domains.''  
 We think of an object of the category $\C$ as a type.  A
 ``recursive'' datatype constructor---a prototypical example is
 $(X,Y)\mapsto Y^{X}$---is presented as a bifunctor $G: \C^{\op}\times
 \C\to \C$.  Note the presence of both covariance and contravariance.
 We expect that such a category $\C$ has a canonical fixed point
 $\Fix G$ such that
 \begin{displaymath}
  G(\Fix G,\Fix G) \iso \Fix G\enspace,
 \end{displaymath} 
 which models the recursive type determined by the datatype constructor $G$.
 Freyd~\cite{Freyd91} showed that if $\C$ is algebraically compact, then we
 can construct such a fixed point as a suitable initial algebra; 
 moreover this fixed point is shown by Fiore~\cite{Fiore96b} to be 
 a canonical one in a suitable sense. The rough idea here is that the covariant part of $G$ is
 taken care of by an initial algebra; the contravariant part is by
 a final coalgebra; the initial algebra-final coalgebra coincidence yields a fixed point of overall $G$.
 
 Typical examples of algebraically compact categories are
 enriched over $\Cppo$ or one of its variants. This conforms the
 traditional use of the word ``domain'' for certain cpo's
 (e.g.\  in~\cite{AbramskyJ94}). 
 
 Although we utilize the initial algebra-final coalgebra coincidence
 result in $\KT$, we are not so much interested in algebraic compactness
 of $\KT$. %\footnote{Algebraic compactness of $\KT$
% is the main concern of~\cite{Simpson92}.} 
This is because 
 our motivation is different from that of axiomatic domain
 theory. In studying trace semantics for coalgebras, we need not deal with
 \emph{every} endofunctor on $\KT$, but only such an
 endofunctor $\oF$ which is a lifting of $F: \Sets\to\Sets$.

 In a different context of functional programming, the
 work~\cite{Pardo01} 
 also studies initial algebras and final coalgebras in a Kleisli category.
 The motivation there is to combine \emph{data types} and \emph{effects}.
 More specifically, an initial algebra and a final coalgebra support the \emph{fold} and
 the \emph{unfold} operators, respectively, used in recursive
 programs over datatypes. A computational effect is presented as 
 a monad, and its Kleisli category is the category of effectful
 computations.
 
 The difference between~\cite{Pardo01} and the current work is as follows.
 In~\cite{Pardo01}, the original category of pure functions is already
 algebraically compact;
 the paper studies the conditions for the algebraic compactness
 to be carried over to Kleisli categories. 
 In contrast, in the current work, it is a monad---with a suitable
 order structure, embodying the essence of ``branching''---which 
 yields the initial algebra-final coalgebra coincidence on a Kleisli
 category; the coincidence is not present in the original category $\Sets$.

 \subsubsection{Local continuity vs.\ local monotonicity}
 \label{subsubsection:localContiVSLocalMonotonicity}

 In axiomatic domain theory, $\Cppo$-enriched categories are said to be
 algebraically compact because, ``in a  2-category setting''~\cite{Freyd92},
 every endofunctor has an initial algebra and a final coalgebra.
 Concretely this means: ``every locally continuous functor.''

 In this spirit, we could have made a stronger assumption of $\oF$'s 
 local continuity in Theorem~\ref{theorem:main} instead of local 
 monotonicity. If we do so, in fact, the proof of
 Theorem~\ref{theorem:main} becomes much simpler: the following proposition
 (Lemma in~\cite[p.98]{Freyd92}) immediately yields the
 initial algebra-final coalgebra coincidence for a locally continuous $\oF$.
 
 \begin{prop}[{\cite{Freyd92}}]\label{proposition:initialAlgFinalCoalgCoincidenceForLocallyContinuousFunctor}
  Let $\D$ be a $\Cppo$-enriched category whose composition is
  left-strict,
  and $G:\D\to\D$ be a locally
  continuous endofunctor.  An initial algebra $\beta: GB\iso B$,
  if it exists, yields a final coalgebra $\beta^{-1}: B\iso GB$.
 \end{prop}
 \proof
  Given a coalgebra $d: Y\to GY$, the function 
         \begin{displaymath}
%	\vcenter{\xymatrix@R=0em@C+5em{
%         {\D(Y,B)}
%                   \ar[r]^{\Phi}
%        &
%         {\D(Y,B)}
%        \\
%         {f}
%                   \ar@{|->}[r]
%        &
%         {\beta\co Gf\co d}
%	}}
	 \Phi\;:\; \D(Y,B)\longrightarrow \D(Y,B)\enspace, \qquad
          f \longmapsto \beta\co Gf\co d
        \end{displaymath}
  is continuous due to the local continuity of $G$. Hence
  it has the least fixed point $\bigsqcup_{n<\omega} \Phi^{n}(\bot)$;
  this proves existence of a morphism from $d$ to $\beta^{-1}$.
     \begin{displaymath}
 \vcenter{
 \xymatrix@C+5em@R=1em{
  {GY}
                \ar@{->}[r]
 &
  {GB}
 \\
  {Y}
                \ar[u]^{d}
                \ar@{->}[r]
 &
  {B}
                \ar[u]_{\beta^{-1}}^{\cong}
}
}
 \end{displaymath}
 
 Now we shall show its uniqueness. Assume that $g: Y\to B$
 is a morphism of coalgebras as above, that is, $\Phi(g) = g$. Similarly to $\Phi$, we define a
 function $\Psi: \D(B,B)\to\D(B,B)$ as the one which carries $h: B\to B$
 to $\beta\co Gh\co \beta^{-1}$. We have
  \begin{align*}
   \textstyle\bigsqcup_{n} \Phi^{n}(\bot)
  &=
   \textstyle\bigsqcup_{n} \Phi^{n}(B\stackrel{g}{\rightarrow} B \stackrel{\bot}{\rightarrow} Y)
  &&
   \text{composition is left-strict, so $\bot\co g = \bot$}
  \\
  &=
   \textstyle\bigsqcup_{n} \bigl(\Psi^{n}(\bot)\co \Phi^{n}(g)\bigr)
  &&
   \text{$\Phi^{n}(\bot\co g)=\Psi^{n}(\bot)\co \Phi^{n}(g)$, by induction}
  \\
  &=
   \bigl(\textstyle\bigsqcup_{n} \Psi^{n}(\bot)\bigr) \co \bigl(\textstyle\bigsqcup_{n} \Phi^{n}(g)\bigr)
  &&
   \text{composition is continuous}
  \\
  &=
   \textstyle\bigsqcup_{n} \Phi^{n}(g)
  &&
   \text{$\textstyle\bigsqcup_{n}
 \Psi^{n}(\bot)=\id$, ($*$)}
  \\
  &=
   g
  &&
   \text{$\Phi(g)=g$ by assumption.}
  \end{align*}
 Here ($*$) holds because $\textstyle\bigsqcup_{n}
 \Psi^{n}(\bot)$, being a fixed point for $\Psi$, is the unique morphism
 of algebras from $\beta$ to $\beta$.
 This shows that the morphism $g$ must be the least fixed point of $\Phi$.
 \qed

  For our main Theorem~\ref{theorem:main} we can do with only local
  monotonicity of the lifted functor $\oF$, by taking a closer look at
  the initial/final sequences.  However at this stage it is not clear
  how much we gain from this generality: up to now we have not found an
  example where the functor $\oF$ is only locally monotone (and not
  locally continuous).

\section{Finite trace semantics via
 coinduction}\label{section:traceSemanticsViaCoinduction}
 In this section we shall further illustrate the observation that 
 the principle of coinduction, when employed in $\Kleisli{T}$,
 captures trace semantics of state-based systems.
 As we have shown in the previous section, an initial algebra in $\Sets$
 constitutes the semantic domain, i.e.\ is a final coalgebra in
 $\Kleisli{T}$. Viewing an initial algebra as the set of
 well-founded terms (such as finite words or finite-depth parse
 trees), this fact means that the ``trace semantics'' induced
 by coinduction is inevitably \emph{finite}, in the sense 
 that it captures only finite  behavior.
 Here we will elaborate
 on this finiteness issue as well.

\subsection{Trace semantics by coinduction}
\label{subsection:traceSemByCoinduction}
 As we have seen in Section~\ref{subsection:liftingFunctorsByDistrLaws}
 various types of state-based systems allow their presentation 
 as coalgebras $X\to \oF X$ in a Kleisli category $\Kleisli{T}$. For example,
 \begin{enumerate}[$\bullet$]
  \item LTSs with explicit termination, with $T=\pow$ and
	$F=1+\Sigma\times\place$;
  \item probabilistic LTSs (also called \emph{generative probabilistic transition
        systems} in~\cite{GSS95:ic,Sok05}) with explicit termination,
        with $T=\dist$ and
	$F=1+\Sigma\times\place$;
  \item context-free grammars with $T=\pow$ and $F=(\Sigma+\place)^{*}$.
 \end{enumerate}
 The main observation  underlying this work is the following.
 If we instantiate the
 parameters 
\begin{center}
  $T$
 for branching type \qquad
 and\qquad $F$ for transition type 
\end{center}
 in the coinduction
 diagram
   \begin{equation}\label{diagram:coinductionInKleisliCategory}
 \vcenter{
 \xymatrix@C+5em@R=1em{
  {\oF X}
                \ar@{-->}[r]^{\oF (\trace{c})} 
 &
  {\oF A}
 \\
  {X}
                \ar[u]^{c}
                \ar@{-->}[r]_{\trace{c}}
 &
  {A}
                \ar[u]_{J\alpha^{-1}}^{\cong}
}
} \qquad
  \text{in $\KT$}
 \end{equation}
 with one of
 the above choices, then the commutativity of the diagram is equivalent to
 the corresponding (conventional) definition of trace semantics in
 Section~\ref{subsection:introVariousTraceSemantics}. 
 Therefore we claim that the diagram
 (\ref{diagram:coinductionInKleisliCategory})
 is the mathematical principle underlying  various ``trace semantics,''
 no matter if it is ``trace set'' (non-deterministic) or
 ``trace distribution'' (probabilistic). 
 \begin{cor}[Trace semantics for coalgebras]
 \label{corollary:traceSemanticsForCoalgebras}
  Assume that $T$ and $F$ are such as in Theorem~\ref{theorem:main},
  and $\alpha: FA\iso A$ is an initial $F$-algebra in $\Sets$.
  Given a coalgebra $c: X\to TFX$ in $\Sets$, we can assign a function 
\begin{displaymath}
   \trace{c}\; :\; X\longrightarrow TA \qquad \text{in $\Sets$}
\end{displaymath}  
which is, as an arrow $X\to A$ in $\Kleisli{T}$,
  the unique one making the diagram
  (\ref{diagram:coinductionInKleisliCategory})
  commute. We shall call this function $\trace{c}$ the \textbf{(finite) trace
  semantics} for the coalgebra $c$. \qed
 \end{cor}

 \begin{exa}
 As further illustration
 we give details for the  choice of parameters 
 $T=\pow$
 and $F=1+\Sigma\times\place$. This is the suitable choice to deal with
 the first system in~(\ref{diagram:firstExampleOfSystems}).

 Now the coinduction
 diagram looks as follows. Recall that an initial $F$-algebra 
 is carried by the set $\Sigma^{*}$ of finite words.
   \begin{equation}\label{diagram:coinductionForLTS}
 \vcenter{
 \xymatrix@C+5em@R=1em{
  {1+\Sigma\times X}
                \ar@{-->}[r]^{1+\Sigma\times\trace{c}} 
 &
  {1+\Sigma\times \Sigma^{*}}
 \\
  {X}
                \ar[u]^{c}
                \ar@{-->}[r]_{\trace{c}}
 &
  {\Sigma^{*}}
                \ar[u]_{J([\mathsf{nil},\mathsf{cons}])^{-1}}^{\cong}
}
} \qquad
  \text{in $\Kleisli{\pow}$}
 \end{equation}
 It assigns, to a system $c$, a function $\trace{c}: X\to
 \pow(\Sigma^{*})$ which carries a state $x\in X$ to the set of
 finite words on $\Sigma$ which can possibly arise as an execution
 ``trace'' of $c$ starting from $x$.
 The commutativity states equality of two arrows $X\rightrightarrows
 1+\Sigma\times\Sigma^{*}$ in $\Kleisli{\pow}$,
 that is, functions $X\rightrightarrows \pow
 (1+\Sigma\times\Sigma^{*})$. 
 Let us denote these functions by
 \begin{displaymath}
  u = (1+\Sigma\times\trace{c})\co c
  \quad\text{(up, then right),}
  \qquad
  v = {J([\mathsf{nil},\mathsf{cons}])^{-1}} \co \trace{c}
  \quad\text{(right, then up).}
 \end{displaymath}
For each $x\in X$,
 the following conditions---derived straightforwardly by
 definition of composition of $\Kleisli{\pow}$, lifting of the functor
 $1+\Sigma\times\place$, etc.---specify $u$ and $v$'s value
 at $x$, as a subset of $1+\Sigma\times \Sigma^{*}$.
 \begin{displaymath}
%  \begin{array}{rclcrcl}
%   \checkmark\in u(x)
%  &
%   \Longleftrightarrow 
%  &
%   \checkmark\in c(x) 
%  &
%   \qquad
%  &
%   \checkmark\in v(x)
%  &
%   \Longleftrightarrow 
%  &
%   \tuple{}\in\trace{c}(x)
%  \\
%   (a,\sigma)\in u(x)
%  &
%   \Longleftrightarrow 
%  &
%\multicolumn{3}{l}{   \exists x'\in X.\;\bigl(
%    \,(a,x')\in c(x) \,\land\, \sigma\in \trace{c}(x')\,\bigr)
%}  \\
%   (a,\sigma)\in v(x)
%  &
%   \Longleftrightarrow 
%  &
%   a\cdot \sigma\in\trace{c}(x)
%  \end{array}
  \begin{array}{rcl}
   \checkmark\in u(x)
  &
   \Longleftrightarrow 
  &
   \checkmark\in c(x) 
  \\
   (a,\sigma)\in u(x)
  &
   \Longleftrightarrow 
  &
   {   \exists x'\in X.\;\bigl(
    \,(a,x')\in c(x) \,\land\, \sigma\in \trace{c}(x')\,\bigr)
}
  \\
   \checkmark\in v(x)
  &
   \Longleftrightarrow 
  &
   \tuple{}\in\trace{c}(x)
  \\
   (a,\sigma)\in v(x)
  &
   \Longleftrightarrow 
  &
   a\cdot \sigma\in\trace{c}(x)
  \end{array}
 \end{displaymath}
 Commutativity of  (\ref{diagram:coinductionForLTS}) amounts to $u=v$; this 
 gives the condition (\ref{equation:conventionalDefOfTraceSetForLTS}).

 From a different point of view we can also express that as follows: finality of the
 coalgebra $\Sigma^{*}\iso 1+\Sigma\times\Sigma^{*}$ in
 (\ref{diagram:coinductionForLTS}) ensures that the conventional
 recursive definition (\ref{equation:conventionalDefOfTraceSetForLTS}) 
 uniquely determines a function $\trace{c}:X\to\pow(\Sigma^{*})$. Hence
 $\trace{c}$ is \emph{well-defined}. 

   An easy consequence of the recursive definition
  (\ref{equation:conventionalDefOfTraceSetForLTS})
 is
 \begin{displaymath}
  a_{1}\dotsc a_{n}\in\trace{c}(x)
  \quad\Longleftrightarrow\quad
  \exists x_{1},\dotsc, x_{n}\in X.\quad
  x\stackrel{a_{1}}{\to}\cdots\stackrel{a_{n}}{\to}x_{n}
  \to \checkmark\enspace.
 \end{displaymath}
 Therefore every trace $a_{1}\dotsc a_{n}\in\trace{c}(x)$ has
  termination $\checkmark$ implicit at its tail.
 In particular, the set $\trace{c}(x)$ is not necessarily prefix-closed:
  $a_{1}\dotsc a_{n}a_{n+1}\dotsc a_{n+m}\in\trace{c}(x)$ does not imply
  $a_{1}\dotsc a_{n}\in\trace{c}(x)$.
\end{exa}

\begin{exa}
\label{example:traceSemForLiftMonadAndLTSWithTerm}
Let us take $T=\lift$ (the lift monad) and $F=1+\Sigma\times\place$. In this
 case a coalgebra $X\stackrel{c}{\to} \lift(1+\Sigma\times X)$ in $\Sets$ is a system which can
 \begin{enumerate}[$\bullet$]
  \item get into a deadlock ($c(x)=\bot$ where $\lift=\{\bot\}+\place$),
  \item successfully terminate ($c(x)=\checkmark$ where $F=\{\checkmark\}+\Sigma\times\place$), or
  \item output a letter from $\Sigma$ and move to the next state ($c(x)=(a,x')$).
 \end{enumerate}
 By examining trace semantics for such systems,
 we shall formally put the difference between the computational meanings
 of the two elements, $\bot$ and $\checkmark$.
 
 The coinduction diagram (\ref{diagram:coinductionInKleisliCategory})
 instantiates to the same diagram as (\ref{diagram:coinductionForLTS}),
 but now in the category $\Kleisli{\lift}$. Easy calculation shows that
 its commutativity amounts to the following condition. The function
 \begin{displaymath}
\vcenter{  
\xymatrix@1@C+2em{
  {X} 
     \ar[r]^-{\trace{c}}
 &
  {\lift (\Sigma^{*})
    =
  \{\bot\} +\Sigma^{*}}
}}  \qquad\text{in $\Sets$}
 \end{displaymath}
 satisfies, for each $x\in X$,
\begin{equation}\label{equation:conventionalDefForTraceSemForLTSWithDeadlock}
   \begin{array}{lcl}
   \trace{c}(x) = \tuple{}
  &
   \Longleftrightarrow 
  &
   c(x) = \checkmark \enspace,
  \\
   \trace{c}(x) = a\cdot \sigma
  &
   \Longleftrightarrow 
  &
   {   \exists x'\in X.\;\bigl(
    \,c(x) = (a,x') \,\land\, \trace{c}(x') = \sigma \,\bigr)\enspace,
 }
  \\
   \trace{c}(x) = \bot
  &
   \Longleftrightarrow 
  &
   c(x) = \bot \quad\text{or}\quad \exists x'\in X.\;\bigl(\,c(x)=(a,x')\,\land\, \trace{c}(x')=\bot\,\bigr)\enspace.
  \end{array}
\end{equation}
 Here $\sigma\in\Sigma^{*}$ is a word in $\Sigma$. 

 For the systems under consideration, we can think of three
 different kinds of possible executions.
\begin{enumerate}[$\bullet$]
 \item An execution eventually hitting $\checkmark$, that is, 
 $x\stackrel{a_{1}}{\to}\cdots\stackrel{a_{n}}{\to}x_{n}\to\checkmark$.
 By the condition
       (\ref{equation:conventionalDefForTraceSemForLTSWithDeadlock})
 it yields a word 
\begin{math}
  \trace{c}(x)= a_{1}\dotsc a_{n}
\end{math} 
as its trace.
 \item An execution eventually hitting $\bot$, that is, 
 $x\stackrel{a_{1}}{\to}\cdots\stackrel{a_{n}}{\to}x_{n}\to\bot$.
 By the third line of
       (\ref{equation:conventionalDefForTraceSemForLTSWithDeadlock})
       we see that $\trace{c}(x_{n})=\bot$; moreover
       $\trace{c}(x_{n-1})=\cdots=\trace{c}(x)=\bot$. It properly reflects our
       intuition that a state $x$ that eventually goes into deadlock does
       not yield a finite (or terminating) trace.
 \item An execution not hitting $\checkmark$ nor $\bot$, that is,
 $x\stackrel{a_{1}}{\to}
  x_{1}\stackrel{a_{2}}{\to}
  \cdots$. In this case, the only possible solution of the ``recursive
       equation''
       (\ref{equation:conventionalDefForTraceSemForLTSWithDeadlock})
       is $\trace{c}(x)=\trace{c}(x_{1})=\cdots=\bot$. The intuition
       here is: a state leading to \emph{livelock} does not yield a finite trace.
\end{enumerate}
\end{exa}

%#BEGIN
\subsection{Infinite traces}
\label{subsection:infiniteTraces}
The trace semantics obtained via coinduction
(Corollary~\ref{corollary:traceSemanticsForCoalgebras})
assigns, to each state $x\in X$, ``a set of'' (if $T=\pow$) 
or ``a distribution over'' (if $T=\dist$) elements of the initial algebra $A$.
Elements of $A$ are thought of as possible linear behavior of
the system determined by the transition type (i.e.\ the functor $F$).

Now the intuition is that an initial $F$-algebra $A$ consists of the 
well-founded (or finite-depth) terms and a final $F$-coalgebra $Z$
consists of the
possibly non-well-founded (or infinite-depth) terms.
For example,
\begin{enumerate}[$\bullet$]
 \item for $F=1+\Sigma\times \place$, $A=\Sigma^{*}$ consists of all
       the finite words, and
       $Z=\Sigma^{\infty}=\Sigma^{*}+\Sigma^{\omega}$ is augmented
       with \emph{streams}, i.e.\ infinite words;
 \item for $F=(\Sigma+\place)^{*}$, $A$ is the set of finite-depth
       \emph{skeletal parse trees} (see~\cite{HasuoJ05b}), and $Z$
       additionally contains  infinite-depth ones;
 \item for $F=\Sigma\times \place$ which models LTSs \emph{without}
       explicit termination, $A=0$ and $Z=\Sigma^{\omega}$.
\end{enumerate}
Therefore our trace semantics $X\to TA$ only takes account of finite, well-founded
linear-time behavior but not infinite ones. This is why the trace set
(\ref{equation:introExampleTraceSet}) does not contain $ab^{\omega}$;
and also why we have been talking about LTSs \emph{with} explicit
termination---otherwise the finite trace semantics is always empty.

Designing a coalgebraic framework to capture possibly infinite trace semantics 
is the main aim of~\cite{Jacobs04c}. The work is done exclusively in a
non-deterministic setting and the main result reads as follows.
\begin{thm}[Possibly infinite trace semantics for
 coalgebras,~\cite{Jacobs04c}]
 \label{theorem:coalgebraicCharOfPossiblyInfiniteTraces}
 Let $F$ be a shapely functor on $\Sets$, and $\zeta: Z\iso FZ$ be a
 final coalgebra in $\Sets$. The coalgebra 
 \begin{displaymath}
  J\zeta \;:\; Z\longrightarrow \oF Z \qquad \text{in $\Kleisli{\pow}$}
 \end{displaymath}
 is weakly final: that is, given a coalgebra $c: X\to\oF X$,
 there is a morphism from $c$ to $J\zeta$ but the morphism is not necessarily unique. 
   \begin{equation}\label{diagram:weaklyFinalCoalgebraInKleisli}
 \vcenter{
 \xymatrix@C+5em@R=1em{
  {\oF X}
                \ar@{~>}[r]^{\oF (\inftytrace{c})} 
 &
  {\oF Z}
 \\
  {X}
                \ar[u]^{c}
                \ar@{~>}[r]_{\inftytrace{c}}
 &
  {Z}
                \ar[u]_{J\zeta}^{\cong}
}
} \qquad
  \text{in $\Kleisli{\pow}$}
 \end{equation}
 
 Still there  
 is a canonical choice $\inftytrace{c}$
% \begin{displaymath}
%  \inftytrace{c}\;:\; X\longrightarrow \pow Z \qquad\text{in $\Sets$}
% \end{displaymath}
 among such morphisms, namely the one which is
 maximal with respect to the inclusion order.
 We shall call the function $\inftytrace{c}: X\to \pow Z$ the
 \textbf{possibly-infinite trace semantics} for $c$.
 \qed
\end{thm}
 Note here that, when we take $F=1+\Sigma\times \place$ and $T=\pow$
 (the choice for LTSs with termination), 
 commutativity of (\ref{diagram:weaklyFinalCoalgebraInKleisli}) 
 boils down to exactly the same conditions as
 (\ref{equation:conventionalDefOfTraceSetForLTS}):
  \begin{equation}\label{equation:notValidDefForInfiniteTraceForLTS}
\begin{array}{rclcrcl}
    \tuple{} \in \inftytrace{c}(x) 
  & \,\Longleftrightarrow\,
  & x\to\checkmark ,
%  & \text{where $\tuple{}$ is the empty word;}
  &
   \quad
  &
   a\cdot \sigma \in \inftytrace{c}(x)
  & \,\Longleftrightarrow\,
  & \exists y.\; (\, x\stackrel{a}{\to} y \;\land\; \sigma\in \inftytrace{c}(y)\,) .
% \multicolumn{4}{r}{\text{where $\sigma=a_{1}a_{2}\dotsc a_{n}$ is a word on actions.}}
\end{array}  
\end{equation}
 Weak finality of $\Sigma^{\infty} \iso
 1+\Sigma\times\Sigma^{\infty}$ (corresponding to $Z\iso \oF Z$ in
 (\ref{diagram:weaklyFinalCoalgebraInKleisli}))
 means the following. The recursive definition
 (\ref{equation:notValidDefForInfiniteTraceForLTS})---although it looks valid at the first
 sight---does
 \emph{not} uniquely determine the infinite trace map $\inftytrace{c}: X\to
 \pow(\Sigma^{\infty})$.
 Instead, the map $\inftytrace{c}$ is the maximal one among those which satisfy
 (\ref{equation:notValidDefForInfiniteTraceForLTS}). 

 As an example take the first system in
 (\ref{diagram:firstExampleOfSystems}).  We expect its possibly-infinite
 trace map $X\to \pow (\Sigma^{\infty})$ to be such that $ x \mapsto
 ab^{*} + ab^{\omega} $ and $ y \mapsto b^{*} + b^{\omega} $.  Indeed
 this satisfies (\ref{equation:notValidDefForInfiniteTraceForLTS}) and
 is moreover the maximal. However, the function $x\mapsto ab^{*}$ and
 $y\mapsto b^{*}$---this is actually the finite trace $X\to\pow
 (\Sigma^{*})$ embedded along $\Sigma^{*}\hookrightarrow
 \Sigma^{\infty}$---also satisfies
 (\ref{equation:notValidDefForInfiniteTraceForLTS}).  In fact,
 \cite[Section 5]{HasuoJ05b}  
 shows a general fact that such an
 embedding of the finite trace map is the minimal one among those morphisms which
 make the diagram (\ref{diagram:weaklyFinalCoalgebraInKleisli}) commute.

  The coalgebraic characterization
 (Theorem~\ref{theorem:coalgebraicCharOfPossiblyInfiniteTraces}) of
 possibly-infinite trace semantics is not yet fully developed. In
 particular the current proof of
 Theorem~\ref{theorem:coalgebraicCharOfPossiblyInfiniteTraces} (in~\cite{Jacobs04c}) is fairly
 concrete and a categorical principle behind it is less clear than the
 one behind finite traces. Consequently  the result's applicability is limited:
 we do not know whether the result
 holds in a probabilistic setting; or whether it holds for any
 weak-pullback-preserving functor $F$.

%#END

\section{Trace semantics as testing
 equivalence}\label{section:traceSemanticsAsTestingEquivalence}
In this section we will observe that, in a non-deterministic setting,
the coalgebraic finite trace
semantics (i.e.\ coinduction in $\Kleisli{\pow}$) gives rise to a
canonical \emph{testing situation} in which a test is an element 
of the initial $F$-algebra $A$ in $\Sets$. Here $F$ specifies the
transition type, just as before.
The notion of testing situations (Definition~\ref{definition:testingSituations})
and its variants have attracted many authors' attention in the 
context of coalgebraic modal logic; our aim here is to demonstrate
genericity and pervasiveness of the notion of testing situations by presenting 
an example which is not much like modal logic (that is, propositional
logic plus modality).

In Section~\ref{subsection:testingSituations} we introduce the notion
of testing situations and investigate some of their general properties. 
Our main concern there is the comparison between two process
equivalences, namely \emph{testing equivalence} and \emph{equivalence
modulo final coalgebra semantics}. We present the equivalences
categorically as suitable kernel pairs; this makes the arguments simple and clean.
In
Section~\ref{subsection:testingForTraceSemantics} we present the
canonical testing situation for trace semantics. Moreover we show that it is
\emph{expressive}: the testing captures final coalgebra semantics, which
is now trace semantics.
% and prove its
%\emph{expressivity}, i.e.\ coincidence of the two process equivalences.
%Since the final coalgebra semantics in $\Kleisli{\pow}$ is trace
%semantics, this means that the canonical testing situation
%\emph{captures} trace semantics.

\subsection{Testing situations}\label{subsection:testingSituations}
Recent
studies~\cite{KupkeKP04,BonsangueK05,BonsangueK06,Kurz06sigact,PavlovicMW06,Klin07}
on coalgebra and modal logic
have identified (variants of) the
following categorical situation as the essential underlying structure.
Following~\cite{PavlovicMW06}, we prefer using a more general term
``testing'': it subsumes ``modal logic'' in the following sense.
We learn properties
of a system through pass or failure of \emph{tests}; modal logic
constitutes a special case where tests are modal formulas.

\begin{defi}\label{definition:testingSituations}
 A \emph{testing situation} is the following situation 
of a contravariant adjunction $S^{\op}\dashv P$ and two endofunctors
$F,M$
 \begin{equation}\label{diagram:testingSituation}
 \xymatrix@1@C+3em{
  {\C^{\op}}
              \ar@(lu,ld)[]_{F^{\op}}
              \ar@/^.7em/[r]^{P}
              \ar@{}[r]|{\top}
 &
  {\A}
              \ar@(rd,ru)[]_{M}
              \ar@/^.7em/[l]^{S^{\op}}
 }
 \end{equation}
 plus a ``denotation'' natural transformation $\delta: MP\Rightarrow PF^{\op}: \C^{\op}\to \A$,
 which consists of arrows
 $MPX\stackrel{\delta_{X}}{\longrightarrow} PFX$ in $\A$.
%such that
%  the functor $P$ have a lifting $\hat{P}:\Coalg{F}^{\op}\to \Alg{M}$.
% \begin{displaymath}
%  \xymatrix@R=1em@C+1em{
%   {\Coalg{F}}
%                    \ar[r]^{(\hat{P})^{\op}}
%                    \ar[d]
%  &
%   {\Coalg{M^{\op}}}
%                    \ar[r]^{\cong}
%                    \ar[d]
%  &
%   {\Alg{M}^{\op}}
%                    \ar[ld]
%  \\
%   {\C}
%                    \ar[r]_{P^{\op}}
%  &
%   {\A^{\op}}
% }
% \end{displaymath}
\end{defi}
\noindent Note that the denotation $\delta$ is a parameter: the same
``syntax for tests''
$M:\A\to\A$ can have different interpretations with different $\delta$.

%We shall compare two process semantics, namely \emph{testing
%equivalence}---which arises naturally from the concept of testing---and
%final coalgebra semantics.
The requirements in Definition~\ref{definition:testingSituations} are the same as
in~\cite{PavlovicMW06,Klin07}. They are what we need to compare two
process semantics, namely  \emph{testing
equivalence}---which arises naturally from the concept of testing---and
final coalgebra semantics.\footnote{In fact we can be even more liberal: existence of a denotation $\delta$ can be
replaced by existence of a lifting $\hat{P}:\Coalg{F}^{\op}\to \Alg{M}$
of $P$. The results in this section nevertheless hold in that case.
The latter condition (there is a lifting $\hat{P}$) is strictly weaker than
the former
(there is a natural transformation $\delta$): obviously
$\delta$ induces $\hat{P}$ but not the other way round. Let
$\C=\omega^{\op}, \A=\omega, P=\id, F=(1+\place)^{\op}$ and $M=2+\place$.
Then both $\Coalg{F}$ and $\Alg{M}$ are the empty category hence
$P$ has the trivial lifting. However there is no natural transformation
$MPX\to PF^{\op}X$.}
We shall explain each ingredient's role, using the well-established terminology
of modal logic.
\begin{enumerate}[$\bullet$]
 \item The endofunctor $F:\C\to\C$
       makes $\Coalg{F}$ the category of ``systems,'' or ``Kripke
       models'' in modal logic.
 \item The category $\A$---typical examples being $\mathbf{Bool}$ of
       Boolean algebras or $\mathbf{Heyt}$ of Heyting algebras---is
       that of ``propositional logic.''
       The functor $M$ specifies ``modality'': modal
       operators and axioms. 
       Then $\Alg{M}$ is the category of ``modal algebras''; the initial
       $M$-algebra $ML\iso L$ is a ``modal logic'' consisting of modal formulas, modulo
       logical equivalence.
 \item The denotation $\delta$ specifies how the modality $M$ is interpreted
       via  transitions of type $F$. This allows to give ``Kripke
       semantics''
       for the modal logic: given a coalgebra (or a ``Kripke model'') $c: X\to FX$, 
       interpretation $\sem{\place}_{c}$
       of modal formulas therein is given by the following induction.
       \begin{equation}\label{diagram:interpretationOfModalFormulaViaInduction}
	\vcenter{\xymatrix@R=1em@C+3em{
          {ML}
                    \ar[dd]^{\cong}_{\scriptstyle\text{initial}}
                    \ar@{-->}[r]
         &
          {MPX}
                    \ar[d]^{\delta_{X}}
         \\
         &
          {PFX}    
                    \ar[d]^{Pc}
         \\
          {L}
                    \ar@{-->}[r]_{\sem{\place}_{c}}
         &
          {PX}                    
        }}
       \end{equation}       
 \item Why a right adjoint $S$ of $P^{\op}$? It allows us, via
       transposition, 
       to assign a modal ``theory'' to each state of a Kripke model.
       \begin{equation}\label{diagram:theoryAsTranspositionOfInterpretation}
       \vcenter{\infer=[(S^{\op}\dashv P)]{
%	        X\stackrel{\theory_{c}}{\longrightarrow} SL
                \xymatrix@1@C+3em{{X} \ar[r]^{\theory_{c}} & {SL}} \quad \text{in $\C$}
	        }
               {
%	        L\stackrel{\sem{\place}_{c}}{\longrightarrow} PX
                \xymatrix@1@C+3em{{L} \ar[r]^{\sem{\place}_{c}} & {PX}}\quad \text{in $\A$}
}}
       \end{equation}
       The theory $\theory_{c}(x)$ associated with a state $x$ contains
       precisely the modal formulas that hold at $x$.
\end{enumerate}
Following the above intuition, we define the categorical notion
of testing equivalence---two states are testing-equivalent
if they have the same modal theory.
\begin{defi}\label{definition:TestEq}
 Assume that we have a testing situation
 (\ref{diagram:testingSituation}), and
 that
 $\C$ has finite limits. 
On 
 a coalgebra $X\stackrel{c}{\to} FX$, the \emph{testing equivalence} $\TestEq_{c}$  is
 the kernel pair of the theory map $\theory_{c}$ defined by
 (\ref{diagram:interpretationOfModalFormulaViaInduction})
 and (\ref{diagram:theoryAsTranspositionOfInterpretation}).
 Equivalently, 
 \begin{equation}\label{diagram:defOfTestEq}
  \vcenter{\xymatrix@1@C+1em{
   {\TestEq_{c}}
                      \ar@{ >->}[r]^{\scriptstyle\tuple{p_{1},p_{2}}}
  &
   {X\times X}
                      \ar@<.2em>[r]^-{\scriptstyle\theory_{c}\co \pi_{1}}
                      \ar@<-.2em>[r]_-{\scriptstyle\theory_{c}\co \pi_{2}}
  &
   {SL}
}}
 \end{equation}
 is an equalizer.
\end{defi}

Similarly, we introduce the categorical notion of ``equivalence modulo
final coalgebra semantics''; we shall call it \emph{FCS-equivalence} for short.
\begin{defi}\label{definition:FCSEq}
 Assume that there is a final $F$-coalgebra $\zeta: Z\iso FZ$, and that
$\C$ has finite limits. 
On a coalgebra $X\stackrel{c}{\to}FX$, 
the \emph{FCS-equivalence} $\FCSEq_{c}$ is the kernel pair of
the unique map $\beh_{c}:X\to Z$ induced by finality. Equivalently,
 \begin{equation}\label{diagram:defOfFCSEq}
  \vcenter{\xymatrix@1@C+1em{
   {\FCSEq_{c}}
                      \ar@{ >->}[r]^{\scriptstyle\tuple{q_{1},q_{2}}}
  &
   {X\times X}
                      \ar@<.2em>[r]^-{\scriptstyle\beh_{c}\co \pi_{1}}
                      \ar@<-.2em>[r]_-{\scriptstyle\beh_{c}\co \pi_{2}}
  &
   {Z}
}}
 \end{equation}
 is an equalizer.
\end{defi}

It is easily seen that the two ``relations'' $\TestEq_{c}$ and $\FCSEq_{c}$
on $X$ are \emph{equivalence
relations} in the sense of~\cite[Section~1.3]{Jacobs99a}. That is, they satisfy the
reflexivity, symmetry, and transitivity conditions when the conditions
are suitably formulated in categorical terms.

\auxproof{
 \begin{rem}
 We present a process equivalence as a \emph{kernel pair}
 $R\rightrightarrows X$ in $\C$.  In~\cite{PavlovicMW06}, instead, a factorization
 structure $(\mathcal{E},\mathcal{M})$ on $\C$ is assumed and an equivalence $R$ on $X$ is presented
 as an $\mathcal{E}$-arrow $X\epi X/R$. The intuition is that this is the quotient
 map for $R$.  It is often the case that 
 these two presentations of  equivalence relations are
 equivalent: obviously in $\C=\Sets$; more generally in a \emph{regular}
  category $\C$.

 A category is said to be regular if it has finite limits
 and a factorization structure which is compatible with limits in a
 suitable sense. If it is the case, a relation $R\rightrightarrows X$ 
 induces a ``quotient map'' $X\epi X/R$ as its coequalizer (which always
 exists).
 % \begin{displaymath}
 %  \vcenter{\xymatrix@1@C+2em{
 %   {R}
 %                 \ar@<.3em>[r]^{\pi_{1}\co i}
 %                 \ar@<-.3em>[r]_{\pi_{2}\co i}
 %  &
 %   {X}
 %                 \ar@{->>}[r]^{e}
 %  &
 %   {X/R}
 %}}
 % \end{displaymath}
 Conversely, given an $\mathcal{E}$-arrow $X\stackrel{e}{\epi} X/R$ the
 corresponding relation $R\rightrightarrows X$ can
 be recovered as its kernel pair. See e.g.~\cite[Section~4.4]{Jacobs99a}
 for details.

 % Let $f: X\to Y$ be an arrow in regular $\C$, $f=m\co e$ be the
 % image factorization and $r_{1}, r_{2}:R\rightrightarrows X$
 % be the kernel pair of $f$. 
 % \begin{displaymath}
 %  \vcenter{\xymatrix@R=1em@C+2em{
 %   {R}
 %                 \ar@<.3em>[r]^{r_{1}}
 %                 \ar@<-.3em>[r]_{r_{2}}
 %  &
 %   {X}
 %                 \ar[r]^{f}
 %                 \ar@{->>}[d]_*+{\scriptstyle e}
 %  &
 %   {Y}
 %  \\
 %  &
 %   {\mathrm{Im}(f)}
 %                 \ar@{ >->}[ru]_{m}
 %}}
 % \end{displaymath}
 % Then the pair of $r_{1}$ and $r_{2}$ is the kernel pair of $e$; conversely
 % $e$ is the coequalizer of $r_{1}$ and $r_{2}$.
 % For the proof of this fact (which is not trivial), see~\cite[Lemma~4.4.6.viii]{Jacobs99a},
 % which is essentially from~\cite[Theorem~1.52]{Johnstone77}.
 \end{rem}
The following standard result shows that both $\TestEq_{c}$ and $\FCSEq_{c}$
are ``equivalence relations'' in a suitable sense. Its proof is straightforward.
\begin{prop}
 In a category with finite limits, a kernel pair $r_{1}, r_{2}: R\rightrightarrows X$
 induces an  \textbf{equivalence
 relation} $\tuple{r_{1}, r_{2}}: R\mono X\times X$ (see
 e.g.~\cite[Section~1.3]{Jacobs99a}). That is, it satisfies
 \begin{enumerate}[$\bullet$]
  \item \textbf{Reflexivity}. The diagonal relation $X\mono X\times X$ 
        factors  as follows.
	\begin{displaymath}
	   \vcenter{\xymatrix@R=1em@C+3em{
	    {R}
			       \ar@{ >->}[r]^{\scriptstyle\tuple{r_{1},r_{2}}}
	   &
	    {X\times X}
%			       \ar@<.2em>[r]^-{\scriptstyle\theory_{c}\co \pi_{1}}
%			       \ar@<-.2em>[r]_-{\scriptstyle\theory_{c}\co \pi_{2}}
	   \\
	    {X}
                               \ar@{-->}[u]
                               \ar@{ >->}[ru]_{\tuple{\id,\id}}
	 }}
	\end{displaymath}
  \item \textbf{Symmetry}. The reversed relation $\tuple{r_{2}, r_{1}}$
	factors as follows.
	\begin{displaymath}
	   \vcenter{\xymatrix@R=1em@C+3em{
	    {R}
			       \ar@{ >->}[r]^{\scriptstyle\tuple{r_{1},r_{2}}}
	   &
	    {X\times X}
%			       \ar@<.2em>[r]^-{\scriptstyle\theory_{c}\co \pi_{1}}
%			       \ar@<-.2em>[r]_-{\scriptstyle\theory_{c}\co \pi_{2}}
	   \\
	    {R}
                               \ar@{-->}[u]
                               \ar@{ >->}[ru]_{\tuple{r_{2}, r_{1}}}
	 }}
	\end{displaymath}
  \item \textbf{Transitivity}. Let $Q$ be a pullback as in the following
	diagram.
	The intuition is that $Q=\{(x,y,z)\mid (x,y) \in R \,\land\, (y,z)\in R\}$.
        Then the subobject $\tuple{r_{1}\co r_{12},\; r_{2}\co r_{23}}: Q\mono X\times X$ factors 
	as in the diagram on the right.
	\begin{displaymath}
	   \vcenter{\xymatrix@R=1em@C=2em{
	    {Q}
	                       \pb{315}
	                       \ar[r]^{r_{23}}
	                       \ar[d]_{r_{12}}
%			       \ar@{ >->}[r]^{\scriptstyle\tuple{r_{1},r_{2}}}
	   &
	    {R}
                               \ar[d]^{r_{1}}
%			       \ar@<.2em>[r]^-{\scriptstyle\theory_{c}\co \pi_{1}}
%			       \ar@<-.2em>[r]_-{\scriptstyle\theory_{c}\co \pi_{2}}
	   \\
	    {R}
                               \ar[r]_{r_{2}}
           &
            {X}
	 }}
           \qquad
	   \vcenter{\xymatrix@R=1em@C+3em{
	    {R}
			       \ar@{ >->}[r]^{\scriptstyle\tuple{r_{1},r_{2}}}
	   &
	    {X\times X}
%			       \ar@<.2em>[r]^-{\scriptstyle\theory_{c}\co \pi_{1}}
%			       \ar@<-.2em>[r]_-{\scriptstyle\theory_{c}\co \pi_{2}}
	   \\
	    {Q}
                               \ar@{-->}[u]
                               \ar@{ >->}[ru]_*+{\scriptstyle\tuple{r_{1}\co r_{12},\; r_{2}\co r_{23}}}
	 }}
	\end{displaymath}
 \end{enumerate}
In particular,  both $\TestEq_{c}$ and $\FCSEq_{c}$ are equivalence
 relations in the above sense. \qed
\end{prop}

Let $r_{1}, r_{2}$ be the kernel pair for $f$. For transitivity, we show
that the pair of $r_{1}\co r_{12}$ and $r_{2}\co r_{23}$ is equalized by
$f$.
\begin{align*}
 f\co r_{1}\co r_{12}
&=
 f\co r_{2}\co r_{12}
 &&\text{$f$ equalizes $r_{1}$ and $r_{2}$}
\\
&=
 f\co r_{1}\co r_{23}
 &&\text{pullback diagram}
\\
&=
 f\co r_{2}\co r_{23}\enspace.
\end{align*}
}

Now our concern is the comparison between two process semantics
$\TestEq_{c}$ and $\FCSEq_{c}$, as subobjects of $X\times X$.
The following lemma is crucial for our investigation; in fact it
is important for coalgebraic modal logic in general and appears e.g.\ as~\cite[Theorem~3.3]{Klin07}.
\begin{lem}\label{lemma:theoryFactorsThroughCoalgebraMorphism}
 A morphism of $F$-coalgebras preserves theory maps. That is,
 \begin{displaymath}
   \vcenter{
 \xymatrix@C+3em@R=1em{
  {F X}
                \ar[r]^{Ff}
 &
  {F Y}
 \\
  {X}
                \ar[u]^{c}
                \ar[r]_{f}
 &
  {Y}
%  \shifted{2em}{-.2em}{.}
                \ar[u]_{d}
}
}
\quad\text{implies}\quad
   \vcenter{
 \xymatrix@C+3em@R=1em{
  {X}
                \ar[rd]^{\theory_{c}}
                \ar[d]_{f}
 &
 \\
  {Y}
                \ar[r]_{\theory_{d}}
 &
  {SL}
  \shifted{2em}{-.2em}{.}
}
}
 \end{displaymath}
\end{lem}
\proof
 The following induction diagram
% ---which commutes due to naturality of $\delta$---
 proves
 $Pf\co\sem{\place}_{d}=\sem{\place}_{c}$.
 Naturality of $\delta$ plays an important role there.
\begin{displaymath}
    \vcenter{
 \xymatrix@C+3em@R=1em{
  {ML}
                \ar@{-->}[r]^{Ff}
                \ar[dd]^{\cong}_{\scriptstyle\text{initial}}
 &
  {MPY}
                \ar[r]^{MPf}
                \ar[d]^{\delta_{Y}}
 &
  {MPX}
                \ar[d]^{\delta_{X}}
 \\
 &
  {PFY}
                \ar[r]_{PFf}
                \ar[d]^{Pd}
 &
  {PFX}
                \ar[d]^{Pc}
 \\
  {L}
                \ar@{-->}[r]_{\sem{\place}_{d}}
 &
  {PY}
                \ar[r]_{Pf}
 &
  {PX}
  \shifted{2em}{-.2em}{.}
}
}
\end{displaymath}
Then the claim follows from naturality of the
 transposition~(\ref{diagram:theoryAsTranspositionOfInterpretation}).
 \qed

We show that in a testing situation like
(\ref{diagram:testingSituation}),
tests respect final coalgebra semantics. 
That is, testing does not distinguish two FCS-equivalent states.
\begin{prop}\label{proposition:TestEqIsCoarserThanFCSEq}
 Consider such a testing situation and equivalence relations as in Definitions~\ref{definition:TestEq}
 and \ref{definition:FCSEq}. For any coalgebra $X\stackrel{c}{\to}FX$ we
 have an inclusion 
\begin{displaymath}
  \FCSEq_{c}\le\TestEq_{c}
\end{displaymath} 
of subobjects of
 $X\times X$.
\end{prop}
\proof
 It suffices to show that the arrow $\tuple{q_{1},q_{2}}$ in
 (\ref{diagram:defOfFCSEq}) equates the parallel arrows in
 (\ref{diagram:defOfTestEq}); then the claim follows from
 universality of an equalizer.
 \begin{align*}
  \theory_{c}\co \pi_{1}\co \tuple{q_{1},q_{2}}
 &=
  \theory_{c}\co q_{1}
 \\
 &=
  \theory_{\zeta}\co \beh_{c} \co q_{1} 
 &&(*)
 \\
 &=
  \theory_{\zeta}\co \beh_{c} \co q_{2}
 &&\text{due to (\ref{diagram:defOfFCSEq})} 
 \\
 &=
  \theory_{c}\co q_{2}
 &&(*)
 \\
 &=
  \theory_{c}\co \pi_{2}\co \tuple{q_{1},q_{2}}\enspace. 
 \end{align*}
Here $(*)$ is an instance of
 Lemma~\ref{lemma:theoryFactorsThroughCoalgebraMorphism}:
$\beh_{c}$ is a morphism of coalgebras from $c$ to the final $\zeta$.
 \qed

The converse  $\TestEq_{c}\le\FCSEq_{c}$ does not hold in general.
For a fixed type of systems (i.e.\ for fixed $F:\C\to\C$), we can
think of logics with varying degree of expressive power; this results in
process equivalences with varying granularity. This view is  
systematically presented by van Glabbeek in~\cite{vanGlabbeek01} as the \emph{linear
time-branching time spectrum}---a categorical version of which we consider as an  important
direction of
future work.

It is when we have $\FCSEq_{c}\iso\TestEq_{c}$ that a modal logic
(considered as a testing situation) is said to be \emph{expressive}.
Recall that $\FCSEq_{c}$ usually coincides with \emph{bisimilarity}
if $\C$ is $\Sets$: in this case an expressive logic \emph{captures}
bisimilarity.

The following proposition states a (rather trivial) equivalent condition
for a testing situation to be expressive.
For more ingenious sufficient conditions---which 
essentially rely on the transpose of $\delta$ being monic---see e.g.~\cite{Klin07}.

\begin{prop}\label{proposition:expressiveTesting}
 Consider a testing situation as in Definitions~\ref{definition:TestEq}
 and \ref{definition:FCSEq}. 
The testing is expressive, that is,
 for any coalgebra $X\stackrel{c}{\rightarrow}FX$ we have 
\begin{displaymath}
  \TestEq_{c}\iso\FCSEq_{c}
\end{displaymath}
 as subobjects of $X\times X$,
 if and only if the theory map $\theory_{\zeta}: Z\to SL$ 
 for the final coalgebra is a mono.
\end{prop}
\proof
 We first prove the ``if'' direction.
 In view of Proposition~\ref{proposition:TestEqIsCoarserThanFCSEq},
 it suffices to show that $\tuple{p_{1}, p_{2}}$ in (\ref{diagram:defOfTestEq}) equalizes
 $\beh_{c}\co \pi_{1}$ and $\beh_{c}\co \pi_{2}$ (which proves
 $\TestEq_{c}\le\FCSEq_{c}$).
 \begin{align*}
  \theory_{\zeta}\co \beh_{c} \co p_{1}
 &=
  \theory_{c}\co p_{1}
 &&\text{by Lemma~\ref{lemma:theoryFactorsThroughCoalgebraMorphism}}
 \\
 &=
  \theory_{c}\co p_{2}
 &&\text{due to (\ref{diagram:defOfTestEq})}
 \\
 &=
  \theory_{\zeta}\co \beh_{c} \co p_{2}
 &&\text{by Lemma~\ref{lemma:theoryFactorsThroughCoalgebraMorphism}}
 \end{align*}
 We have $\beh_{c}\co p_{1}=\beh_{c}\co p_{2}$ since $\theory_{\zeta}$
 is a mono. 

 To prove the ``only if'' direction, first we observe that the
 FCS-equivalence on the final coalgebra $\zeta: Z\iso FZ$ is
 the diagonal relation: that is,
 \begin{displaymath}
    \vcenter{\xymatrix@R=1em@C+1em{
   {\FCSEq_{\zeta}}
                      \ar@{ >->}[r]%^{\scriptstyle\tuple{q_{1},q_{2}}}
                      \ar[d]_{\cong}
  &
   {Z\times Z}
                      \ar@<.2em>[r]^-{\scriptstyle\beh_{\zeta}\co \pi_{1}}
                      \ar@<-.2em>[r]_-{\scriptstyle\beh_{\zeta}\co \pi_{2}}
  &
   {Z}
  \\
   {Z}
                      \ar@{ >->}[ru]_{\scriptstyle\tuple{\id,\id}}
  &&
   {}\shifted{2em}{-.2em}{.}
}}
 \end{displaymath}
 This is because $\beh_{\zeta}=\id: Z\to Z$.
 Now assume that $\theory_{\zeta}\co k = \theory_{\zeta}\co l$ for $k,l: Y\rightrightarrows Z$.
   Universality of an equalizer  $\TestEq_{\zeta}$  induces a
 mediating arrow $m$ in the following diagram.
 \begin{displaymath}
    \vcenter{\xymatrix@R=1em@C+2em{
   {Y}
                      \ar[rr]^-{\tuple{k,l}}
                      \ar@{-->}[d]_{m}
  &&
   {Z\times Z}
                      \ar@<.2em>[r]^-{\scriptstyle\beh_{\zeta}\co \pi_{1}}
                      \ar@<-.2em>[r]_-{\scriptstyle\beh_{\zeta}\co \pi_{2}}
  &
   {Z}
  \\
   {\TestEq_{\zeta}}
                      \ar@{ >->}[rru]%_{\scriptstyle\tuple{\id,\id}}
                      \ar[r]_{\cong}
  &
   {\FCSEq_{\zeta}}
                      \ar[r]_{\cong}
  &
   {Z}
                      \ar@{ >->}[u]_{\tuple{\id,\id}}
}}
 \end{displaymath}
 The whole diagram commutes since $\TestEq_{\zeta}\cong \FCSEq_{\zeta}$
 (by assumption) and $\FCSEq_{\zeta}\cong Z$ (by the above observation),
 both as subobjects of $Z\times Z$. This proves $k=l$.
 \qed

\begin{rem}
 The literature~\cite{BonsangueK05,BonsangueK06} considers
 more restricted settings than the testing situations in
 Definition~\ref{definition:testingSituations}. There
 an adjunction $S^{\op}\dashv P$ is replaced by a dual equivalence of
 categories, and a denotation $\delta$ is required to be a natural
 isomorphism. These additional restrictions allow one to say more
 about the situations: logics are always expressive; the main concern
 of~\cite{BonsangueK06} is how to present an abstract modal logic
 $M: \A\to\A$ by concrete syntax. However, for our purpose in
 Section~\ref{subsection:testingForTraceSemantics} the greater generality
 of our notion of testing situations is needed.
\end{rem}

\subsection{Canonical testing for trace semantics in
 \texorpdfstring{$\Kleisli{\pow}$}{Kl(P)}} 
\label{subsection:testingForTraceSemantics} 
In this section we shall
present a canonical testing situation 
%(as defined in
%Definition~\ref{definition:testingSituations})
for coalgebras in $\Kleisli{\pow}$. We shall also show that the testing
is ``expressive,'' in
the sense that the testing captures final coalgebra semantics.
The intuition is as follows.

Trace semantics for non-deterministic systems assigns to each system $c$
%  (i.e.\ a coalgebra $X\to\oF X$ in $\Kleisli{\pow}$) 
its ``(finite) trace set'' map
$\trace{c}: X\to \pow A$, where $A$ carries an initial algebra in $\Sets$.
This suggests a natural testing framework where: an element $t$ of $A$
is a test; a state $x\in X$ of a system passes a test $t$ if and only if
the trace set of $x$ includes $t$ (i.e.\ $x\models t\Longleftrightarrow t\in\trace{c}(x)$). 
%Moreover this testing framework should be expressive enough to capture
%finite trace semantics, which is the final coalgebra semantics in $\Kleisli{\pow}$.
An important point here is that
$A$, carrying an initial algebra in $\Sets$, usually gives a \emph{well-founded
syntax} for tests.\footnote{%
   Recall the construction of an initial
   algebra in $\Sets$ via the initial sequence
   (Proposition~\ref{LemInitialAlgebraViaInitialSequence}).
   The set $A$ is the colimit (\emph{union} in $\Sets$) of the initial sequence
   $0\to F0\to F^{2}0\to\cdots$. Each $F^{n}0$ can be thought of as
   the set of terms with depth $\le n$.}

We focus on a non-deterministic setting (i.e.\ $T=\pow$)
in this section and leave a probabilistic one as future work.
Although the above intuition is true in probabilistic settings as 
well---where the 2-valued (pass/failure) observation scheme is replaced
by the refined $[0,1]$-valued one---we do not know yet how to extend
the current material to probabilistic settings.
The difficulty is that the category $\Kleisli{\dist}$ is not self-dual,
as opposed to $\Kleisli{\pow}$; see
(\ref{diagram:testingSituationForTraceSemantics}) below.

The canonical testing situation which captures finite trace semantics 
is the following one.
\begin{equation}\label{diagram:testingSituationForTraceSemantics}
 \xymatrix@1@C+3em{
  {\Kleisli{\pow}^{\op}}
              \ar@(ld,rd)[]_(.3){\oF^{\op}}
              \ar@/^.7em/[r]^{\Op}
              \ar@{}[r]|{\cong}
 &
  {\Kleisli{\pow}}
              \ar@/^.7em/[l]^{\Op^{\op}}
%              \ar@(lu,ld)[]_{\oF}
              \ar@/^.7em/[r]^{K}
              \ar@{}[r]|{\top}
 &
  {\Sets}
              \ar@(ru,lu)[]_(.3){F}
              \ar@/^.7em/[l]^{J}
 }
\end{equation}
Here $J\dashv K$ is the canonical Kleisli adjunction. 
Recall the self duality $\Op:\Kleisli{\pow}^{\op}\iso\Kleisli{\pow}$
from Section~\ref{subsection:simplerProofInRel}.
The denotation is given by (the components of) the distributive law
$\lambda:F\pow\Rightarrow\pow F$. The following lemma establishes naturality
of the denotation.
\begin{lem}\label{lemma:naturalityOfDenotationForTraceSemantics}
Let $F:\Sets\to\Sets$ be a functor which preserves weak pullbacks, and
 $\oF$ be its lifting induced by the relation lifting
 (Lemma~\ref{lemma:distributiveLawForPowersetMonad}). Then the
 components
 $F\pow X\stackrel{\lambda_{X}}{\to} \pow FX$
of the corresponding distributive law $\lambda$
also form a natural transformation 
\begin{displaymath}
 F\co K\co \Op \Longrightarrow K\co\Op\co \oF^{\op}\; :\; 
 \Kleisli{\pow}^{\op}\longrightarrow \Sets\enspace.
\end{displaymath}
\end{lem}
\proof
The desired natural transformation is obtained from another natural 
transformation
\begin{displaymath}
 \lambda'\; :\; FK\Longrightarrow K\oF \;:\; \Kleisli{\pow}\longrightarrow \Sets
\end{displaymath}  
which we describe in a moment, by post-composing the functor $\Op$. 
That is, the desired one is the composite
\begin{displaymath}
 FK\Op 
 \;\stackrel{\lambda'\co \Op}{\Longrightarrow}\;
 K\oF\Op
 \;\stackrel{(*)}{=} \;
 K\Op\oF^{\op}\enspace,
\end{displaymath}
or equivalently, in a 2-categorical presentation,
\begin{displaymath}
 \xymatrix@R=1.2em@C+2em{
   {\Kleisli{\pow}^{\op}}
                  \ar[r]^{\Op}
                  \ar[d]_{\oF^{\op}}
               	 \drtwocell<\omit>{={\hspace{-.8em}(*)}}
%                  \ar@{}[dr]|*[@ld]{{=}}
%                  \ar@{}[dr]^{(*)}
  &
   {\Kleisli{\pow}}   
                  \ar[r]^{K}
                  \ar[d]_{\oF}
               	 \drtwocell<\omit>{<-.1>{\,\,\lambda'}}                
  &
   {\Sets}   
                  \ar[d]^{F}
  \\
   {\Kleisli{\pow}^{\op}}
                  \ar[r]_{\Op}
  &
   {\Kleisli{\pow}}   
                  \ar[r]_{K}
  &
   {\Sets}   
   \shifted{2em}{-.2em}{.}
}
\end{displaymath}
 Here the equality ($*$) is the one
 in~(\ref{diagram:liftedFAndSelfDualityOfRelAreCompatible}).
 
 Now we describe the natural transformation $\lambda'$. Its components are given by those of
 $\lambda$; naturality of $\lambda'$ is an easy consequence of
 $\lambda$'s being a distributive law. Indeed, given an arrow $f: X\to
 Y$ in $\Kleisli{\pow}$, the following shows that the naturality square commutes.
  \begin{align*}
   K\oF f\co \lambda_{X}
  &
   = \mu^{\pow}_{FY}\co \pow\oF f\co \lambda_{X}  
  &&\text{definition of $K$}
  \\
  &
   = \mu^{\pow}_{FY}\co \pow\lambda_{Y}\co \pow Ff\co \lambda_{X}  
  &&\text{definition of $\oF$}
  \\
  &
   = \mu^{\pow}_{FY}\co \pow\lambda_{Y}\co \lambda_{\pow Y}\co F\pow f  
  &&\text{naturality of $\lambda$}
  \\
  &
   = \lambda_{Y}\co F\mu^{\pow}_{Y}\co F\pow f  
  &&\text{$\lambda$ is compatible with the multiplication $\mu^{\pow}$ of $\pow$}
  \\
  &
   = \lambda_{Y}\co FKf
  &&\text{definition of $K$}  \tag*{$\qEd$}
  \end{align*}

\noindent The previous lemma establishes that the situation
(\ref{diagram:testingSituationForTraceSemantics})
is indeed a testing situation as defined in
Definition~\ref{definition:testingSituations}.

In the previous Section~\ref{subsection:testingSituations}, the use of
testing situations is demonstrated through comparing testing equivalence
and final coalgebra semantics, both described as suitable kernel pairs.
Unfortunately this argument is not valid in the current
situation~(\ref{diagram:testingSituationForTraceSemantics}), since the category
$\Kleisli{\pow}$ does not have kernel pairs. 

Still, we shall claim that the situation
(\ref{diagram:testingSituationForTraceSemantics}) is ``expressive,'' in
the sense that final coalgebra
semantics is captured by testing. This claim is supported by the
following fact: in the current situation  the two arrows $\trace{c}$ and $\theory_{c}$ 
simply coincide. Therefore their kernel relations---in any reasonable formalization---should coincide as
well.

\begin{prop}
 Let $X\stackrel{c}{\to}\oF X$ be a coalgebra in $\Kleisli{\pow}$.
 In the testing situation
 (\ref{diagram:testingSituationForTraceSemantics}),
 the following arrows in $\Kleisli{\pow}$ coincide.
 \begin{enumerate}[$\bullet$]
  \item $\trace{c}: X\to A$,  giving the final
	coalgebra (trace) semantics for $c$.
  \item $\theory_{c}: X\to A$, giving the testing semantics, i.e.\ the set of passed tests.
 \end{enumerate}
 Therefore the testing is ``expressive'': 
 tests from an initial $F$-algebra captures trace semantics
 (which is via a final $\oF$-coalgebra).
\end{prop}
\noindent 
 Here $A$ is the carrier of an initial $F$-algebra, hence that of a final $\oF$-coalgebra.
Note that, in the general setting in
Section~\ref{subsection:testingSituations},
the codomains of $\trace{c}$ and $\theory_{c}$ need not
coincide.
\proof
 We shall show that the transpose 
\begin{displaymath}
  \trace{c}^{\lor}\;:\; A\longrightarrow \pow X \qquad\text{in $\Sets$}
\end{displaymath} 
of $\trace{c}$ under the adjunction
 in (\ref{diagram:testingSituationForTraceSemantics}) makes the
 diagram
 (\ref{diagram:interpretationOfModalFormulaViaInduction})---which
 defines $\sem{\place}_{c}$---commute. This
 proves $\trace{c}^{\lor}=\sem{\place}_{c}$, hence
 $\trace{c}={\sem{\place}_{c}}^{\lor} = \theory_{c}$.
 
 First note that the transpose $\trace{c}^{\lor}: A\to \pow X$ is given by the arrow $\Op(\trace{c}): A\to X$ in
 $\Kleisli{\pow}$ thought of as an arrow in $\Sets$.  In the
 sequel we shall write $\Op(\trace{c})$ for $\trace{c}^{\lor}$.

 Commutativity of the diagram
 (\ref{diagram:coinductionInKleisliCategory})---defining
 $\trace{c}$---yields the following equality.
% shows that the following two arrows in $\Kleisli{\pow}$
% are equal.
 \begin{displaymath}
  \Op(\trace{c})\co \Op (J\alpha^{-1}) = \Op(c)\co \Op(\oF^{\op} \trace{c})
  \qquad\text{in $\Kleisli{\pow}$.}
 \end{displaymath}
 By the definition of composition in $\Kleisli{\pow}$, it reads as
 follows in $\Sets$.
 \begin{equation}\label{equation:hogehogehoog}
  \mu_{X} \co
  \pow(\Op(\trace{c})) \co
  \Op(J\alpha^{-1})
 = 
  \mu_{X} \co
  \pow(\Op(c)) \co
  \Op(\oF^{\op} \trace{c})
 \end{equation}
 We use this equality in showing that $\Op(\trace{c})$ makes the diagram
 (\ref{diagram:interpretationOfModalFormulaViaInduction}) commute.
 \begin{align*}
  \Op(\trace{c})\co \alpha
 &=
  \mu_{X}\co\eta_{X}\co\Op(\trace{c})\co \alpha  
 &&\text{unit law}
 \\
 &=
  \mu_{X}\co\pow(\Op(\trace{c}))\co \eta_{A}\co \alpha  
 &&\text{naturality of $\eta$}
 \\
 &=
  \mu_{X}\co\pow(\Op(\trace{c}))\co   \Op(J\alpha^{-1})
 &&\text{$\Op(J\alpha^{-1}) = J\alpha=\eta_{A}\co \alpha$}
 \\
 &=
  \mu_{X} \co
  \pow(\Op(c)) \co
  \Op(\oF^{\op} \trace{c})
 &&\text{by (\ref{equation:hogehogehoog})}
 \\
 &=
  \mu_{X} \co
  \pow(\Op(c)) \co
  \oF\Op( \trace{c})
 &&\text{$\Op\oF^{\op} = \oF\Op$, (\ref{diagram:liftedFAndSelfDualityOfRelAreCompatible})}
 \\
 &=
  \mu_{X} \co
  \pow(\Op(c)) \co
  \lambda_{X} \co
  F\Op( \trace{c})
 &&\text{definition of $\oF$}
 \\
 &=
  K\Op(c) \co
  \lambda_{X} \co
  F\Op( \trace{c})\enspace.
% &&\text{definition of $\oF$}
 \end{align*}
 Recall that $M$ in
 (\ref{diagram:interpretationOfModalFormulaViaInduction}) is now $F$;
 $P$ in
 (\ref{diagram:interpretationOfModalFormulaViaInduction})
 is now $K\Op$. This concludes the proof.
\qed

 The proposition establishes a connection between two semantics for
 $\oF$-coalgebras in $\Kleisli{\pow}$, namely: $\trace{c}$ via a
 final $\oF$-coalgebra, and $\theory_{c}$ via an initial
 $F$-algebra. One may well say that it is a ``degenerate'' case because,
 as we have shown in Section~\ref{section:finalCoalgebraInKleisliCat},
 coinduction in $\Kleisli{\pow}$ and induction in $\Sets$ are essentially
 the same principle. Our emphasis is more on the fact that the
 coincidence of induction and coinduction yields a rather uncommon example of testing
 situations. Testing situations are of interest in modal
 logic---where the underlying contravariant adjunction
 $S^{\op}\dashv P: \A\to\C^{\op}$ in~(\ref{diagram:testingSituation}) is
 often the Stone duality or one of its variants. Our
 example $\Kleisli{\pow}^{\op}\leftrightarrows\Sets$ here does not look like one of those familiar examples.

\section{Conclusions and future work}
We have developed a mathematical principle underlying 
``trace semantics'' for various kinds of branching systems, namely
coinduction in a Kleisli category. This general view is supported by
a technical result that a final coalgebra in a Kleisli category is induced by
an initial algebra in $\Sets$.

%In study of coalgebras as `categorical presentation of state-based
%systems', there are three ingredients playing crucial roles:
%\emph{coalgebras} as systems; \emph{coinduction} yielding process
%semantics; and \emph{morphisms of coalgebras} as behavior-preserving
%maps. In this paper we have studied the first two in a Kleisli
%category. What about morphisms of coalgebras?

%In~\cite{Hasuo06a}  this question is answered by identifying
%\emph{lax/oplax morphisms of coalgebras} in a Kleisli category as
%\emph{forward/backward simulations}. Use of traces and simulations is a common
%technique in formal verification of systems (see e.g.~\cite{LynchV95}): a desirable property is
%expressed in terms of traces; and then a system is shown to satisfy the
%property
%by finding a suitable simulation. This paper, together
%with~\cite{Hasuo06a}, forms an essential part of developing 
%a ``generic theory of traces and simulations'' using coalgebras in
%a Kleisli category. The categorical genericity---especially the fact
%that we can treat non-deterministic and probabilistic branching in a
%uniform manner---is exploited in~\cite{HasuoK06} to obtain a
%simulation-based proof method for a probabilistic notion of anonymity
%for network protocols. Currently we are investigating how much more 
%applicational impact can be brought about by our generic theory of traces and simulations.

%non-det, prob.
The possible instantiations of our generic framework include
non-deterministic systems and probabilistic systems, but do \emph{not} yet include
systems with both non-deterministic and probabilistic branching.
The importance of having both of these branchings in system verification
has been claimed by many authors e.g.~\cite{Vardi85,Seg95:thesis}, with
an intuition that probabilistic branching models the choices ``made by
the system, i.e.\ on \emph{our} side,'' while (coarser) non-deterministic
choices are ``made by the (unknown) environment of the system, i.e.\ on
the \emph{adversary's} side.''
A typical example of such systems is given by \emph{probabilistic
automata} introduced
by Segala~\cite{Seg95:thesis}.

In fact this combination of non-deterministic and probabilistic
branching
is a notoriously difficult one from a theoretical point of view~\cite{Cheung06,VaraccaW05,TixKP05}:
many mathematical tools that are useful in a purely non-deterministic
or probabilistic setting cease to work in the presence of both.
For our framework of generic trace semantics, the problem is that
we could not find a suitable monad $T$ with an order structure.

We have used the order-enriched structure of a Kleisli category
(expressing ``more possibilities'') to obtain the 
initial algebra-final coalgebra coincidence result. However, an
order structure is not the only one that can yield such coincidence:
other examples include metric, quasi-metric and quantale-enriched
structures (in increasing generality). See 
e.g.~\cite{TuriR98,Fiore96a} for
the potential use of such enriched structures in a coalgebraic setting.
The relation of the current work to such structures is yet to be investigated.

In the discipline of process algebra, a system is represented by an
\emph{algebraic} term (such as $a.P\parallel a.Q$) and a structural
operational semantics (SOS) rule determines its dynamics, that is, its
\emph{coalgebraic} structure. This is where ``algebra
meets coalgebra'' and the interaction is studied
e.g.\ in~\cite{TuriP97,Bartels04,Klin05}. In our recent work~\cite{HasuoJS07a}
we claim the importance of the \emph{microcosm principle} in this
context and provide a ``general compositionality theorem'': under suitable
assumptions,  the final coalgebra semantics is compatible with 
the algebraic structure. The results of the current paper say that
the final coalgebra semantics can be interpreted as finite trace
semantics, hence the result in~\cite{HasuoJS07a} also yields a general
compositionality result for trace semantics.

In this paper we have included some material---on possibly-infinite traces 
and testing situations---which, unfortunately,
we have worked out only in a non-deterministic setting.
A fully general account on these topics is left as future work.

Finally, there are so many
different process semantics for branching systems, between two
edges of bisimilarity and trace equivalence in the linear time-branching
time spectrum~\cite{vanGlabbeek01}. How to capture them in a coalgebraic setting
is, we believe, an important and challenging question.

%#END

\section*{Acknowledgment}
 Thanks are due to  Ji\v{r}\'{i} Ad\'{a}mek, Chris Heunen, Stefan Milius,
Tarmo Uustalu and the anonymous referees 
 for helpful discussions and comments.

\bibliographystyle{FirstInitialPlain}
\bibliography{../../Macros/tex/myrefs}

\newcommand{\noopsort}[1]{} \newcommand{\singleletter}[1]{#1}
\begin{thebibliography}{10}

\bibitem{AbramskyJ94}
S.~Abramsky and A.~Jung.
\newblock Domain theory.
\newblock In S.~Abramsky, D.M. Gabbai, and T.S.E. Maibaum, editors, {\em
  Handbook of Logic in Computer Science}, volume~3, pages 1--168. Oxford Univ.
  Press, 1994.

\bibitem{AdamekK79}
J.~Ad\'{a}mek and V.~Koubek.
\newblock Least fixed point of a functor.
\newblock {\em Journ. Comp. Syst. Sci}, 19(2):163--178, 1979.

\bibitem{BarrW85}
M.~Barr and C.~Wells.
\newblock {\em Toposes, Triples and Theories}.
\newblock Springer, Berlin, 1985.
\newblock Available online.

\bibitem{Bartels04}
F.~Bartels.
\newblock {\em On generalised coinduction and probabilistic specification
  formats. Distributive laws in coalgebraic modelling}.
\newblock PhD thesis, Free Univ. Amsterdam, 2004.

\bibitem{BonsangueK05}
M.M. Bonsangue and A.~Kurz.
\newblock Duality for logics of transition systems.
\newblock In V.~Sassone, editor, {\em FoSSaCS}, volume 3441 of {\em Lect. Notes
  Comp. Sci.}, pages 455--469. Springer, 2005.

\bibitem{BonsangueK06}
M.M. Bonsangue and A.~Kurz.
\newblock Presenting functors by operations and equations.
\newblock In L.~Aceto and A.~Ing{\'o}lfsd{\'o}ttir, editors, {\em FoSSaCS},
  volume 3921 of {\em Lect. Notes Comp. Sci.}, pages 172--186. Springer, 2006.

\bibitem{Borceux94}
F.~Borceux.
\newblock {\em Handbook of Categorical Algebra}, volume 50, 51 and 52 of {\em
  Encyclopedia of Mathematics}.
\newblock Cambridge Univ. Press, 1994.

\bibitem{Cheung06}
L.~Cheung.
\newblock {\em Reconciling Nondeterministic and Probabilistic Choices}.
\newblock PhD thesis, Radboud Univ. Nijmegen, 2006.

\bibitem{Fiore96b}
M.P. Fiore.
\newblock {\em Axiomatic Domain Theory in Categories of Partial Maps}.
\newblock Distinguished Dissertations in Computer Science. Cambridge Univ.
  Press, 1996.

\bibitem{Fiore96a}
M.P. Fiore.
\newblock A coinduction principle for recursive data types based on
  bisimulation.
\newblock {\em Inf. \& Comp.}, 127(2):186--198, 1996.

\bibitem{Fokkinga94}
M.M. Fokkinga.
\newblock Monadic maps and folds for arbitrary datatypes.
\newblock {\em Memoranda Informatica, University of Twente}, 94--28, 1994.

\bibitem{Freyd91}
P.J. Freyd.
\newblock Algebraically complete categories.
\newblock In A.~Carboni, M.C. Pedicchio, and G.~Rosolini, editors, {\em Como
  Conference on Category Theory}, number 1488 in Lect. Notes Math., pages
  95--104. Springer, Berlin, 1991.

\bibitem{Freyd92}
P.J. Freyd.
\newblock Remarks on algebraically compact categories.
\newblock In M.P. Fourman, P.T. Johnstone, and A.M. Pitts, editors, {\em
  Applications of Categories in Computer Science}, number 177 in LMS, pages
  95--106. Cambridge Univ. Press, 1992.

\bibitem{Hasuo06a}
I.~Hasuo.
\newblock Generic forward and backward simulations.
\newblock In C.~Baier and H.~Hermanns, editors, {\em International Conference
  on Concurrency Theory (CONCUR 2006)}, volume 4137 of {\em Lect. Notes Comp.
  Sci.}, pages 406--420. Springer, Berlin, 2006.

\bibitem{HasuoJ05c}
I.~Hasuo and B.~Jacobs.
\newblock Coalgebraic trace semantics for probabilistic systems.
\newblock In P.~Mosses, J.~Power, and M.~Seisenberger, editors, {\em CALCO-jnr
  Workshop}, 2005.

\bibitem{HasuoJ05b}
I.~Hasuo and B.~Jacobs.
\newblock Context-free languages via coalgebraic trace semantics.
\newblock In J.L. Fiadeiro, N.~Harman, M.~Roggenbach, and J.J.M.M. Rutten,
  editors, {\em International Conference on Algebra and Coalgebra in Computer
  Science (CALCO'05)}, volume 3629 of {\em Lect. Notes Comp. Sci.}, pages
  213--231. Springer, Berlin, 2005.

\bibitem{HasuoJS06a}
I.~Hasuo, B.~Jacobs, and A.~Sokolova.
\newblock Generic trace theory.
\newblock In N.~Ghani and A.J. Power, editors, {\em International Workshop on
  Coalgebraic Methods in Computer Science (CMCS 2006)}, volume 164 of {\em
  Elect. Notes in Theor. Comp. Sci.}, pages 47--65. Elsevier, Amsterdam, 2006.

\bibitem{HasuoJS07a}
I.~Hasuo, B.~Jacobs, and A.~Sokolova.
\newblock The microcosm principle and concurrency in coalgebras, 2007.
\newblock Preprint, available from \verb+http://www.cs.ru.nl/~ichiro/papers+.

\bibitem{HasuoK06}
I.~Hasuo and Y.~Kawabe.
\newblock Probabilistic anonymity via coalgebraic simulations.
\newblock In R.~{De Nicola}, editor, {\em European Symposium on Programming
  (ESOP 2007)}, volume 4421 of {\em Lect. Notes Comp. Sci.}, pages 379--394.
  Springer, 2007.

\bibitem{HermidaJ98}
C.~Hermida and B.~Jacobs.
\newblock Structural induction and coinduction in a fibrational setting.
\newblock {\em Inf. \& Comp.}, 145:107--152, 1998.

\bibitem{Hoare85}
C.A.R. Hoare.
\newblock {\em {Communicating} {Sequential} {Processes}}.
\newblock {Prentice} {Hall}, 1985.

\bibitem{HughesJ04}
J.~Hughes and B.~Jacobs.
\newblock Simulations in coalgebra.
\newblock {\em Theor. Comp. Sci.}, 327(1-2):71--108, 2004.

\bibitem{Jacobs99a}
B.~Jacobs.
\newblock {\em Categorical Logic and Type Theory}.
\newblock North Holland, Amsterdam, 1999.

\bibitem{Jacobs04c}
B.~Jacobs.
\newblock Trace semantics for coalgebras.
\newblock In J.~Ad{\'a}mek and S.~Milius, editors, {\em Coalgebraic Methods in
  Computer Science}, volume 106 of {\em Elect. Notes in Theor. Comp. Sci.}
  Elsevier, Amsterdam, 2004.

\bibitem{JacobsR97}
B.~Jacobs and J.J.M.M. Rutten.
\newblock A tutorial on (co)algebras and (co)induction.
\newblock {\em {EATCS} Bulletin}, 62:222--259, 1997.

\bibitem{Jacobs05:coalgebraBook}
B.~Jacobs.
\newblock Introduction to coalgebra. {Towards} mathematics of states and
  observations.
\newblock {Draft} of a book, \\\verb+www.cs.ru.nl/B.Jacobs/PAPERS/index.html+,
  2005.

\bibitem{Jay95}
C.B. Jay.
\newblock A semantics for shape.
\newblock {\em Science of Comput. Progr.}, 25:251--283, 1995.

\bibitem{Kelly82}
G.M. Kelly.
\newblock {\em Basic Concepts of Enriched Category Theory}.
\newblock Number~64 in LMS. Cambridge Univ. Press, 1982.

\bibitem{KickPS06}
M.~Kick, A.J. Power, and A.~Simpson.
\newblock Coalgebraic semantics for timed processes.
\newblock {\em Inf. \& Comp.}, 204(4):588--609, 2006.

\bibitem{Klin05}
B.~Klin.
\newblock From bialgebraic semantics to congruence formats.
\newblock In {\em Workshop on Structural Operational Semantics (SOS 2004)},
  volume 128 of {\em Elect. Notes in Theor. Comp. Sci.}, pages 3--37, 2005.

\bibitem{Klin07b}
B.~Klin.
\newblock Bialgebraic operational semantics and modal logic.
\newblock In {\em Logic in Computer Science}, pages 336--345. IEEE Computer
  Society, 2007.

\bibitem{Klin07}
B.~Klin.
\newblock Coalgebraic modal logic beyond $\mathbf{Sets}$.
\newblock In {\em {MFPS} {XXIII}}, volume 173, pages 177--201. Elsevier,
  Amsterdam, 2007.

\bibitem{Kock70a}
A.~Kock.
\newblock Monads on symmetric monoidal closed categories.
\newblock {\em Arch. Math.}, XXI:1--10, 1970.

\bibitem{KupkeKV04}
C.~Kupke, A.~Kurz, and Y.~Venema.
\newblock Stone coalgebras.
\newblock {\em Theor. Comp. Sci.}, 327(1-2):109--134, 2004.

\bibitem{KupkeKP04}
C.~Kupke, A.~Kurz, and D.~Pattinson.
\newblock Algebraic semantics for coalgebraic logics.
\newblock {\em Elect. Notes in Theor. Comp. Sci.}, 106:219--241, 2004.

\bibitem{kurz:Phdthesis}
A.~Kurz.
\newblock {\em Logics for Coalgebras and Applications to Computer Science}.
\newblock PhD thesis, Universit\"{a}t M\"{u}nchen, April 2000.

\bibitem{Kurz06sigact}
A.~Kurz.
\newblock Coalgebras and their logics.
\newblock {\em SIGACT News}, 37(2):57--77, 2006.

\bibitem{Lawvere73}
F.W. Lawvere.
\newblock Metric spaces, generalized logic, and closed categories.
\newblock {\em Seminario Matematico e Fisico. Rendiconti di Milano},
  43:135--166, 1973.
\newblock Reprinted in \textit{Theory and Applications of Categories}, 1:1--37,
  2002.

\bibitem{LenisaPW00}
M.~Lenisa, A.J. Power, and H.~Watanabe.
\newblock Distributivity for endofunctors, pointed and co-pointed endofunctors,
  monads and comonads.
\newblock In H.~Reichel, editor, {\em Coalgebraic Methods in Computer Science},
  volume~33 of {\em Elect. Notes in Theor. Comp. Sci.} Elsevier, Amsterdam,
  2000.

\bibitem{LenisaPW04}
M.~Lenisa, J.~Power, and H.~Watanabe.
\newblock Category theory for operational semantics.
\newblock {\em Theor. Comp. Sci.}, 327(1--2):135--154, 2004.

\bibitem{LynchV95}
N.~Lynch and F.~Vaandrager.
\newblock Forward and backward simulations. {I}.~{U}ntimed systems.
\newblock {\em Inf. \& Comp.}, 121(2):214--233, 1995.

\bibitem{MacLane71}
S.~{Mac~Lane}.
\newblock {\em Categories for the Working Mathematician}.
\newblock Springer, Berlin, 2nd edition, 1998.

\bibitem{Murly94}
P.S. Mulry.
\newblock Lifting theorems for {Kleisli} categories.
\newblock In {\em Mathematical Foundations of Programming Semantics (MFPS IX)},
  pages 304--319, London, UK, 1994. Springer-Verlag.

\bibitem{Pardo01}
A.~Pardo.
\newblock Fusion of recursive programs with computational effects.
\newblock {\em Theor. Comp. Sci.}, 260(1--2):165--207, 2001.

\bibitem{PavlovicMW06}
D.~Pavlovi{\'c}, M.~Mislove, and J.B. Worrell.
\newblock Testing semantics: connecting processes and process logics.
\newblock In M.~Johnson and V.~Vene, editors, {\em Algebraic Methodology and
  Software Technology (AMAST 2006)}, volume 4019 of {\em Lect. Notes Comp.
  Sci.} Springer, 2006.

\bibitem{PowerT99}
J.~Power and D.~Turi.
\newblock A coalgebraic foundation for linear time semantics.
\newblock In {\em Category Theory and Computer Science}, volume~29 of {\em
  Elect. Notes in Theor. Comp. Sci.} Elsevier, Amsterdam, 1999.

\bibitem{Rutten00a}
J.J.M.M. Rutten.
\newblock Universal coalgebra: a theory of systems.
\newblock {\em Theor. Comp. Sci.}, 249:3--80, 2000.

\bibitem{Seg95:thesis}
R.~Segala.
\newblock {\em Modeling and verification of randomized distributed real-time
  systems}.
\newblock PhD thesis, MIT, 1995.

\bibitem{Segala95}
R.~Segala.
\newblock A compositional trace-based semantics for probabilistic automata.
\newblock In {\em International Conference on Concurrency Theory (CONCUR '95)},
  pages 234--248. Springer-Verlag, 1995.

\bibitem{Simpson92}
A.K. Simpson.
\newblock Recursive types in {Kleisli} categories.
\newblock Unpublished paper, available at
  \\\verb+http://homepages.inf.ed.ac.uk/als/Research/+, 1992.

\bibitem{SmythP82}
M.B. Smyth and G.D. Plotkin.
\newblock The category theoretic solution of recursive domain equations.
\newblock {\em SIAM Journ. Comput.}, 11:761--783, 1982.

\bibitem{Sok05}
A.~Sokolova.
\newblock {\em Coalgebraic Analysis of Probabilistic Systems}.
\newblock PhD thesis, Techn.\ Univ.\ Eindhoven, 2005.

\bibitem{StoelingaV03}
M.~Stoelinga and F.W. Vaandrager.
\newblock A testing scenario for probabilistic automata.
\newblock In J.C.M. Baeten, J.K. Lenstra, J.~Parrow, and G.J. Woeginger,
  editors, {\em ICALP}, volume 2719 of {\em Lect. Notes Comp. Sci.}, pages
  464--477. Springer, 2003.

\bibitem{TixKP05}
R.~Tix, K.~Keimel, and G.D. Plotkin.
\newblock Semantic domains for combining probability and non-determinism.
\newblock {\em Elect. Notes in Theor. Comp. Sci.}, 129:1--104, 2005.

\bibitem{TuriP97}
D.~Turi and G.~Plotkin.
\newblock Towards a mathematical operational semantics.
\newblock In {\em Logic in Computer Science}, pages 280--291. IEEE, Computer
  Science Press, 1997.

\bibitem{TuriR98}
D.~Turi and J.J.M.M. Rutten.
\newblock On the foundations of final semantics: non-standard sets, metric
  spaces and partial orders.
\newblock {\em Math. Struct. in Comp. Sci.}, 8(5):481--540, 1998.

\bibitem{vanGlabbeek01}
R.J. van Glabbeek.
\newblock The linear time--branching time spectrum {I}; the semantics of
  concrete, sequential processes.
\newblock In J.A. Bergstra, A.~Ponse, and S.A. Smolka, editors, {\em Handbook
  of Process Algebra}, chapter~1, pages 3--99. Elsevier, 2001.
\newblock Available at {\tt http://boole.stanford.edu/\-pub/\-spectrum1.ps.gz}.

\bibitem{GSS95:ic}
R.J. van Glabbeek, S.A. Smolka, and B.~Steffen.
\newblock Reactive, generative, and stratified models of probabilistic
  processes.
\newblock {\em Inf. \& Comp.}, 121:59--80, 1995.

\bibitem{VaraccaW05}
D.~Varacca and G.~Winskel.
\newblock Distributing probabililty over nondeterminism.
\newblock {\em Math. Struct. in Comp. Sci.}, 16(1):87--113, 2006.

\bibitem{Vardi85}
M.Y. Vardi.
\newblock Automatic verification of probabilistic concurrent finite-state
  programs.
\newblock In {\em FOCS '85}, pages 327--338, 1985.

\end{thebibliography}
%% in general the use of bibtex is encouraged

%\begin{thebibliography}{Kos97}

%\bibitem[Kos97]{koslowski:mib}
%J{\"u}rgen Koslowski.
%\newblock Monads and interpolads in bicategories.
%\newblock {\em Theory Appl. Categ.}, 3(8):182--212, 1997.

%\end{thebibliography}

\appendix
\section{Preliminaries}
\subsection{Initial/final sequences}\label{appendix:preliminariesInitialFinalSequence}
Here we recall the standard construction \cite{AdamekK79} of the initial algebra
(or the final coalgebra) via the initial (or final) sequence.
Notice that
the base category need not be $\Sets$.

Let $\C$ be a category with initial object $0$, and $F:
\C\to\C$ be an endofunctor.
The \emph{initial sequence}\footnote{%
 In this paper we consider only initial/final sequences of length $\omega$.}
of $F$ is a diagram
\begin{displaymath}
 \xymatrix@1@C+1.5em{
  0
        \ar^{\fromInit}[r]
 &
  F0
        \ar^-{F\fromInit}[r]
 &
  {\quad \cdots\quad}
        \ar^-{F^{n-1}\fromInit}[r]
 &
  {F^{n}0}
        \ar^{F^{n}\fromInit}[r]
 &
  {\quad\cdots}
 }
\end{displaymath}
where $\fromInit: 0\to X$ is the unique arrow.

Now assume that:
\begin{enumerate}[$\bullet$]
 \item the initial sequence has an $\omega$-colimit\footnote{%
 An $\omega$-colimit is a colimit of a diagram whose shape is the ordinal $\omega$.}
$(F^{n}0\stackrel{\alpha_{n}}{\longrightarrow} A)_{n<\omega}$;
 \item the functor $F$ preserves that $\omega$-colimit.
\end{enumerate} 
Then we have two cocones $(\alpha_{n})_{n<\omega}$
and $(F\alpha_{n-1})_{n<\omega}$ over the initial sequence.
Moreover, the latter is again a colimit: hence we have mediating
isomorphisms between these cones.
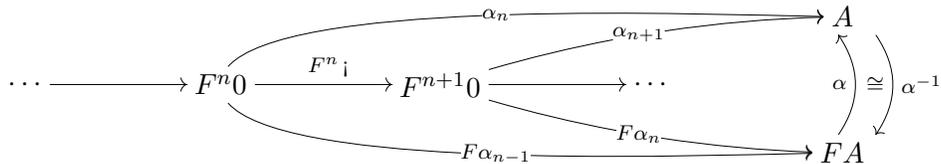
\begin{displaymath}
  \vcenter{\xymatrix@C+2.5em@R=1em{
  &&&&
   A
            \ar@<1.5ex>@/^1pc/[dd]^-{\alpha^{-1}}_-{\cong}
  \\
   {\cdots}
            \ar[r]
  &
   F^{n}0
            \ar@(u,l)[rrru]|{\alpha_{n}}
            \ar@(d,l)[rrrd]|{F\alpha_{n-1}}
            \ar[r]^{F^{n}{\fromInit}}
  &
   F^{n+1}0
            \ar@/^/[rru]|{\alpha_{n+1}}
            \ar@/_/[rrd]|{F\alpha_{n}}
            \ar[r]
  &
   {\cdots}
  &
  \\
  &&&&
   FA
            \ar@<1.5ex>@/_1pc/[uu]^-{\alpha}
}}
\end{displaymath}

\begin{prop}\label{LemInitialAlgebraViaInitialSequence}
 The $F$-algebra $\alpha: FA\iso A$ is initial.
\end{prop}
\proof
 For future reference we prove the dual result:
 see Proposition \ref{LemFinalCoalgebraViaFinalSequence}. 
\qed

\auxproof{ Any $F$-algebra $a: FX\to X$ induces a cocone $(\beta_{n}: F^{n}0\to X)_{n<\omega}$ 
 over the initial sequence in the following way.
 \begin{displaymath}
  \beta_{0} = {\fromInit} : 0\to X\enspace, \qquad
  \beta_{n+1} = a\co F\beta_{n}\enspace.
 \end{displaymath}
 Now we can prove the following: for an arrow $f: A\to X$, $f$ is a
 morphism of algebras from $\alpha$ to $a$ if and only if $f$ is a
 mediating arrow from the colimit $(\alpha_{n})_{n<\omega}$ to the
 cocone $(\beta_{n})_{n<\omega}$. 
 Hence such a morphism of algebras uniquely exists. 
}

The dual of this construction yields a final $F$-coalgebra.
Assume that the base category $\C$ has a terminal object $1$.
The \emph{final sequence} of $F$ is
\begin{displaymath}
 \xymatrix@1@C+1.5em{
  1 
        \ar_{{\toTerm}}[r];[]
 &
  F1
        \ar_-{F{\toTerm}}[r];[]
 &
  {\quad\cdots\quad}
        \ar_-{F^{n-1}{\toTerm}}[r];[]
 &
  {F^{n}1}
        \ar_{F^{n}{\toTerm}}[r];[]
 &
  {\quad\cdots,}
 }
\end{displaymath}
where ${\toTerm}: X\to 1$ is the unique arrow.
Assume that it has an $\omega^{\op}$-limit $(Z\stackrel{\zeta_{n}}{\longrightarrow} F^{n}1)_{n<\omega}$, and
also that $F$ preserves that $\omega^{\op}$-limit.
We have the following situation.
\begin{displaymath}
  \vcenter{\xymatrix@C+2.5em@R=1em{
  &&&&
   Z
            \ar@<1.5ex>@/_1pc/[dd];[]^-{\zeta^{-1}}_-{\cong}
  \\
   {\cdots}
            \ar[r];[]
  &
   F^{n}1
            \ar@(l,u)[rrru];[]|{\zeta_{n}}
            \ar@(l,d)[rrrd];[]|{F\zeta_{n-1}}
            \ar[r];[]_{F^{n}{\toTerm}}
  &
   F^{n+1}1
            \ar@/_/[rru];[]|{\zeta_{n+1}}
            \ar@/^/[rrd];[]|{F\zeta_{n}}
            \ar[r];[]
  &
   {\cdots}
  &
  \\
  &&&&
   FZ
            \ar@<1.5ex>@/^1pc/[uu];[]^-{\zeta}
}}
\end{displaymath}

\begin{prop}\label{LemFinalCoalgebraViaFinalSequence}
 The coalgebra $\zeta: Z\iso FZ$ is final.
\end{prop}
\proof
 Any $F$-coalgebra $c: X\to FX$ induces a cone $(X\stackrel{\beta_{n}}{\longrightarrow} F^{n}1)_{n<\omega}$ 
 over the final sequence in the following way.
 \begin{displaymath}
  \beta_{0} =\; {\toTerm} : X\longrightarrow 1\enspace, \qquad
  \beta_{n+1} = F\beta_{n}\co c\enspace.
 \end{displaymath}
 Now we can prove the following: for an arrow $f: X\to Z$, $f$ is a
 morphism of coalgebras from $c$ to $\zeta$ if and only if $f$ is a
 mediating arrow from the 
 cone $(\beta_{n})_{n<\omega}$ to the limit $(\zeta_{n})_{n<\omega}$.
 Hence such a morphism of coalgebras uniquely exists. 
\qed

It is easy to see that every shapely functor in $\Sets$ preserves $\omega$-colimits
and $\omega^{\op}$-limits. Hence we have the following.
\begin{lem}\label{ShInAlLem}
 A shapely functor $F$ has both an initial algebra and
 a final coalgebra in $\Sets$. \qed
\end{lem}

\subsection{limit-colimit coincidence}\label{appendix:preliminariesLimitColimitCoincidence}
We recall some relevant notions and results from \cite{SmythP82}.
The idea is that in a suitable order-enriched setting, (co)limits are
equivalently described as an order-theoretic notion of
\emph{$\mathbf{O}$-(co)limits}.
Due to the inherent coincidence between $\mathbf{O}$-limits and
$\mathbf{O}$-colimits, 
we also obtain the so-called \emph{limit-colimit coincidence}. 
\begin{displaymath}
 \vcenter{\xymatrix@R=.8em@C+5em{
  {\text{limit}}
                \ar@{=}[d]
 &
  {\text{colimit}} 
                \ar@{=}[d]
 \\
  {\text{$\mathbf{O}$-limit}}
                \ar@{-}[r]_{\text{obvious coincidence}}
 &
  {\text{$\mathbf{O}$-colimit}}
}}
\end{displaymath}

The notions of $\mathbf{O}$-(co)limits are stated in terms of 
\emph{embedding-projection pairs} which we can define in an
order-enriched category. In the sequel we assume the $\Cppo$-enriched
structure.
\begin{defi}[Embedding-projection pairs]\label{definition:embeddingProjectionPairs}
 Let $\C$ be a $\Cppo$-enriched category.
 A pair of arrows 
 %$(X\stackrel{e}{\to} Y, Y\stackrel{p}{\to} X)$ 
 \begin{displaymath}
  \vcenter{\xymatrix@1@C+3em{
   {X}
               \ar@/^/[r]^{e}
               \ar@/^/[r];[]^{p}         
  &
   {Y}
}}
 \end{displaymath}
in $\mathbb{C}$ is said to be an \emph{embedding-projection pair} if 
 we have $p\co e = \id$ and $e\co p \sqsubseteq \id$.
 Diagrammatically presented,
\[
 \xymatrix@C+2em@R=1.5em{
  {X}
                \ar@{ >->}[r]^{e}
                \ar[rd]_{\id}
                % <\omit> for omitting 1-cells
                % = for making 2-cell identity
                % <-3> is for adjusting position of 2-cell arrow
                \drtwocell<\omit>{=<-2>{}}              
 &
  {Y}
                \ar@{->>}[d]^{p}
                % *... is usual Xy-pic decoration
                % [@ru] is for rotation
                \ar[rd]^{\id}_(.4)*[@ru]{\sqsubseteq}
 &
 \\
 &
  {X}
                \ar@{ >->}[r]_{e}
 &
  {Y.}
}
\]
\end{defi}
 \noindent By $p\co e =\id$ we automatically have that $e$ is a mono and $p$ is an
 epi. Both split.

\begin{prop}\label{proposition:embeddingDeterminesCorrProjection}
Let $(e,p), (e',p'): X\rightleftarrows Y$
be two embedding-projection pairs with the same (co)domains. 
Then $e\sqsubseteq e'$ holds if and only if $p' \sqsubseteq p$.
As a consequence, one component of an embedding-projection pair
 determines the other. \qed
\end{prop}
\noindent This proposition justifies the notation $e^{P}$ for the projection
corresponding to a given embedding $e$, and $p^{E}$ for the embedding 
corresponding to a given projection $p$.
It is easy to check that
% $(e \co f)^{P} = f^{P}\co e^{P}$ and
%$(p\co q)^{E} = q^{E} \co p^{E}$.
\begin{displaymath}
 (e \co f)^{P} = f^{P}\co e^{P} \qquad\text{and}\qquad
(p\co q)^{E} = q^{E} \co p^{E}\enspace.
\end{displaymath}

\begin{defi}[$\mathbf{O}$-(co)limits]
 Let 
 $X_{0}
  \stackrel{f_{0}}{\to}
  X_{1}
  \stackrel{f_{1}}{\to}
  \cdots$ 
 be an
 $\omega$-chain in a $\Cppo$-enriched $\mathbb{C}$. 
 A cocone $(X_{n}\stackrel{\sigma_{n}}{\to} C)_{n<\omega}$ over this chain is said to be an 
 \emph{$\mathbf{O}$-colimit} if:
 \begin{enumerate}[$\bullet$]
  \item each $\sigma_{n}$ is an embedding;
  \item the sequence of arrows
	%$(\sigma_{n}\co \sigma_{n}^{P}: C\to	C)_{n<\omega}$
	$(\xymatrix@1{
           {C}
                \ar@{->>}[r]^{\sigma_{n}^{P}}
          &
           {X_{n}}
                \ar@{ >->}[r]^{\sigma_{n}}
          &
           {C}}
         )_{n<\omega}$
    is increasing. Moreover its join taken in the cpo $\mathbb{C}(C,C)$
    is $\id_{C}$.
 \end{enumerate} 
 \[
  \xymatrix@C+3em@R+0em{
  &
   {C}
           \ar@/^.4pc/@{ >->}[dl];[]^{\sigma_{0}}
           \ar@/^.4pc/@{->>}[dl]^{\sigma_{0}^{P}}
           \ar@/^.4pc/@{ >->}[d];[]^{\sigma_{1}}
           \ar@/^.4pc/@{->>}[d]^{\sigma_{1}^{P}}
           \ar@{}[dr]|{\cdots}
  &
  \\
   {X_{0}}  
           \ar[r]_{f_{0}}
  &
   {X_{1}}
           \ar[r]_{f_{1}}
  &
   {\cdots}
 }
 \]
 
 Dually, a cone $(C\stackrel{\gamma_{n}}{\to} Y_{n})_{n<\omega}$ over an
 $\omega^{\op}$-chain 
 $Y_{0}
   \stackrel{g_{0}}{\leftarrow}
  Y_{1}
   \stackrel{g_{1}}{\leftarrow}
  \cdots$
 is an \emph{$\mathbf{O}$-limit} if:
 each $\gamma_{n}$ is a projection, and
 the sequence $(\gamma_{n}^{E}\co\gamma_{n}: C\to C)_{n<\omega}$
       is increasing and its join is $\id_{C}$.
%\begin{enumerate}
% \item each $\gamma_{n}$ is a projection;
% \item the sequence $(\gamma_{n}^{E}\co\gamma_{n}: C\to C)_{n<\omega}$
%       is increasing and its join is $\id_{C}$.
%\end{enumerate}
\end{defi}

The following proposition establishes the equivalence between
(co)limits and $\mathbf{O}$-(co)li\-mits.
For its full proof the reader is referred to \cite{SmythP82}.
\begin{prop}[Propositions A, B, C, D in \cite{SmythP82}]
\label{colimit_implies_O-colimit}
 Let  $X_{0}
  \stackrel{e_{0}}{\to}
  X_{1}
  \stackrel{e_{1}}{\to}
  \cdots$ 
 be an $\omega$-chain where each $e_{n}$ is an embedding.
 \begin{enumerate}
  \item Let $( X_{n}\stackrel{\sigma_{n}}{\to} C)_{n<\omega}$ be the colimit over the chain.
	Then each $\sigma_{n}$ is also an embedding. 
	Moreover, $(\sigma_{n})_{n<\omega}$ is an $\mathbf{O}$-colimit.
  \item Conversely, an $\mathbf{O}$-colimit $(X_{n}\stackrel{\sigma_{n}}{\to} C)_{n<\omega}$ over the
	chain is a colimit.
 \end{enumerate}
 \noindent  Dually, let  $X_{0}
  \stackrel{p_{0}}{\leftarrow}
  X_{1}
  \stackrel{p_{1}}{\leftarrow}
  \cdots$ 
 be an $\omega^{\op}$-chain where each $p_{n}$ is a projection.
 \begin{enumerate}
 \setcounter{enumi}{2}
  \item  Let $(D\stackrel{\tau_{n}}{\to} X_{n})_{n<\omega}$ be a limit over the
	 chain.
	 Then each $\tau_{n}$ is also a projection.
	 Moreover $(\tau_{n})_{n<\omega}$ is an $\mathbf{O}$-limit.
  \item\label{OLimitIsLimit} 
	 Conversely, an $\mathbf{O}$-limit $(D\stackrel{\tau_{n}}{\to} X_{n})_{n<\omega}$ over the chain
	 is a limit.
 \end{enumerate}
\end{prop}
\proof
For later reference we present the proof 
of (\ref{OLimitIsLimit}). Let
 $(B\stackrel{\beta_{n}}{\to} X_{n})_{n<\omega}$ be an arbitrary cone over
 the chain $X_{0}
  \stackrel{p_{0}}{\leftarrow}
  X_{1}
  \stackrel{p_{1}}{\leftarrow}
  \cdots$. 
 First we prove the uniqueness of a mediating map $f: B\to D$.
 \begin{align*}
  f = \id_{D}\co f
    &= \bigl(\textstyle\bigsqcup_{n<\omega}(\tau_{n}^{E}\co \tau_{n})\bigr) \co f
    &&\text{$(\tau_{n})_{n<\omega}$ is an $\mathbf{O}$-limit}
  \\
    &= \textstyle\bigsqcup_{n<\omega}(\tau_{n}^{E}\co \tau_{n}\co f) 
    &&\text{composition is continuous}
  \\
    &= \textstyle\bigsqcup_{n<\omega}(\tau_{n}^{E}\co \beta_{n})
    &&\text{$f$ is mediating}\enspace.
 \end{align*}
 We conclude the proof by showing that the sequence $(\tau_{n}^{E}\co\beta_{n})_{n<\omega}$
 is increasing, hence such $f$ indeed exists.
 \begin{align*}
  \tau_{n}^{E}\co\beta_{n}
  \;=\;
  \tau_{n}^{E}\co p_{n}\co\beta_{n+1}
  \;= \;
  \tau_{n+1}^{E}\co p_{n}^{E} \co p_{n}\co \beta_{n+1}
  \;\sqsubseteq\;
  \tau_{n+1}^{E}\co \beta_{n+1}
  \enspace
%  \tag*{$\qEd$}
 \end{align*} 
 The last inequality holds because $p_{n}^{E} \co p_{n}\sqsubseteq \id$
 from the definition of embedding-projection pairs. \qed

\begin{thm}[Limit-colimit coincidence]
\label{limit_colimit_coincidence}
 Let $X_{0}
  \stackrel{e_{0}}{\to}
  X_{1}
  \stackrel{e_{1}}{\to}
  \cdots$ 
 be an $\omega$-chain where each $e_{n}$ is an embedding, and
 $(X_{n}\stackrel{\sigma_{n}}{\to} C)_{n<\omega}$ be the colimit over the chain.
 Then each $\sigma_{n}$ is an embedding, and the cone
 $( C\stackrel{\sigma_{n}^{P}}{\to} X_{n})_{n<\omega}$ is a limit over the
 $\omega^{\op}$-chain 
 $X_{0}
  \stackrel{e_{0}^{P}}{\leftarrow}
  X_{1}
  \stackrel{e_{1}^{P}}{\leftarrow}
 \cdots$.
 \begin{displaymath}
  \vcenter{\xymatrix{
   &&
    {C}
%        \shifted{3em}{0em}{\text{colimit}}
   \\
    {X_{0}}
            \ar@/^/@{ >->}[rru]^{\sigma_{0}}
            \ar@{ >->}[r]_{e_{0}}
   &
    {X_{1}}
            \ar@{ >->}[ru]_{\sigma_{1}}
            \ar@{ >->}[r]_{e_{1}}
   &
    {\cdots}
}}\text{ : colimit}
  \quad\Longrightarrow\quad
  \vcenter{\xymatrix{
   &&
    {C}
%        \shifted{3em}{0em}{\text{limit}}
   \\
    {X_{0}}
            \ar@/_/@{->>}[rru];[]_{\sigma_{0}^{P}}
            \ar@{->>}[r];[]^{e_{0}^{P}}
   &
    {X_{1}}
            \ar@{->>}[ru];[]^{\sigma_{1}^{P}}
            \ar@{->>}[r];[]^{e_{1}^{P}}
   &
    {\cdots}
}}
  \text{ : limit}
 \end{displaymath}  

 Dually, the limit of an $\omega^{\op}$-chain of projections consists of
 projections. 
 By taking the corresponding embeddings we obtain a colimit of an
 $\omega$-chain
 of embeddings.
\end{thm}
\proof
 We prove the first statement. 
 By Proposition \ref{colimit_implies_O-colimit} each $\sigma_{n}$ is an
 embedding, and moreover $(\sigma_{n})_{n<\omega}$ is an $\mathbf{O}$-colimit.
 Now obviously $(\sigma_{n}^{P})_{n<\omega}$ is a cone over 
 $X_{0}
  \stackrel{e_{0}^{P}}{\leftarrow}
  X_{1}
  \stackrel{e_{1}^{P}}{\leftarrow}
 \cdots$.
 Here we use the inherent coincidence of $\mathbf{O}$-(co)limits: namely,
 the condition that $(\sigma_{n})_{n<\omega}$ is an $\mathbf{O}$-colimit
 is exactly the same as that $(\sigma_{n}^{P})_{n<\omega}$ is an
 $\mathbf{O}$-limit.
 We use Proposition \ref{colimit_implies_O-colimit} to conclude the
 proof. 
\qed

%\subsection{A lifted functor does not yield a ``denotation'' natural
%  transformation}
%\label{appendix:liftingDoesNotYieldNaturalTransformation}
%The essence of the question here is presented as follows.
%Assume we have the following situation of functors.
%\begin{displaymath}
%  \xymatrix@1@C+3em{
%  {\C}
%              \ar@(lu,ld)[]_{F}
%              \ar[r]^{H}
% &
%  {\D}
%              \ar@(rd,ru)[]_{G}
% }
%\end{displaymath}
%What is the relation between
%\begin{enumerate}
% \item natural transformations $\lambda: GH\Rightarrow HF$, and
% \item liftings $\hat{H}: \Alg{F}\to\Alg{G}$ of $H$?
%\end{enumerate}
%It is straightforward to see that $\lambda$ yields $\hat{H}$: it
% carries an algebra $a: FX\to X$ to
%$GHX\stackrel{\lambda_{X}}{\to} HFX \stackrel{Ha}{\to} HX$.

%However, there is no correspondence the other way round.
%Let $\C$ and $\D$ be the preorder $\omega$, $Fn=1+n$, $Gn=2+n$
%and $H=\id$. Then the categories $\Alg{F}$ and $\Alg{G}$ are
%empty, so there is a trivial lifting $\hat{H}=0$. But there is
%no natural transformation $GH\Rightarrow HF$.

\end{document}